%% file: main.tex
\documentclass[11pt]{article}
\usepackage{amsmath,amsfonts,amssymb,amsthm,bm}
\usepackage[colorlinks]{hyperref}
\hypersetup{linkcolor=cyan,filecolor=cyan,citecolor=cyan,urlcolor=cyan}
\usepackage{xspace}
\usepackage{thm-restate,color,xcolor}
\usepackage{fullpage}
\usepackage[boxed]{algorithm}
\usepackage{mathtools}
\usepackage{framed}
\usepackage[framemethod=tikz]{mdframed}
\usepackage{cleveref,aliascnt}
\usepackage{tikz}
\usepackage{wrapfig}
\usepackage{framed}
\usepackage{paralist}

\newcommand{\cA}{\mathcal{A}}

\newcommand{\cT}{\mathcal{T}}
\newcommand{\cH}{{\mathcal H}}

\DeclareMathOperator{\polylog}{polylog}
\DeclareMathOperator{\poly}{poly}

\theoremstyle{plain}
\newtheorem{theorem}{Theorem}[section]
\newtheorem{claim}{Claim}

\newtheorem{lemma}[theorem]{Lemma}
\newtheorem{observation}[theorem]{Observation}

\newtheorem{corollary}[theorem]{Corollary}

\newtheorem{question}[theorem]{Question}
\newtheorem{definition}{Definition}
\newtheorem{invariant}{Invariant}

\newcommand{\congest}{\textsf{CONGEST}\xspace}
\newcommand{\congc}{\textsf{CONGESTED CLIQUE}\xspace}

\newcommand{\Ebad}{E_{\mathrm{bad}}}
\newcommand{\Egood}{E_{\mathrm{good}}}
\newcommand{\Pgood}{P^{\mathrm{good}}}
\newcommand{\Pbad}{P^{\mathrm{bad}}}
\newcommand{\ID}{\mathrm{id}}

\newcommand{\load}{\mathrm{load}}
\newcommand{\cj}{\ensuremath{c}\xspace}
\newcommand{\congestion}{\mathrm{cong}}
\newcommand{\dilation}{\mathrm{dilation}}

\newcommand{\alphagood}{\alpha^{\mathrm{good}}}
\newcommand{\alphabad}{\alpha^{\mathrm{bad}}}
\newcommand{\nat}{\ensuremath{\mathbb{N}}}

\newcommand{\Message}[3]{\pi_{#2,#1}(#3_{#1})}
\newcommand{\MessageStatic}[3]{\pi_{#2,#1}(#3)}

\newcommand{\MessageEvent}[5]{Y_{#2,#3}(#1,#4)}
\newcommand{\MessageEventStatic}[5]{Y_{#2,#3}(#1,#4)}
\newcommand\code[2]{$\left[ #1 \right]_{#2}$\xspace}

\DeclareMathOperator{\Hamm}{Hamm}

\newcommand{\abs}[1]{\left|#1\right|}
\newcommand{\sett}[2]{\left\{ #1 \left| \; \vphantom{#1 #2} \right. #2  \right\}} 

\DeclareMathOperator{\View}{View}

\DeclareMathOperator{\supp}{supp}

\DeclareMathOperator{\DeleteLast}{\mathrm{DeleteLast}}
\newcommand{\good}{\mathrm{good}}

\newcommand{\GoodState}{\mathrm{GoodState}}

\newcommand{\prefix}{\mathrm{prefix}}
\newcommand{\parent}{\mathrm{parent}}
\newcommand{\argmax}{\mathrm{argmax}}
\newcommand{\Col}{\mathrm{Col}}

\newcommand{\stt}{\medspace | \medspace}

\newcommand{\ImprovedMobileByznatineSim}{\mathsf{ImprovedMobileByznatineSim}}
\newcommand{\StaticSecureUnicast}{\mathsf{StaticSecureUnicast}}
\newcommand{\MobileSecureUnicast}{\mathsf{MobileSecureUnicast}}
\newcommand{\MobileSecureMulticast}{\mathsf{MobileSecureMulticast}}
\newcommand{\MobileBroadcast}{\mathsf{MobileBroadcast}}
\newcommand{\ModifiedECCSafeBroadcast}{\mathsf{ModifiedECCSafeBroadcast}}

\newcommand{\RSScheduler}{\mathsf{RSScheduler}}
\newcommand{\ECCSafeBroadcast}{\mathsf{ECCSafeBroadcast}}

\newcommand{\MAJ}{\mathsf{MAJ}}
\newcommand{\DM}{\mathsf{DM}}
\newcommand{\kDiam}{\mathsf{D}_{\mathsf{TP}}}
\newcommand{\RSthresh}{c_\mathsf{RS}}
\newcommand{\RStime}{t_\mathsf{RS}}
\newcommand{\RS}{\mathsf{RS}}

\newcommand{\Merge}{\mathrm{Merge}}
\newcommand{\Query}{\mathrm{Query}}

\AtBeginDocument{%
	}

\begin{document}
	
	
	

\date{}
\author{
Orr Fischer \\
	\small Weizmann Institute of Science \\
	\small orr.fischer@weizmann.ac.il \\
	\and
	Merav Parter \\ 
	\small Weizmann Institute of Science \\
	\small merav.parter@weizmann.ac.il				
}
	
\title{Distributed CONGEST Algorithms against Mobile Adversaries}
	
	
	\maketitle
	\thispagestyle{empty}
	\input{abstract.tex}
	\pagebreak
	\thispagestyle{empty}
	\tableofcontents
	\pagebreak
	\setcounter{page}{1}
	\input{intro.tex}

	\input{contribution.tex}

	\input{techniques.tex}
	\input{prelim.tex}
	\input{roadmap.tex}
	
	\input{eavesdropper-mobile.tex} 
	\input{byzantine-improved-new-merged.tex}

	\input{byzantine-applications.tex}
	
	\input{byzantine-budgeted.tex}

	\input{correction-via-cycle-covers.tex}

	\section*{Acknowledgments}
		We are very grateful to Yanic Maus for his extremely valuable comments and suggestions, and for the time and effort he invested in
		carefully reading our paper. This project is funded by the European Research Council (ERC) under the European Union’s Horizon 2020 research and innovation programme (grant agreement No. 949083), and by the Israeli Science Foundation (ISF), grant No. 2084/18.

	\bibliographystyle{plain}
	\bibliography{crypto}

\appendix

\input{eavesdropper-mobile-improved.tex}

\input{byzantine-improved-for-budgeted.tex}

\input{treecomputation.tex}
	
\end{document}

%% file: abstract.tex
\begin{abstract}
In their seminal PODC 1991 paper, Ostrovsky and Yung introduced the study of distributed computation in the presence of mobile adversaries which can dynamically appear throughout the network, analogous to a spread of a virus. Over the years, this setting has been studied mostly under the assumption that the communication graph is fully-connected. Resilient \congest\ algorithms for \emph{general} graphs, on the other hand, are currently known only for the classical \emph{static} setting, i.e., where the set of corrupted edges (or nodes) is fixed throughout the entire computation.

We fill this missing gap by providing round-efficient simulations that translate given \congest\ algorithms into equivalent algorithms that are resilient against $f$-mobile edge adversaries, i.e.,  where the adversary controls a (possibly distinct) subset of $f$ edges $F_i$ in each round $i$. Our main results are:
%

\begin{itemize}
\item{\textbf{Perfect-Security with Mobile Eavesdroppers.}} A translation of any $r$-round $f$-\emph{static}-secure algorithm into an equivalent $\Theta(f)$-\emph{mobile}-secure algorithm with $\Theta(r)$ rounds. We also show that the $f$-static-secure algorithms of [Hitron, Parter and Yogev, DISC 2022 \& ITCS 2023] can be modified into $f$-\emph{mobile}-secure algorithms with the \emph{same} number of rounds.

\item{\textbf{Resilience with Mobile Byzantine Adversaries.}} An $f$-mobile-byzantine simulation which is based on a decomposition of the graph into low-diameter edge-disjoint spanning trees. This provides us with near-optimal \congest\ compilers for expander graphs. It also leads to near-optimal compilers in the congested-clique model against $\Theta(n)$-mobile adversaries. For general $(2f+1)$ edge-connected graphs with $f$-mobile adversary, we almost match the bounds known for the $f$-static setting, when provided a trusted pre-processing phase.
%
\end{itemize}

Our results are based on a collection of tools borrowed from the area of interactive coding [Gelles, Found. Trends Theor. Comput. Sci. 2017], linear sketches and low-congestion graph decomposition. The introduced toolkit might have further applications for resilient computation. 
\end{abstract}

%% file: intro.tex
\section{Introduction}
Following our increased dependence on distributed infrastructures, protecting the correctness and the privacy of users' information in the presence of faults has become an imperative mission in distributed network design. 
The inherent vulnerability of these systems seems inevitable as in distributed algorithms the output of one node is used in the computation of another. Modern network instantiations, e.g., the Blockchain, call for new kinds of distributed algorithms. 

The study of \emph{resilient} and \emph{secure} distributed computation has evolved along two lines of research. The line on resilient byzantine computation has been initiated by the work of Pease et al. \cite{pease1980reaching} and Lamport et al. \cite{LamportSP82, pease1980reaching} on the \emph{byzantine agreement problem}. The second line which focuses on \emph{information-theoretic security} dates back to the work of Yao \cite{Yao86}, and has been extensively addressed by the Cryptographic community under the Multi-Party-Communication (MPC) model \cite{BenorGW88}. While earlier work assumed \emph{static} adversaries (in which the set of corruptions is fixed), the arguably more realistic \emph{mobile} (or dynamic) faulty setting has attracted a lot of attention as well, in both of these communities. In this mobile setting, faults might be introduced in a dynamic and adaptive manner, similarly to a spread of a computer virus. A key limitation of many of these existing algorithms, however, is their restriction to fully-connected communication graphs. 

A recent line of works \cite{ParterYPODC19,ParterY19s,ParterYSODA19a,HitronP21,HitronP21a,EavesBroad2022,EavesMST2023} mitigated this gap, by providing resilient and secure algorithms, for any graph topology, in the \congest\ model of distributed computing \cite{Peleg:2000}. These algorithms have been limited, so far, to \emph{static} adversaries that control a fixed number of edges (or nodes) in the graph\footnote{Many of these works can handle adaptive adversaries, a stronger variant of the static setting, which allows the adversary to place the total of $f$ corruptions in an adaptive manner.}. The primary objective of this paper is in providing a new algorithmic approach for handling \emph{mobile adversaries} while keeping the round overhead as close as possible to the static counterparts. We focus on the following fundamental question that has been addressed so-far mainly in the complete graph setting:
\vspace{-3pt}
\begin{question}\label{ques:cost}
\emph{What is the cost (in terms of the number of \congest\ rounds) for providing resilience against \emph{mobile} vs. \emph{static} adversaries in general distributed networks?}
\end{question}
\vspace{-3pt}
In terms of feasibility, earlier work, e.g., by Srinathan et al. \cite{SRR07} has demonstrated that the graph connectivity requirements are the same for both static and mobile adversaries. Our main contribution is in providing new algorithms for mobile adversaries that almost match the state-of-the-art results for their static counterparts. An additional benefit of our approach is that in some cases it leads to improved bounds (and new results) already for the static setting. 

\smallskip 
\noindent \textbf{Line 1: Resilient Computation, in Complete Graphs.}  In the classical (static) byzantine setting, an all-powerful adversary controls a fixed subset of edges (or nodes) by sending malicious messages through these edges. Time-efficient and communication-efficient algorithms have been devised for various of distributed tasks that can tolerate up to a \emph{constant} fraction of corrupted edges and nodes in complete network topologies. Examples include: broadcast and consensus \cite{Dolev82, DolevFFLS82, fischer1983consensus,BrachaT85,toueg1987fast,santoro1989time,BermanGP89,SantoroW90,BermanG93,FeldmanM97,GarayM98,fitzi2000partial,Koo04, PelcP05,katz2006expected,DolevH08, MaurerT12, imbs2015simple,CohenHMOS19,KhanNV19}, gossiping \cite{blough1993optimal,bagchi1994information,censor2017fast}, and agreement \cite{DolevFFLS82,pease1980reaching,bracha87,CoanW92,GarayM98}. 

Mobile byzantine (node) faults have been addressed by Garay \cite{Garay94} in the context of the byzantine agreement problem. Tight bounds for this problem, in terms of the allowed number of faults per round, have been provided by Bonnet et al. \cite{BonnetDNP16}. See Yung \cite{Yung15} for an overview on mobile adversaries.

\smallskip 
\noindent \textbf{Line 2: Secure Computation, in Complete Graphs.} The notion of information-theoretic security is among the most fundamental and long-studied concepts in the area of secure MPC. Starting with the earlier work of Yao \cite{Yao86} for $n=2$, Goldreich, Micali and Wigderson \cite{GoldreichMW87} for general $n$, to the well-known Ben-Or, Goldwasser and Widgderson (BGW) protocol \cite{BenorGW88} that provides information-theoretic security against semi-honest adversaries controlling almost half of the parties. 

Inspired by the mobility of computer viruses and swarms, Ostrovsky and Yung \cite{OstrovskyY91} initiated the study of mobile adversarial settings, where corruptions are introduced, and removed, in a dynamic manner, throughout the course of execution. The extensive line of \emph{mobile} secure algorithms has developed into the well-established topic of \emph{proactive-security} \cite{CanettiH94,SRR07,BaronDLO13,DolevDLOY16,EldefrawyOY18}.

 As in the static setting, most of these algorithms are usually designed for complete networks, and relatively little is known on the complexity of such computations in general graphs.

\smallskip
\noindent\textbf{Line 3: Resilient and Secure Computation for Any Graph, Static Adversaries.} Throughout, an algorithm is denoted as $f$-static-secure (resp., $f$-static-resilient) if it guarantees information-theoretic security (resp., correctness) in the presence of (static) adversaries controlling at most $f$ edges in the graph\footnote{The precise definition of the adversarial settings are elaborated later on.}.  It is well-known that handling $f$-static eavesdroppers requires an edge-connectivity of $f+1$. In contrast, $f$-static byzantine adversaries require an edge-connectivity of $2f+1$ \cite{Dolev82,DolevDWY90,Pelc92}.  In a sequence of works, Parter and Yogev \cite{ParterYPODC19,ParterY19s,ParterYSODA19a} introduced a graph-theoretic paradigm for round-efficient \congest\ algorithms that are $f$-static-secure and $f$-static-resilient, for sufficiently connected graphs. Their approach is based on providing low-congestion \emph{reliable} paths between every pair of neighboring nodes in the graph. This yields a general compilation of any fault-free \congest\ algorithm into an equivalent $f$-secure (or resilient) algorithm. The round overhead depends on the length of the $\Theta(f)$ edge-disjoint paths between neighbors, which might be bounded by $O(\min\{n,(D/f)^{\Theta(f)}\})$, where $D$ is the graph diameter \cite{HitronP21a}.

In a sequence of two very recent works, Hitron, Parter and Yogev \cite{EavesBroad2022,EavesMST2023} bypassed this $D^f$ barrier for the adversarial setting of eavesdroppers. By employing the secure unicast algorithm of Jain \cite{Jain04}, they provide $f$-static-secure broadcast algorithms \cite{EavesBroad2022} with round complexity\footnote{As usual, $\widetilde{O}()$ hides factors poly logarithmic in $n$.} of $\widetilde{O}(D+\sqrt{f n})$, for $D$-diameter $n$-node graphs with edge-connectivity of $\Theta(f)$. \cite{EavesMST2023} provided $f$-static-secure compilers for low-congestion \congest\ algorithms, along with near-optimal $f$-static-secure algorithms for Minimum-Spanning-Tree (MST). 

\smallskip
\noindent \textbf{New: \congest\ Algorithms with Mobile Adversaries.} We provide $f$-mobile-secure (resp., $f$-mobile-resilient) whose privacy and resilience guarantees hold in the presence of mobile adversary that controls distinct subsets of at most $f$ edges, in each round. Srinathan et al. \cite{SRR07} show that the connectivity requirements are the same for both static and mobile adversaries (either eavesdroppers or byzantine). Providing round-efficient algorithms for this dynamic and adaptive adversarial behaviors calls for a new algorithmic paradigm that borrows useful techniques from streaming algorithms, graph decomposition and \emph{interactive coding}. While mobile secure algorithms can be provided quite readily, our major struggles go into mobile resilient algorithms against mobile byzantine adversaries. This is based on a completely different approach than that taken in the prior (static) work of \cite{HitronP21a,HitronP21}. We note that a special attention in the literature has been devoted to the unicast problem (a.k.a., the \emph{Secure Message Transmission} problem) \cite{FranklinGY00,SRR07,PatraCRSR09}. In our general compilation scheme for a given $m$-edge graph, it is required to solve $m$ many unicast instances, i.e., for every $(u,v)\in E$. 

\smallskip
\noindent \textbf{Related Setting: Interactive Coding.} While there is no general machinery for providing mobile resilience in the \congest\ model, the closest setting to ours is that of interactive coding, in which the adversary is allowed to corrupt a bounded \emph{fraction} of the \emph{total} communication bits (i.e., bounded \emph{communication-error-rate}). Rajagopalan and Schulman \cite{RajagopalanS94} provided a network analog of Shannon's coding theorem against stochastic noise. Computationally efficient protocols for this setting were subsequently provided by Gelles, Moitra and Sahai \cite{GellesMS11}. Hoza and Schulman \cite{HozaS16} provided the first network protocols against adversarial noise, that also fit the bandwidth limitation of the \congest\ model\footnote{In fact, their algorithms send one bit of information on each of the graph edges, in a given round.}. 
Censor{-}Hillel, Gelles and Haeupler \cite{Censor-HillelGH18} presented the first \emph{fully} distributed interactive coding scheme in which the topology of the communication network is not assumed to be known in advance, as in prior works in this setting. See \cite{G17} for an excellent review. 

Our $f$-mobile setting is in some sense incomparable to that of interactive coding. Assume an $n$-node graph with $m=\Theta(n^2)$ edges. Then, in the case where the protocol sends $O(n)$ messages per round, our adversary is \emph{stronger} as the interactive coding adversary is limited to an error rate of $O(1/n)$, and therefore cannot corrupt even a single edge in each and every round. On the other hand, if the $r$-round protocol sends $\Omega(m)$ messages per round, the interactive coding setting allows for $\Omega(m/n)$ corruptions, while our $f$-mobile setting allows for a total of $f r$ corruptions. 

%% file: contribution.tex
\vspace{-5pt}\subsection{New Results}\vspace{-3pt}

We present a new algorithmic framework for distributed computation in the presence of mobile edge adversaries, in which the set of corrupted edges changes dynamically and adaptively throughout the execution. We investigate two main adversarial settings: (i) mobile eavesdroppers where the key objective is \emph{security} of information and (ii) mobile byzantine adversaries where we strive for maintaining the correctness of the computation.

\smallskip
\noindent \textbf{Security against Mobile Eavesdroppers.} We present a general simulation result that translates any given static-secure algorithm into a mobile-secure algorithm while keeping the same asymptotic bound on the number of controlled edges and round complexity. We show:

\begin{theorem}\label{thm:simulation}
Let $\cA$ be an $r$-round $f$-static-secure algorithm for $r\leq \poly(n)$. Then for any positive integer $t$, there exists an equivalent $r'$-round $f'$-mobile-secure algorithm $\cA'$ such that: 
$r'=2r+t \mbox{~and~} f'=\lfloor (f \cdot (t+1))/(r+t) \rfloor~.$ Moreover, an equivalent protocol exists for any $t \geq 2fr$, $r'=2r+t \mbox{~and~} f'=f$. Consequently, any $r$-round $f$-static-secure algorithm $\cA$ can be turned into $r'$-round $f'$-mobile-secure algorithm with: (i) $r'=O(r)$ and $f'=\Theta(f)$ and (ii) $r'=O(fr)$ and $f'=f$. 
\end{theorem}
To avoid the extra $f$ factor in the round overhead (when insisting on $f'=f$), we also provide a white-box modification of the existing $f$-static-secure algorithms of \cite{EavesBroad2022} and \cite{EavesMST2023}. A notable tool introduced in \cite{EavesMST2023} is a general \emph{congestion-sensitive} compiler whose performances are optimized for (fault-free) distributed algorithms with low-congestion. A distributed algorithm is said to have \emph{$\congestion$-congestion} for an integer $\congestion$ if the maximum number of messages that the algorithm sends over any given edge in the graph throughout its entire execution is bounded by $\congestion$. 
We show:

\begin{theorem}[Congestion-Sensitive Compiler with Perfect Mobile Security]\label{thm:mobile-compilers}
For every $(2f+3)(1+o(1))$ edge-connected $D$-diameter $n$-vertex graph $G$, any $r$-round $\congestion$-congestion algorithm $\cA$ for $G$ in the fault-free setting can be compiled into an equivalent $f$-mobile-secure algorithm $\cA'$ that runs in $\widetilde{O}(r +D + f \cdot \sqrt{\congestion \cdot n} + f \cdot \congestion)$ \congest\-rounds. The correctness of the simulation holds w.h.p.\footnote{As usual, w.h.p. refers to a success guarantee of $1-1/n^c$ for any desired constant $c\geq 1$.}
\end{theorem}
This matches the $f$-static, statistically-secure compilers of \cite{EavesMST2023}. Our compilers have the benefit of achieving perfect security. This is obtained by replacing the implicit balls-into-bins ingredient of \cite{EavesMST2023} with bounded-independence hash functions.  To prove Theorem \ref{thm:mobile-compilers}, we also provide matching bounds for the mobile variants of the secure broadcast and unicast problems, studied by \cite{EavesBroad2022}. 

\smallskip
\noindent \textbf{Resilience against Mobile Byzantine (Edge) Adversaries.} An $f$-mobile byzantine adversary can maliciously corrupt the messages exchanged over at most $f$ edges $F_i$ in each round $i$. We first provide a brute-force extension of the cycle-cover based solution of \cite{HitronP21a} to the mobile setting. 
%

\begin{theorem}[$f$-Mobile-Resilient Compilers for General Graphs]\label{cor:CC-mobile-simulation-final}
Given any $n$-node $D$-diameter graph $G$ with edge-connectivity $2f+1$, any $r$-round algorithm $\cA$ for $G$ can be compiled into equivalent $r'$-round algorithm $\cA'$ that is $f$-mobile resilient and $r'=D^{\Theta(f)}\cdot \log n$. This holds provided that either (i) all nodes know the graph topology (a.k.a., the supported-\congest\ model), or (ii) there is a fault-free preprocessing step of $D^{\Theta(f)}$ rounds. 
\end{theorem} 
This extends the $1$-mobile-resilient compilation of \cite{ParterY19s}. It also matches the state-of-the-art of \cite{HitronP21a} for the $f$-static setting. While \cite{ParterY19s} also requires a fault-free preprocessing, \cite{HitronP21a} does not.
%
%
%

To handle $f=\Omega(\log n)$ faults, our key technical contribution is in providing a new compilation scheme which is based on \emph{low-diameter tree packing}. For a graph with edge connectivity $k$, the \emph{tree-packing-diameter} of the graph is measured by the minimum diameter $\kDiam$ such that one can decompose the graph into $\Omega(k/\log n)$ near\footnote{In this context, near edge-disjoint means that each edge $e \in E$ appears in at most $\tilde{O}(1)$ many trees in the packing} edge-disjoint spanning trees (a.k.a tree-packing) of diameter at most $\kDiam$. We show that given a $\kDiam$-diameter tree-packing for $k=\Theta(f\log n)$, any (fault-free) algorithm can become $f$-mobile-resilient with a round overhead of $\widetilde{O}(\kDiam)$.
\begin{theorem}
\label{thm:tree-packing-comp}
Given a $\kDiam$-diameter tree-packing with $k=\Theta(f\log n)$ trees, any $r$-round algorithm $\cA$ can be compiled into an $r'$-round $f$-mobile-resilient algorithm $\cA'$ where $r'=\widetilde{O}(\kDiam)$.
\end{theorem}

\noindent \textbf{Useful Applications.}  Theorem \ref{thm:tree-packing-comp} leads to several applications of interest. Most notably a general $\Theta(n)$-mobile compiler in the classical \congc\ model \cite{LotkerPPP05} where the underlying communication graph is a clique.

\begin{theorem}[Mobile-Resilient Compilers in the Congested Clique]\label{thm:CC-compiler}
Any $r$-algorithm in the \congc\ model can be compiled against $\Theta(n)$-mobile adversaries using $\widetilde{O}(r)$ rounds. 
\end{theorem} 
This theorem requires no preprocessing step, as the clique configuration trivially defines a tree packing of diameter $2$. Our second application is for expander graphs, where we compute in the $f$-mobile setting, a (weaker variant) of tree packing, which provides the following:

\begin{theorem}[Mobile-Resilient Compilers for Expander Graphs]\label{thm:expander-compiler}
Assume $G$ is a $\phi$-expander with minimum degree $k=\widetilde{\Omega}(1/\phi^2)$. Then any $r$-round algorithm $\cA$ can be compiled into an $f$-mobile-resilient algorithm $\cA'$ for $f=\widetilde{O}(k\phi)$ that runs in $\widetilde{O}(r/\phi)$ \congest\-rounds. 
\end{theorem}

Finally, Theorem \ref{thm:tree-packing-comp} also provides compilers for general graphs, in which the round overhead depends (up to poly-log factors) on the (instance) optimal length of $k$ edge-disjoint paths between neighboring pairs, in the given graph. This is in contrast to prior work (e.g., \cite{HitronP21a}) which competes with the worst-case bound on this length. 

\smallskip
\noindent \textbf{Even Stronger Adversaries: Resilience with Bounded Round-Error-Rate.} Finally, we extend our $f$-mobile compilation scheme to the stronger setting in which the adversary is allowed to corrupt a total of $f r$ edges in an $r$-round algorithm. That is, corrupting at most $f$ edges, per round, on \emph{average}. By using the rewind-if-error technique from interactive coding \cite{S92}, we match the round overhead provided for the $f$-mobile setting. This also provides stronger formulations of Theorem \ref{thm:CC-compiler} and \ref{thm:expander-compiler}. For example, for the \congc\ model, one can compile an $r$-round algorithm in $\widetilde{O}(r)$ rounds, while tolerating a total of $\widetilde{\Theta}(r\cdot n)$ corruptions.

%% file: techniques.tex
\vspace{-5pt}\subsection{Technical Overview}

\input{techniques-mobile-eavesdropper.tex}

\input{techniques-mobile-byzantine.tex}

%% file: techniques-mobile-eavesdropper.tex
\subsubsection{Perfect-Security with Mobile Adversaries} 

Simulating a given $f$-static-secure $r$-round algorithm $\cA$ securely in the $f$-mobile setting is based on the following observation: Assume that all but a subset of $f$ edges, denoted as $F^*$, hold $r$ secret random messages, that are hidden from the adversary. That is, assume that for every $(u,v)\in E \setminus F^*$, $u$ and $v$ hold $R_1(u,v),\ldots, R_r(u,v)$ random messages, which the adversary does not know. Then, one can simulate $\cA$ in a round-by-round manner, where in round $i$, each $u$ sends $m_i(u,v) \oplus R_i(u,v)$ to each neighbor $v$, where $m_{i}(u,v)$ is the message that $u$ sends to $v$ in round $i$ of Alg. $\cA$. We then claim that the resulting compiled algorithm, $\cA'$, is $f$-mobile secure. Observe that all the messages of $\cA$ exchanged over the edges of $E \setminus F^*$ are distributed uniformly at random, in the eyes of the adversary. We then use the $f$-static security guarantees of $\cA$ to show that the information exchanged over $F^*$, where $|F^*|\leq f$, leaks no information as well. 

We therefore conclude that our key task is in providing all, but at most $f$ neighboring pairs, a sufficiently large pool of secret keys, in the presence of the $f$-mobile adversary. This task is captured by the neat formulation of the \emph{Bit-Extraction} problem introduced by Chor et al. \cite{ChorGHFRS85}. In this problem, it is desired to extract random bits from several bits, where a bounded number of these bits are controlled by an adversary and the rest are uniformly distributed. 
%

To improve upon the extra $f$ factor overhead (when insisting on $f$-mobility, see Thm. \ref{thm:simulation}), we show that a white-box combination of the Bit-Extraction procedure of Chor et al. \cite{ChorGHFRS85} with the framework of \cite{EavesBroad2022,EavesMST2023} yields $f$-mobile algorithms with the same number of rounds. As an appetizer, we provide the following very simple, yet at first glance, surprising observation which serves as the basis for adapting \cite{EavesBroad2022,EavesMST2023} to the $f$-mobile setting. 

\smallskip
\noindent\textbf{Key Observation: Mobile-Secure Unicast is Easy.} At the heart of the algorithms of \cite{EavesBroad2022,EavesMST2023} lies a (static) secure unicast procedure of Jain \cite{Jain04}. This procedure allows a given pair of nodes $s,t$ to exchange a secret message in $O(D)$ rounds, provided that the set of edges $F$ controlled by the \emph{static} adversary does not disconnect $s$ and $t$. A remarkable property of this algorithm is its \emph{lightness}: exactly one message is exchanged along each of the graph edges (throughout the algorithm). This leads to a very simple mobile compilation: Let all neighbors $u,v$ exchange 
a random message $R(u,v)$, within a single round. Then, simulate Jain's algorithm in a round-by-round-manner where the messages are encrypted with the $\{R(u,v)\}$ keys.
It is then easy to prove perfect security provided that the following minimal condition holds. Let $F_i$ be the edges controlled by the adversary in round $i$. Then, security holds if $F_1$ does not disconnect $s$ and $t$ and $F_i=E$ for every $i \geq 2$. 

This exercise illustrates that mobile security is \emph{easy} when the given static-secure algorithm has low \emph{congestion}. Thm. \ref{thm:simulation} allows one to handle the general case of \emph{arbitrary} congestion. 
%

%% file: techniques-mobile-byzantine.tex
\subsubsection{Resilience with Mobile Byzantine Adversaries} \vspace{-5pt}

We now turn to consider the considerably more challenging task of providing resilience in the presence of $f$-mobile byzantine adversary. 
Unlike the mobile security setting, our simulation translates any fault-free algorithm into an $f$-mobile-resilient algorithm. This provides also an alternative approach for the $f$-static setting, in the regime where $f=\Omega(\log n)$. The main application of our technique is a general compiler in the \congc\ model, which can handle $\Theta(n)$ mobile byzantine faults (in every round!) while paying only a poly-logarithmic overhead in the number of rounds. To illustrate our ideas, we take a gradual approach in terms of the delta w.r.t prior work. 

\smallskip
\noindent\textbf{Handling $f=O(1)$ Mobile Faults with Fault-Tolerant (FT) Cycle Covers.} Patra et al. \cite{PatraCRSR09} presented an $f$-mobile-resilient algorithm that allows a node pair $s,t$ to exchange a message $m$. Their algorithm is based on sending $m$, in a pipeline manner, along $2f+1$ edge-connected $s$-$t$ paths, for a sufficient number of rounds. Note that the length of these paths can be bounded by $D^{\Theta(f)}$ where $D$ is the diameter of the graph, see e.g., \cite{parter2019small}. To simulate a fault-free algorithm $\cA$ in the $f$-mobile-byzantine setting, it is desired to employ the solution of \cite{PatraCRSR09} for all neighboring pairs $u,v$. A naive application leads to a super-linear round complexity, in the worst case, as a single edge might appear on the $u$-$v$ path collection of potentially $\Omega(n)$ many $u,v$ pairs. This \emph{congestion} barrier is mitigated by the notion of fault-tolerant (FT) cycle-covers \cite{ParterY19s,HitronP21a}. 

Informally, a $k$-FT cycle cover is a collection of cycles such that each edge $(u,v)$ is \emph{covered} by $(k-1)$ edge-disjoint cycles (except for the common edge $(u,v)$). \cite{ParterY19s} and \cite{HitronP21a} showed that any $k$ edge-connected $D$-diameter graph admits a $k$-FT cycle cover such that: (i) the largest cycle length is $D^{\Theta(k})$ and (ii) the largest (edge) overlap between the cycle is $D^{\Theta(k})$. Employing the algorithm of \cite{PatraCRSR09} for each neighboring pair $u,v$ on top of $k$-FT cycle cover for $k=2f+1$, allows us to compile any round of a given fault-free algorithm within $D^{\Theta(f)}$ rounds. 

The key limitation of this technique is in handling a larger number of faults. It is easy to show (by an averaging argument) that for any graph, any $k$-FT cycle cover induces a cycle overlap of $\Omega(k)$. Therefore, providing $\Theta(n)$-mobile \congc\ compilers with $\widetilde{O}(1)$ overhead calls for a new approach.

\smallskip
\noindent\textbf{Handling $f=\Omega(\log n)$ Mobile Faults with Low-Depth Tree Packing.}
The notion of low-diameter tree packing, introduced by Chuzhoy, Parter and Tan \cite{ChuzhoyPT20}, decomposes the graph into multiple near edge-disjoint trees of bounded depth. A graph $G$ is $(k,\kDiam)$-connected if for every pair $u,v$ there is a collection of $k$ edge-disjoint paths of length at most $\kDiam$. \cite{ChuzhoyPT20} presented a centralized construction that decomposes every $(k,\kDiam)$-connected graph into $O(k/\log n)$ near edge-disjoint spanning-trees of depth $O(\kDiam\cdot \log n)$. Our key result provides an $f$-mobile-resilient compilation of any fault-free algorithm while paying a round overhead of $\widetilde{O}(\kDiam)$, given a distributed knowledge\footnote{By distributed knowledge we mean that each node knows its parent in each of the trees.} of a $\kDiam$-diameter tree packing with $k=\widetilde{O}(f)$. 
We first explain a strategy for obtaining a round overhead of $\widetilde{O}(\kDiam+f)$.

\smallskip
\noindent\textbf{Compilation with a Round Overhead of $\widetilde{O}(\kDiam+f)$.} It is instructive to explain first the simulation in the $f$-static setting when given a collection of $k=\Omega(f \log n)$ nearly edge-disjoint spanning trees of depth $\widetilde{O}(\kDiam)$. We root all trees at some root node $v_r$.  Consider round $i$ of the given fault-free algorithm $\cA$, and let $m_i(u,v)$ be the message sent by $u$ to $v$ on that round, for every (directed) edge $(u,v)\in E$. At the start of the simulated round $i$, we let all nodes exchange the $\{m_i(u,v)\}_{(u,v)\in E}$ messages, as in Alg. $\cA$. Let $m'_i(u,v)$ be the message received by $v$ from $u$, on that round. As the adversary corrupts at most $f$ bidirectional edges, it might be that $m'_{i}(u,v)\neq m_i(u,v)$ for at most $2f$ ordered pairs $(u,v)$. We call a message $m_i(u,v)$ a \emph{mismatch} if $m'_{i}(u,v)\neq m_i(u,v)$, hence we have at most $2f$ mismatches, that we need to ``correct".  

We introduce a message-correction-procedure which is based on the powerful tool of \emph{sparse recovery sketches}, commonly employed in the context of the \emph{turnstile streaming} model \cite{CF14}. In that setting, we are given a stream of elements which arrive with some (positive or negative) frequency, and our goal at the end of the stream is to output all elements with non-zero frequency. Detecting $s$ elements can be done with a memory of $\widetilde{O}(s)$ bits. 

Consider a (turnstile) stream $S$ formed by adding each of the \emph{sent} messages $m_i(u,v)$ with frequency $1$, and each of the received messages $m'_i(u,v)$ with frequency $(-1)$. Since all the messages such that $m_i(u,v) = m'_i(u,v)$ cancel-out, we are left with only the sent and received copies of the mismatches. We utilize the mergeability property of the sparse recovery sketches, and aggregate local sketch information on each of the trees, in parallel. Initially, each node $v$ locally computes a sketch $\sigma(v)$ of its incoming and outgoing messages\footnote{I.e., it computes a stream in which $m_i(v,u)$ is added with a frequency $1$, and $m'_i(u,v)$ with frequency $-1$, for every $u \in N(v)$.}. By aggregating these $\sigma(v)$ values over the trees, the root $v_r$ obtains the final sparse recovery sketch, and detects the mismatches. Since the majority of trees do not have a corrupted edge (in the $f$-static setting), the sketch returned by the majority of the trees contains the correct list of mismatches, and we can broadcast this list through the trees to have all nodes correct their received messages. Since the sparse recovery sketches are implemented with a sparsity parameter of $s=\Theta(f)$, this computation can be implemented in $O(f+\kDiam)$ rounds using a pipelining argument. 

If implemented naively in a mobile setting, the adversary may alter the result of $f$ trees \emph{per} round, and eventually corrupt, at least a single edge, in each of the given $\Theta(f\log n)$ trees\footnote{It can corrupt $f$ edges in each round, and each edge appears on $O(\log n)$ many trees.}. This is the critical point where \emph{interactive-coding} comes to rescue. We use the compiler of \cite{RajagopalanS94,HozaS16}, denoted hereafter by \emph{RS-compiler}, to compile the sketch aggregation procedure in each of the trees. The RS-compiler is designed for a setting in which the adversary can maliciously corrupt an $O(1/m)$ fraction of the total communication, where $m$ is the number of graph edges. In our context, this compiler is applied on a tree subgraph, hence tolerating $O(1/n)$ fraction of corrupted messages, with a round-overhead $O(1)$. Since the $f$-mobile adversary may only corrupt $O(f\log n)$ many trees in any given round, for most of the RS-compiled protocols, the total fraction of corrupted communication is $o(1/n)$. Consequently, the majority of the RS-compiled protocols are \emph{successful}. 

On the conceptual level, the RS-compilers allow us to utilize the collection of $\Omega(f\log n)$ near edge-disjoint trees in an $f$-mobile setting in an \emph{almost} analogous manner to the $f$-static setting.
We cannot guarantee that a majority of the trees are fault-free (as we could, in the static case), but we can still guarantee that a majority of RS-compiled algorithms over these trees end successfully. This comes with a cost of increasing the edge-connectivity requirement by a constant factor which depends on the hidden constants of the RS-compilers.

\smallskip

\noindent\textbf{Omitting the Dependency in $f$.} The improved bound is obtained by replacing the sparse recovery sketches by $\ell_0$-sampling sketches which have only $\widetilde{O}(1)$ bits. 
The basic intuition for this procedure is the following: Given that we have $\Omega(f\log n)$ many spanning trees, if each tree propagates $O(\log{n})$ uniformly random real mismatches (obtained by independent $\ell_0$-sampling sketches), then the root observes all real mismatches w.h.p. Note, however, that some of these observed mismatches might be \emph{fake}, as the adversary might control some trees and introduce mismatches that are not obtained from the $\ell_0$-sketches. To overcome this, we set a minimal threshold $\Delta$, and make the root node $v_r$ ignore observed mismatches that are sampled by less than $\Delta$ trees. The threshold $\Delta$ should be set with care: high enough to filter-out unreal mismatches, but also sufficiently low to detect \emph{many} real mismatches. As we cannot expect to capture all $2f$ (real) mismatches at once while producing no new mismatches, we have $\ell=O(\log f)$ repetitions.

At the beginning of each phase $j \in [\ell]$, each nodes $v$ recomputes its local $\ell_0$-sketches, based on the current estimate $m'_{i,j}(u,v)$ of its received messages, for each $u \in N(v)$. Our goal is to reduce the number of mismatches by a constant factor in each phase, hence eventually correcting all mismatches within $O(\log f)$ phases.  For phase $j$, we define a threshold $\Delta_j=\widetilde{O}(2^j)$ and the root node $v_r$ only considers the mismatches that received a support by at least $\Delta_j$ trees, and ignores the rest. These highly-supported mismatches are downcast from $v_r$ to all the nodes, on each of the trees, in parallel. Assume, for now, that all the nodes correctly received this information from $v_r$. Then, one can show by induction, that the number of unfixed (real) mismatches drops by a constant factor in a phase. As the number of real mismatches decreases, each real mismatch will be sampled more times by the good trees which allows us to increase the supported threshold $\Delta_j$ accordingly. 

The remaining caveat is the assumption on correctly receiving the root's information. This might not hold, in general, as the adversary may introduce $f$ incorrect mismatches in each phase when downcasting the sketch information from the root. To overcome this last hurdle, we combine error correction codes with the RS-compilers. The root $v_r$ encodes the $O(f)$ detected mismatches by a codeword $w$. It splits $w$ into $\Theta(f)$ shares and broadcasts each share on some tree using the RS-compiler. This guarantees that a large fraction of the trees broadcasts their share correctly, and each node can locally recover the list of observed mismatches.   

\smallskip
\noindent\textbf{Handling Adversaries with Bounded Round-Error-Rates.} 
In \Cref{thm:budgeted_main_intro}, we consider a stronger setting in which the adversary can corrupt  ``on average'' $f$ messages in each round. In particular, the adversary might corrupt a large number of messages in given rounds. The compiler is based on the \emph{rewind-if-error} technique \cite{S92}, originally introduced for the two-party setting. On a high level in this paradigm parties keep on verifying whether or not errors have occurred so-far in the protocol execution. If there is no indication of errors, the parties continue to simulate the next round. Otherwise, they \emph{rewind} by omitting the possibly incorrect last messages, and repeat. The key challenges is in detecting the errors and deciding simultaneously whether to rewind or not. 

We provide a network extension to this paradigm, that is somewhat different than the approach taken in prior works, e.g., in \cite{HozaS16}. Recall that the classical interactive coding setting allows a communication-error-rate, while we account for round-error-rate. In each given point of our compilation, each node $u$ simulates a round $t_u$ of Alg. $\cA$, where possibly $t_u \neq t_v$ for distinct nodes $u,v$. Once errors are detected, only the nodes of largest $t_u$ values apply a rewind step. The analysis is based on defining a potential function which provides a global progress guarantees for the entire network, over time. 

%% file: prelim.tex
\subsection{Preliminaries}

For an integer $a$, we denote by $[a] = \{0,1,\dots,a-1\}$. For a matrix $M$, let $M_{ij}$ its value in index $(i,j)$. 

\begin{definition}[Vandermonde Matrix]
Given a field $\mathbb{F}$, an $k \times n$ matrix $A$ is called a \emph{Vandermonde matrix} if there exist some $k$ distinct non-zero field elements $\alpha_1,\dots,\alpha_k \in \mathbb{F}$, such that the value $A_{ij} = \alpha_i^{j-1}$, where the multiplication is defined by the multiplication operator of the field.
\end{definition}

\noindent \textbf{Error-Correcting Codes.} We recall the definition of error correcting codes and the standard Reed-Solomon code construction. We use the notion of \emph{Hamming} distance used in coding theory and then define error correcting codes with its various parameters. 

\begin{definition}[Distance]
Let $\Sigma$ be a finite set and $\ell\in\mathbb{N}$, then the distance between $x,y\in \Sigma^\ell$ is defined by
$\Hamm(x,y) = \abs{\sett{i\in[\ell]}{x_i\neq y_i}}$.
\end{definition}

\begin{definition}[Error Correcting Code]\label{def:ECC}
Let $\Sigma$ be a finite set. For every $k\in\mathbb{N}$, a subset $C\subseteq \Sigma^k$ is said to be an error correcting code with block length $k$, message length $\ell$, and relative distance $\delta$ if $|C|\ge |\Sigma|^\ell$ and for every $x,y\in C$, $\Hamm(x,y)\geq \delta \cdot k$. We denote then $\Hamm(C)=\delta$. Moreover, we say that $C$ is a \code{\ell,k,\delta}{q} code to mean that $C$ is a code defined over alphabet set of size $q$ and is of message length $\ell$, block length $k$ and relative distance $\delta$. The elements of $C$ are denoted as codewords. 
\end{definition}


\begin{theorem}[Reed-Solomon Codes \cite{RS60}]\label{thm:rs}
For every prime power $p$, message length $\ell$ and block size $k\leq p^m=q$, there exists a \code{\ell, k, \delta_C}{q} code for $\delta_C=(k-\ell+1)/k$.
\end{theorem}

\noindent\textbf{Graph Notations.} For a graph $G=(V,E)$ and $u \in V$, let $N(u)$ denote the neighbors of $u$ in $G$. For a given $G$-subgraph family $\mathcal{G}=\{G_1,\ldots, G_k\}$, let $\load(e)=|\{ G_i \in \mathcal{G} ~\mid~ e \in G_i\}|$ for every $e \in E(G)$ and $\load(\mathcal{G})=\max_{e \in E(G)}\load(e)$. We say that a subgraph family $\mathcal{G}$ is \emph{known in a distributed manner} if each $u \in V(G_i)$ knows an ID of $G_i$ and its incident edges in $G_i$, for every $G_i \in \mathcal{G}$. 

For a rooted spanning tree $T$ and node $v \in V$, we denote by $\mathrm{Children}(v,T) \subseteq V$ the set of child nodes of $v$ in $T$.

For vertex set $A \subseteq V$ we denote by $E(A) \subseteq E$ the set of edges incident to $A$ in $G$. For vertex sets $A,B \subseteq V$, we denote by $E(A,B) \subseteq E$ the set of edges in $G$ with one endpoint in $A$ and one endpoint in $B$. We say that a graph $G$ is a $\phi$-\emph{expander} if for any set $S \subseteq V$ it holds that $(|E(S,V \setminus S)|)/\min(|E(S)|,|E(V\setminus S)|) \geq \phi$. This is also known as the graph $G$ having \emph{conductance} $\geq \phi$.

\subsection{Model, Security and Resilience Notions} 

\noindent \textbf{The Adversarial Communication Model.}	Throughout, we consider the adversarial \congest\ model introduced by \cite{HitronP21a,HitronP21}. The synchronous communication follows the standard $B$-\congest\ model, where initially, each node knows the IDs of its neighbors in the the graph $G$. This is usually referred to as the KT1 setting \cite{AGPV90}. In each round, nodes can exchange $B$-bit messages on all graph edges for $B=O(\log n)$. Some of our results hold in the \congc\ model \cite{LotkerPPP05}, in which \emph{each} pair of nodes (even non-neighboring) can exchange $O(\log n)$ bits, in every round. 

We study two main (edge) adversarial settings, namely, eavesdroppers and byzantine adversaries. 
All adversarial settings considered in this paper assume an \emph{all-powerful} adversary that controls subsets of \emph{edges} whose identity is not known to the nodes. The adversary is allowed to know the topology of the graph $G$ and the algorithm description run by the nodes. It is oblivious, however, to the randomness of the nodes.  In the case of \emph{active} adversaries (e.g., byzantine), the adversary is allowed to send $B$-bit messages, on each of the edges it controls (in each direction), in every round of the computation.  In the static settings, the adversary controls a fixed set of at most $f$ edges, while in the mobile setting, it is allowed to control a distinct set of $f$ edges in each round. In the context of (passive) eavesdroppers, we aim at providing perfect-security guarantees. For byzantine adversaries, we strive for correctness. It will be interesting to extend our techniques to provide both correctness and security guarantees against a byzantine adversaries. We next formally define the desired security and resilience guarantees under these adversarial settings, respectively.

\smallskip
\noindent\textbf{Perfect Security with Eavesdroppers.}	In the \emph{static} adversarial setting, a (computational unbounded) eavesdropper adversary controls a fixed set of edges $F^*$ in the graph. The nodes do not know the identity of $F^*$, but rather only a bound on $|F^*|$. In the \emph{mobile} eavesdropper setting, the adversary is allowed to control a distinct subset of edges $F_i$ in each round $i$. 

Let $\cA$ be a randomized algorithm running on a graph $G$. Denote the input domain of the algorithm $\cA$ by $\mathcal{X}$. We say that an eavesdropper is \emph{listening} over an edge $(u,v)$ in round $i$ of Alg. $\cA$ if the eavesdropper observes the message that $u$ sent to $v$ and the message that $v$ sent to $u$ in round $i$.  For a subset of edges $F^* = \{e_1,\dots,e_k\} \subseteq E$, and input $x \in \mathcal{X}$, let $\View_{A}(F^*,x)$ be a random variable vector indicating the messages of the edges of $F^*$ throughout the execution of $\cA$ given input $x$.  

Algorithm $\cA$ is said to be $f$-\emph{static-secure} against an eavesdropper adversary, if for every choice of $|F^*| \leq k$, and every possible configuration of input values $x_1,x_2 \in \mathcal{X}$, it holds that the following two are equivalent distributions: $\View_{G,A}(F^*,x_1) \equiv \View_{G,A}(F^*,x_2)$. This notion is known as \emph{perfect security}. For an $r$-round algorithm $\cA$, input $x \in \mathcal{X}$ and a collection of $r$ subsets of edges, $F_1,\dots,F_r \subseteq E$ let $\View_{A}((F_1,\dots,F_r),x)$ be a random variable vector for the messages exchanged over each edge $e \in F_i$ at round $i$ given the input $x$, for all $1\leq i \leq r$. Alg. $\cA$ is $f$-\emph{mobile-secure} in a graph $G$ if for any $F_1,\dots,F_r \subseteq E$, of size $|F_i| \leq f$ and for any inputs $x_1,x_2 \in \mathcal{X}$ it holds that $\View_{G,A}((F_1,\dots,F_r),x_1) \equiv \View_{G,A}((F_1,\dots,F_r),x_2)$.
	
\smallskip
\noindent \textbf{Resilience with Byzantine Adversaries.} The graph edges are controlled by a computationally unbounded \emph{byzantine} adversary. Unlike the eavesdropper setting, the adversary is allowed to see the messages sent through \emph{all} graph edges in each round, but can manipulate the messages exchanged over a bounded subset of controlled edges. 
%
An $f$-\emph{static} byzantine adversary can manipulate the messages sent through a fixed $F^* \subseteq E$ where $|F^*|\leq f$.  An $f$-\emph{mobile} byzantine adversary can manipulate at most $f$ edges $F^*_i$ in \emph{each} round $i$, where possibly $F^*_i\neq F^*_{j}$ for $i\neq j$. 

We say that an algorithm is $f$-\emph{static} (resp., \emph{mobile}) \emph{resilient} if its correctness holds in the presence of $f$-static (resp., mobile) byzantine adversary. In the stronger setting of $f$-\emph{round-error-rate}, the adversary is allowed to corrupt at most $f$ edges per round, on \emph{average}. That is, for an $r$-round algorithm the adversary is allowed to corrupt a total of $f \cdot r$ edges. 
%

\smallskip
\noindent \textbf{Distributed Scheduling.} The congestion of an algorithm $\cA$ is defined by the worst-case upper bound on the number of messages exchanged through a given graph edge when simulating $\cA$.  Throughout, we make an extensive use of the following random delay approach of \cite{leighton1994packet}, adapted to the \congest\ model. 
\begin{theorem}[{\cite[Theorem 1.3]{Ghaffari15}}]\label{thm:delay}
	Let $G$ be a graph and let $\cA_1,\ldots,\cA_m$ be $m$ distributed algorithms, each algorithm takes at most $\dilation$ rounds, and where for each edge of $G$, at most $\congestion$ messages need to go through it, in total 
	\emph{over all} these algorithms. Then, there is a randomized distributed
	algorithm that w.h.p.\ runs all the algorithms in $\widetilde{O}(\congestion +\dilation)$ rounds.
\end{theorem}

\subsection{Useful Tools}

\begin{lemma}[Chernoff Bound]
	\label{lem:chernoff}
	Let $X_1,\dots,X_n$ be i.i.d random variables over the values $\{0,1\}$. Let $X = \sum^n_i X_i$ and $\mu = E(X)$. Then for any $0 < \delta < 1$, $\Pr(X \leq (1-\delta) \mu) \leq e^{-\mu \delta^2/2}$.
\end{lemma}

\smallskip
\noindent \textbf{Families of bounded-independence hash functions.} Some of our algorithms are based on generating $\cj$-wise independent random variables from a short random seed. For that purpose, we use the concept of families of bounded-independence hash functions:
\begin{definition}
	For $N, L, \cj \in \nat$ such that $\cj \le N$, a family of functions $\mathcal{H} = \{h : [N] \rightarrow [L]\}$ is \emph{$\cj$-wise independent} if for all distinct $x_1, \dots, x_\cj \in [N]$, the random variables $h(x_1), \dots, h(x_\cj)$ are independent and uniformly distributed in $[L]$ when $h$ is chosen uniformly at random from~$\mathcal{H}$.
\end{definition}
\begin{lemma}
	\label{lem:hash}[Corollary~3.34 in \cite{Vadhan12}]
	For every $a$, $b$, $\cj$, there is a family of $\cj$-wise independent hash functions $\mathcal{H} = \{h : \{0,1\}^a \rightarrow \{0,1\}^b\}$ such that choosing a random function from $\mathcal{H}$ takes $\cj \cdot \max\{a,b\}$ random bits, and evaluating a function from $\mathcal{H}$ takes $\poly(a,b,\cj)$ computation.
\end{lemma}

%% file: roadmap.tex
\smallskip
\noindent \textbf{Roadmap.} The paper is split into two parts: security against an eavesdropper adversary, and resilience towards a byzantine adversary. In the first part, we prove \Cref{thm:simulation} in \Cref{sec:mobile-Eavesdropper}, and \Cref{thm:mobile-compilers} in \Cref{sec:fault_free_to_f_secure}. In the second part, we prove \Cref{thm:tree-packing-comp}, \Cref{thm:CC-compiler} and \Cref{thm:expander-compiler} in \Cref{sec:exactly_k_edges_improved}. Results for the round-error rate setting of a byzantine adversary are proven in in \Cref{sec:byz_budgeted}. \Cref{cor:CC-mobile-simulation-final} is proven in \Cref{sec:FTcycle-cover}.

%% file: eavesdropper-mobile.tex
\section{Security with Mobile Eavesdropper Adversaries}\label{sec:mobile-Eavesdropper}


In this section we prove \Cref{thm:simulation} by providing a round-efficient simulation that converts an $r$-round $f$-static-secure into an $r'$-round $f'$-mobile-secure algorithm. The main lemma optimizes the ratios $r'/r$ and $f'/f$ by reducing to the problem of Bit-Extraction introduced by Chor et al. \cite{ChorGHFRS85}. 


\paragraph{The Bit-Extraction Problem and $t$-Resilient Functions.} Let $n,m,t$ be arbitrary integers. In \cite{ChorGHFRS85}, Chor et al. considered the adversarial situation where for a given vector $x \in \{0,1\}^n$ the adversary knows $t$ entries in $x$ while the remaining $n-t$ entries are uniformly distributed in $\{0,1\}^{n-t}$. It is then required to output $m$ uniform random bits that are completely hidden from the adversary. The question is how large can $m$ be as a function of $n$ and $t$. 

\begin{definition}
Let $f:\{0,1\}^{n\cdot k} \to \{0,1\}^{m \cdot k}$ be a function and $\{y_1,\ldots, y_n\}$ be a set of random variables assuming values in $\{0,1\}^k$. The function $f$ is said to be $k$-unbiased with respect to $T \subset \{1,2,\ldots, n\}$ if the random variable $f(y_1,y_2, \ldots, y_n)$ is uniformly random on $\{0,1\}^{m \cdot k}$ when $\{y_i ~\mid~ i \notin T\}$ is a set of independent uniformly random variables on $\{0,1\}^k$ and $\{y_i ~\mid~ i \in T\}$ is a set of constant random variables. A function $f:\{0,1\}^{n\cdot k} \to \{0,1\}^{m \cdot k}$ is $(t,k)$-\emph{resilient} if for every $T \subseteq \{1,2,\ldots, n\}$ of cardinality $t$, $f$ is $k$-unbiased w.r.t $T$. 
\end{definition}

Let $B_k(n,t)$ be the maximum $m$ such that there exist a $(t,k)$-resilient function $f:\{0,1\}^{nk} \to \{0,1\}^{mk}$. In the (Block) Extraction Problem given $n,t$, it is required to determine $B_k(n,t)$. 

\begin{theorem}[\cite{ChorGHFRS85}] \label{thm:extract}
For $n\leq 2^k-1$ it holds that $B_k(n,t)=n-t$. Moreover, the following explicit function $f$ obtains this bound: let $M$ be an arbitrary $n \times (n-t)$ Vandermonde matrix over the finite field $\mathbb{F}_{2^k}$. Then for any random variables $x_1,\dots,x_n  \in \mathbb{F}_{2^k}$ such that at least $n-t$ of the $x_i$'s are uniform random variables on $\mathbb{F}_{2^k}$ and the rest are constants, the values $y_1,\dots,y_{n-t} \in \mathbb{F}_{2^k}$,  defined as $y_i = \sum_{j=1}^{n} M_{ji} \cdot x_j$,
where operations are made over the field $\mathbb{F}_{2^{k}}$, are independent uniform random variables on $\mathbb{F}_{2^k}$.
\end{theorem}

For further applications of Vandermonde matrices in static byzantine settings, see \cite{BH08}.

%
%
\smallskip
\noindent \textbf{The Static $\to$ Mobile Simulation.} Assume we are given an $r$-round protocol $\mathcal{A}$ which is $f$-static-secure with round complexity $r \in \poly(n)$, and an integer parameter $t$. Our goal is to construct an $r' = 2r+t$ round algorithm which is $f' =\Theta((f \cdot t)/(r+t))$-mobile-secure. Let $\mathbb{F}_q$ be a finite field of size $q = 2^{O(\log{n})}$ and let $M$ be an arbitrary $(r+t) \times r$ Vandermonde matrix over the field $\mathbb{F}_q$. Assume all messages in $\mathcal{A}$ are encoded as elements of $\mathbb{F}_q$.

Algorithm $\cA'$ has two phases, the first phase consists of $\ell=r+t$ rounds, and the second has $r$ rounds. In the first phase, for every $j=1,\dots,r+t$ rounds, for each ordered neighboring pair $(u,v) \in E$, $u$ sends to $v$ a uniform random number $R_j(u,v) \in \mathbb{F}_q$. At the end of this phase, each node $u$ locally computes for every $1 \leq i \leq r$, the  values $K_i(u,v)$ and $K_i(v,u)$ for every neighbor $v$, defined as $K_i(u,v) = \sum_{j=1}^{r+t} M_{ji} \cdot R_j(u,v)$.\footnote{all $+,\times$ operations on field elements are done over the field $\mathbb{F}_q$ throughout the section.}
 The second phase simulates Alg. $\mathcal{A}$ in a round by round fashion, where the $i$-round messages of $\mathcal{A}$ are encrypted using the keys $\{K_i(u,v)\}_{(u,v)\in E}$: For $i=1,\ldots,r$ rounds, each node $u$ sends to each neighboring node $v$ the message $m'_i(u,v) = m_i(u,v) + K_i(u,v)$, where $m_i(u,v)$ is the message $u$ sends $v$ in the $i$'th round of $\mathcal{A}$. Locally, each $v$ decodes every received $m'_i(u,v)$ by applying $m_i(u,v) = m'_i(u,v) - K_i(u,v)$. Consequently, it is easy to see that each node $u$ obtains the exact same messages as in Alg. $\cA$.


\def\APPENDSECUREEAVES{
\paragraph{Proof of Theorem \ref{thm:simulation}.} As correctness and running time follow immediately, we focus on showing that Algorithm $\cA'$ is $f'$-mobile-secure. Consider a simulation of $\cA'$ and for every round $i \in \{1,\ldots, r'\}$, let $F_i$ be the set of edges that the adversary eavesdrops on that round, where $|F_i|\leq f'$. 
We partition the edges of $G$ into two classes $\Egood$ and $\Ebad$ depending on the total number of rounds that the given edge has been eavesdropped by the adversary. The input parameter $t$ serves as a threshold that determines the partitioning, as follows. For every edge $e$, let $R(e)=|\{F_i ~\mid~ e \in F_i, i \in \{1,\ldots, \ell\}\}|$ be the number of rounds in which $e \in F_i$ for $i \in \{1,\ldots, \ell\}$. An edge $e$ is denoted as \emph{good} if $R(e)\leq t$ and it is \emph{bad} otherwise. The set $\Egood$ (resp., $\Ebad$) consists of all good (resp., bad) edges. I.e.,  
$\Egood=\{e \in E ~\mid~ R(e)\leq t\}$ and $\Ebad=E(G)\setminus \Egood$. 
By an averaging argument, we have that $|\Ebad|\leq \lfloor (f'\cdot \ell)/(t+1) \rfloor \leq f$. 
For the special case of $t \geq 2fr$, note that since $|\Ebad|$ is an integer, then $|\Ebad|\leq \lfloor (f'\cdot \ell)/(t+1) \rfloor$. Observe that the condition $t \geq 2rf$ is equivalent to $t \geq (t+r)/(1+1/2f) = \ell/(1+1/2f)$, hence $\lfloor \frac{\ell}{t+1}f \rfloor = f$.  Therefore, $|\Ebad|\leq \lfloor (f'\cdot \ell)/(t+1) \rfloor = \lfloor f\ell/(t+1) \rfloor = f$. 

\begin{observation}
(i) For every $e=(u,v)\in \Egood$, it holds that $\{K_i(u,v)\}_{i \in \{1,\ldots, r\}}$ are distributed uniformly at random in $\mathbb{F}_q$. (ii) $|\Ebad| \leq f$.
\end{observation}
\begin{proof}
(i) follows by a direct application of Theorem \ref{thm:extract} and (ii) follows by a simple counting argument. The eavesdropper controls at most $f'$ edges in each round, and therefore in the first phase of $\ell$ rounds, it controls at most $f'\ell$ edges. Therefore at most $(f'\cdot \ell)/(t+1)$ edges are controlled for at least $t+1$ rounds. 
\end{proof}

\noindent Let $\mathcal{X}$ be the input domain of $\mathcal{A}$. Denote by $\MessageStatic{i}{\mathcal{B}(x)}{F}$ the messages sent over edges in $F$ at round $i$ of algorithm $\mathcal{B}$ with input $x$. Throughout, we treat the messages in algorithms $\cA$ and $\cA'$ as field elements in $\mathbb{F}_q$. Assume by contradiction that $\mathcal{A}'$ is not $f'$-mobile-secure. Then there exist some $F_1,\dots,F_{r'} \subseteq |E|$ of size at most $f'$, and inputs $x_1,x_2 \in \mathcal{X}$ for which
\begin{equation}
	\label{eq:modified_is_secure}
	\View_{G,\mathcal{A}'}((F_1,\dots,F_{r'}),X = x_1) \not\equiv \View_{G,\mathcal{A}'}((F_1,\dots,F_{r'}),X = x_2).
\end{equation}
	%
Let $P_i = F_i \cup \Ebad$ and $P=(P_1,\ldots, P_{r'})$. It therefore also holds:
\begin{equation*}
	\View_{G,\mathcal{A}'}((P_1,\dots,P_{r'}),X = x_1) \not\equiv \View_{G,\mathcal{A}'}((P_1,\dots,P_{r'}),X = x_2).
\end{equation*}

Hence in particular, there exist $\alpha_1,\dots,\alpha_{r'} \in (\mathbb{F}_q)^{\leq f}$ such that:
\begin{equation}
\Pr\left(\Message{1}{\mathcal{A}'(x_1)}{P} = \alpha_{1},\dots,\Message{r'}{\mathcal{A}'(x_1)}{P} = \alpha_{r'} \right) \neq \Pr\left(\Message{1}{\mathcal{A}'(x_2)}{P} = \alpha_{1},\dots, \Message{r'}{\mathcal{A}'(x_2)}{P} = \alpha_{r'}\right).
\end{equation}
That is, in our notation, $\alpha_i$ is a vector of $|P_{i}|\leq f$ field elements in $\mathbb{F}$, where the $j^{th}$ entry, namely $\alpha_{i,j}$, specifies the message sent over the $j^{th}$ edges in $P_{i}$ in the $i^{th}$ round. 
We also define two $r'$-sets of subset of edges $\Pgood =(F_1 \cap \Egood, \ldots, F_{r'} \cap \Egood)$, and $\Pbad=(\Ebad, \ldots, \Ebad)$. For any $r'$-edge set $W = (W_1 \subseteq P_1,\ldots, W_{r'} \subseteq P_{r'})$, Algorithm $\mathcal{B}$, input $x \in \mathcal{X}$ and indices $0 \leq i \leq j \leq r'$, let $\MessageEvent{\mathcal{B}(x)}{i}{j}{W}{}$ denote the event where:
\[\Message{i}{\mathcal{B}(x)}{W} = \alpha_i(W_i),\dots,\Message{j}{\mathcal{B}(x)}{W} = \alpha_j(W_j),\] where $\alpha_j(W_j)$ denotes the sub-vector of $\alpha_j$ restricted on the coordinates of $W_j$. Since the communication in the first $\ell$ rounds depends only on the private randomness of the nodes, and does not depend on the input $X$ of Alg. $\cA$, for any such $r'$-edge set $W$, it holds:
	
\begin{equation}
\label{eq:first_phase_indep_of_input}
	\Pr\left(\MessageEvent{\mathcal{A}'(x_1)}{1}{\ell}{W}{\alpha}\right) = \Pr\left(\MessageEvent{\mathcal{A}'(x_2)}{1}{\ell}{W}{\alpha}\right).
	\end{equation}

	By \Cref{thm:extract}, all keys of the good edges $\{K_i(u,v)\}_{\{u,v\} \in \Egood; i \in [r]}$ are i.i.d distributed in $\mathbb{F}_q$ conditioned on the transcript of messages exchanged over\footnote{Note that this holds despite the fact that $|P_i|$ might be larger than $f'$, as we only added $\Ebad$ to $F_i$.} $P_j$ in round $j$ for every $j \in \{1,\ldots, \ell\}$.  The security provided by the OTP guarantees that all the messages exchanged over $\Pgood$ in the second phase are distributed i.i.d in $\mathbb{F}_q$ even when conditioned on the observed transcript. Therefore,
\begin{flalign*}
	&\Pr\left(\MessageEvent{\mathcal{A}'(x_1)}{\ell+1}{r'}{\Pgood}{\alphagood} \stt \MessageEvent{\mathcal{A}'(x_1)}{1}{\ell}{\Pgood}{\alphagood}\right) = \Pr\left(\MessageEvent{\mathcal{A}'(x_2)}{\ell+1}{r'}{\Pgood}{\alphagood} \stt \MessageEvent{\mathcal{A}'(x_2)}{1}{\ell}{\Pgood}{\alphagood}\right).
	\end{flalign*}
	
By combining with Equation~(\ref{eq:first_phase_indep_of_input}), we get:

	\begin{equation}
	\label{eq:good_edges}
	\Pr\left(\MessageEvent{\mathcal{A}'(x_1)}{\ell+1}{r'}{\Pgood}{\alphagood}\right) = \Pr\left(\MessageEvent{\mathcal{A}'(x_2)}{\ell+1}{r'}{\Pgood}{\alphagood}\right).
	\end{equation}
	
Moreover, \Cref{thm:extract} implies that the observed transcript $\Pgood_{\ell+1},\dots,\Pgood_{r'}$ in the second phase is independent of the transcript of $\Pbad$ in the second phase, when conditioned on the view of the adversary in the first phase. Therefore, for $P = (P_1,\dots,P_{r'})$,
	\begin{equation}
	\label{eq:good_and_bad_edges_indep}
	\begin{aligned}
	&\Pr\left(\MessageEvent{\mathcal{A}'(x_1)}{\ell+1}{r'}{P}{\alpha} \stt \MessageEvent{\mathcal{A}'(x_1)}{1}{\ell}{P}{\alpha}\right) \\&= \Pr\left(\MessageEventStatic{\mathcal{A}'(x_1)}{\ell+1}{r'}{\Pbad}{\alphabad} \stt \MessageEventStatic{\mathcal{A}'(x_1)}{1}{\ell}{\Pbad}{\alphabad}\right) \cdot \Pr\left(\MessageEvent{\mathcal{A}'(x_1)}{\ell+1}{r'}{\Pgood}{\alphabad} \stt \MessageEvent{\mathcal{A}'(x_1)}{1}{\ell}{\Pgood}{\alphabad}\right)~.
	\end{aligned}
	\end{equation} 
	
\begin{claim}\label{cl:bad-equality}
$\Pr\left(\MessageEvent{\mathcal{A}'(x_1)}{\ell+1}{r'}{\Pbad}{\alphabad} \stt \MessageEvent{\mathcal{A}'(x_1)}{1}{\ell}{\Pbad}{\alphabad}\right) = \Pr\left(\MessageEvent{\mathcal{A}'(x_2)}{\ell+1}{r'}{\Pbad}{\alphabad} \stt  \MessageEvent{\mathcal{A}'(x_2)}{1}{\ell}{\Pbad}{\alphabad}\right)~.$
\end{claim}
\begin{proof}	
Since $\mathcal{A}$ is $f$-static-secure and $|\Ebad|\leq f$, we have:
	\begin{equation}
	\label{eq:security_aprime}
	\Pr\left(\MessageEvent{\mathcal{A}(x_1)}{1}{r}{\Pbad}{\alphabad}\right) = \Pr\left(\MessageEvent{\mathcal{A}(x_2)}{1}{r}{\Pbad}{\alphabad}\right).
	\end{equation}
	
By the second phase of Alg. $\cA'$, for any edge $(u,v) \in E(G)$ and any $x \in \mathcal{X}$, we have $\MessageStatic{\ell+i}{\mathcal{A}'(x)}{(u,v)} = \MessageStatic{i}{\mathcal{A}(x)}{(u,v)} + K_i(u,v)$. Therefore,	
\begin{flalign*}	
	&\Pr\left(\MessageEventStatic{\mathcal{A}'(x_1)}{\ell+1}{r'}{\Pbad}{\alphabad} \stt \MessageEventStatic{\mathcal{A}'(x_1)}{1}{\ell}{\Pbad}{\alphabad}\right) = \Pr\left(\MessageEventStatic{\mathcal{A}(x_1)}{1}{r}{\Pbad}{\alphabad}\right) =\Pr\left(\MessageEventStatic{\mathcal{A}(x_2)}{1}{r}{\Pbad}{\alphabad}\right) \\&= 
	\Pr\left(\MessageEventStatic{\mathcal{A}'(x_2)}{\ell+1}{r'}{\Pbad}{\alphabad} \stt \MessageEvent{\mathcal{A}'(x_2)}{1}{\ell}{\Pbad}{\alphabad}\right)~,
	\end{flalign*}
where the second equality holds due to Equation~(\ref{eq:security_aprime}). 
\end{proof}

\noindent We next show the following which provides a contradiction to our assumption and establishes the security guarantees of Alg. $\cA'$. 
\begin{claim}\label{cl:bad-final}
$\Pr\left(\MessageEvent{\mathcal{A}'(x_1)}{1}{r'}{P}{\alpha}\right) =\Pr\left(\MessageEvent{\mathcal{A}'(x_2)}{1}{r'}{P}{\alpha}\right)$~.
\end{claim}
\begin{proof}
Combining Cl. \ref{cl:bad-equality} with Equality (\ref{eq:first_phase_indep_of_input}), we obtain
\begin{equation}
\label{eq:bad_edges}
\Pr\left(\MessageEventStatic{\mathcal{A}'(x_1)}{1}{r'}{\Pbad}{\alphabad}\right) = \Pr\left(\MessageEventStatic{\mathcal{A}'(x_2)}{1}{r'}{\Pbad}{\alphabad}\right).
\end{equation}
Therefore,
	\begin{flalign*}
	\Pr\left(\MessageEvent{\mathcal{A}'(x_1)}{1}{r'}{P}{\alpha}\right) &= \Pr\left(\MessageEvent{\mathcal{A}'(x_1)}{1}{r'}{\Pgood}{\alphagood}\right) \cdot \Pr\left(\MessageEventStatic{\mathcal{A}'(x_1)}{1}{r'}{\Pbad}{\alphabad}\right) \\&= 
	\Pr\left(\MessageEvent{\mathcal{A}'(x_2)}{1}{r'}{\Pgood}{\alphagood}\right) \cdot \Pr\left(\MessageEventStatic{\mathcal{A}'(x_1)}{1}{r'}{\Pbad}{\alphabad}\right) \\&=
	\Pr\left(\MessageEvent{\mathcal{A}'(x_2)}{1}{r'}{\Pgood}{\alphagood}\right) \cdot \Pr\left(\MessageEventStatic{\mathcal{A}'(x_2)}{1}{r'}{\Pbad}{\alphabad}\right) \\&=
	\Pr\left(\MessageEvent{\mathcal{A}'(x_2)}{1}{r'}{P}{\alpha}\right)~,
	\end{flalign*}
where the first equality follows from Equation~(\ref{eq:good_and_bad_edges_indep}), the second equality from Equation~(\ref{eq:good_edges}), and the third equality from Equation~(\ref{eq:bad_edges}). The claim follows.
\end{proof}

}
\APPENDSECUREEAVES

%% file: byzantine-improved-new-merged.tex
\section{Resilience with Mobile Byzantine Adversaries} \label{sec:exactly_k_edges_improved}

Our goal in this section is providing $f$-mobile-resilient distributed algorithms for graphs with edge-connectivity of $\Omega(f\log n)$. Unlike our results in Sec. \ref{sec:mobile-Eavesdropper}, the $f$-mobile-resilient algorithms are not obtained by translating an $f$-static-resilient algorithm into an $f'$-mobile algorithm, but rather translating any given (non-faulty) distributed $r$-round \congest\ algorithm $\cA$ into an equivalent $r'$-round $f$-mobile resilient algorithm $\cA'$.


\input{tools.tex}

\subsection{$f$-Resilient Compilation with Round Overhead $\widetilde{O}(\kDiam)$}  
\label{sec:comp-exactly_k_edges_fast}

Throughout, we consider $(k,\kDiam)$-connected graphs for $k=d \cdot f\log n$, for some constant $d$ to be specified later. It is assumed that a weak $(k,\kDiam, \eta)$ tree packing of $G$ is known in a distributed manner, where each node knows its parent in each of the trees. We show the following stronger statement of Thm. \ref{thm:tree-packing-comp}. 

\begin{theorem}\label{thm:res-mobile-improved}
Given a distributed knowledge of a \emph{weak} $(k,\kDiam,\eta)$ tree packing $\mathcal{T}$ for graph $G$, then for any $r$-round algorithm $\cA$ over $G$, there is an equivalent $r'$-round algorithm $\cA'$ for $G$ that is $f$-mobile resilient, where $r'=\widetilde{O}(r \cdot \kDiam)$ and $f=\Theta(k /\eta)$. 
\end{theorem}

\noindent \textbf{Immediate Corollary: Compilers in the Congested-Clique Model.} An $n$-node clique over vertices $V=\{v_1,\ldots, v_n\}$ contains a $(k,\kDiam,\eta)$ tree packing $\mathcal{T}=\{T_1,\ldots, T_k\}$ for $k=n$ and $\kDiam,\eta=2$.  Specifically, for every $i \in \{1,\ldots, n\}$, let $T_i = (V,E_i)$ where $E_i = \{(v_i,v_j) \mid v_j \in V\}$, be the star centered at $v_i$. It is easy to see that the diameter and the load are exactly $2$. Theorem \ref{thm:CC-compiler} is then immediate by Theorem \ref{thm:res-mobile-improved}. 

 %

\subsubsection{Sub-Procedure for Safe Broadcast} 
Before presenting the proof of Theorem \ref{thm:res-mobile-improved}, we present a useful sub-procedure, called $\ECCSafeBroadcast$, which allows us to broadcast messages in an $f$-mobile-resilient manner by using Error-Correcting-Codes. 

We assume the network has a distributed knowledge of a weak $(k,\kDiam,\eta)$ tree packing $\mathcal{T}$ (see Def. \ref{def:weakTP}), and that the root node $v_r$ of the packing holds a broadcast message $M$ of size $O(k\log{n})$ bits. Our Safe Broadcast procedure allows every node $v \in V$ to output the message $M$ within $O((\kDiam+\log{n})\eta)$ rounds.

We represent the broadcast message $M$ as a list $[\alpha_1,\dots,\alpha_\ell] \in [q]^{\ell}$ where $q=2^p$ for some positive integer $p$ and $q \geq k$, and such that $\ell$ is an integer satisfying $k \geq c''\ell$ for a sufficiently large $c'' > 0$. Let $C$ be a $[\ell,k,\delta_C]_q$-code for $\delta_C=(k-\ell+1)/k$, known as the Reed Solomon Code (see \Cref{thm:rs}).
The root encodes the broadcast message $[\alpha_1,\dots,\alpha_\ell]$ into a codeword $C([\alpha_1,\dots,\alpha_\ell]) =[\alpha'_1,\dots,\alpha'_k]$. Next, the algorithm runs $k$ RS-compiled $\kDiam$-hop broadcast algorithms, in parallel. That is, for every $T_j \in \mathcal{T}$, let $\Pi(T_j)$ be a $\kDiam$-hop broadcast algorithm in which the message $\alpha'_j$ starting from the root $v_r$, propagates over $T_j$ for $\kDiam$ hops (hence taking $O(\kDiam+\log q)$ rounds). Let $\Pi_{RS}(T_j)$ be the RS-compilation of that algorithm, denoted hereafter as RS-broadcast algorithm. All $k$ RS-broadcast algorithms are implemented in parallel by using the scheduling scheme of Lemma \ref{lem:scheduler_security}.

Let $\alpha'_j(u)$ be the value that a node $u$ receives from $T_j$ (or $\alpha'_j(u) = 0$ if it received no value) at the end of this execution. To determine the broadcast message $[\alpha_1,\dots,\alpha_\ell]$, each $u$  calculates first the closest codeword $\alpha(u)$ to $\alpha'(u) = [\alpha'_{1}(u),\dots,\alpha'_{k}(u)]$. 
Its output is then given by $\widetilde{\alpha}(u)=[\widetilde{\alpha}_1(u),\dots,\widetilde{\alpha}_\ell(u)]=C^{-1}(\alpha(u))$. 
%
%
This completes the description of this procedure. 

\begin{lemma}\label{lem:correct-BroadcastECC}
Consider the execution of Alg. $\ECCSafeBroadcast$ in the presence of $f$-mobile byzantine adversary with a given broadcast message $[\alpha_1,\dots,\alpha_\ell] \in [q]^\ell$; and a distributed knowledge of a \emph{weak} $(k,\kDiam,\eta)$ tree packing for $k\geq \max\{c''\cdot \ell, c^* \eta f\}$ for large constants $c'', c^*$. Then, $\widetilde{\alpha}(u)=[\alpha_1,\dots,\alpha_\ell]$ for every node $u$. In addition, the round complexity is $O((\kDiam+\log{q})\eta)$  $1$-$\congest$ rounds. 
\end{lemma}
\def\APPENDBROADECC{
\begin{proof}
Let $\mathcal{T}'\subseteq \mathcal{T}$ be the collection of $\kDiam$-spanning trees rooted at a common root $v_r$. By the definition of weak tree-packing, we have that $|\mathcal{T}'|\geq 0.9k$. 

Broadcasting a message of size $O(q)$ for $\kDiam$-hops takes $O(\kDiam+\log{q})$ rounds using $1$-bit messages. Therefore, a single application of the RS-broadcast algorithm $\Pi_{RS}(T_j)$ over some subgraph $T_j \in \mathcal{T}$ takes $O(\kDiam+\log{q})$ rounds. By \Cref{lem:scheduler_security}, all $k$ algorithms $\{\Pi_{RS}(T_j)\}_{j=1}^k$ can be performed in $O((\kDiam+\log{q})\eta)$ rounds, such that all but $c \cdot f\cdot \eta$ end correctly, for some constant $c$.  Therefore, at least $|\mathcal{T}'|-c \cdot f\cdot \eta \geq (1-1/c')k$ algorithms are valid, by taking $k=c^* \eta f$ for a sufficiently large $c^*$. This implies that for at least $(1-1/c')k$ algorithms we have $\alpha'_{j}(u) = \alpha'_{j}$. (For that to happen, we have $T_j \in \mathcal{T}'$ and the RS-compiled algorithm $\Pi_{RS}(T_j)$ is valid. Or in other words,
	\[\frac{\Hamm(\widetilde{\alpha}(u),C(\alpha_1,\dots,\alpha_\ell))}{k} \leq \frac{1}{c'}.\]

\noindent On the other hand, since $k \geq c''\ell$, the relative distance of the code $C$ can be bounded by:
$$\delta_C =\frac{k-\ell+1}{k} \geq 1-\frac{\ell}{k} \geq 1-1/c''.$$  
Therefore, for any given point $x \in \mathbb{F}_{q}^k$, there is at most one codeword of relative distance less than $\delta_C/2$. As for sufficiently large $c',c''$, one obtains $(1-\frac{1}{c''})/2 \geq \frac{1}{c'}$, we get that: \[\frac{\Hamm(\widetilde{\alpha}(u),C(\alpha_1,\dots,\alpha_\ell))}{k} < \frac{\delta_C}{2}.\] 
We conclude that the decoding of every node is correct, i.e., that $\widetilde{\alpha}(u)=[\alpha_1,\dots,\alpha_\ell]$ for every $u$.
\end{proof}
}\APPENDBROADECC

\subsubsection{$f$-Resilient Compiler}

Let $\cA$ be a given $r$-round distributed algorithm for $G$, and for every $i \in \{1,\ldots, r\}$ and an edge $(u,v)\in E$, let $m_{i}(u,v)$ be the message sent from $u$ to $v$ in round $i$ in Alg. $\cA$. Throughout, we assume w.l.o.g. that the last $O(\log n)$ bits of each message $m_{i}(u,v)$ specifies the ID of the message defined by the sender identifier ($\ID(u)$) appended by the receiver identifier ($\ID(v)$), i.e., $\ID(m_i(u,v))=\ID(u)\circ \ID(v)$.\footnote{This requires the KT1 assumption, i.e. that nodes know the identifiers of their neighbors.} Recall that every node $u \in V$ holds private $\poly(n)$ random coins that are unknown to the adversary. Given this randomness, the simulation of $\cA$ is deterministic. 

Our goal is to simulate $\cA$ in a round-by-round manner, such that the resulting algorithm $\cA'$ provides the same output distribution as that of $\cA$ in the presence of an $f$-mobile byzantine adversary. Every round $i$ of Alg. $\cA$ is simulated by a phase of $O(k+\kDiam)$ rounds. At the end of the phase, each node $v$ can determine the set of incoming messages $\{m_i(u,v)\}_{u \in N(v)}$ as in the fault-free simulation of Alg. $\cA$. For the remainder of the section, we fix round $i \geq 1$ and assume that each node already holds the correct list of received messages $\{m_j(u,v)\}_{u \in N(v)}$ for every $j \leq i-1$. Hence, all nodes can simulate their state in a fault-free simulation of the first $i-1$ rounds of Alg. $\cA$.

\subsubsection*{Simulation of the $i^{th}$ Round}  We are now ready to describe the reliable simulation of round $i$. On a high level, in the first round, we let all nodes exchange their $i^{th}$-round messages as in (the fault-free) Alg. $\cA$. Then, the bulk of the simulation is devoted for a message correction procedure, which allows all nodes to hold the correct received messages as in the fault-free setting, despite the corruptions of the mobile adversary.

\noindent \textbf{Step (1): Single Round $i^{th}$ Message Exchange.}
The phase starts by letting each node $u$ sending the $i^{th}$-round messages of Alg. $\cA$, namely, $\{m_{i}(u,v)\}_{v \in N(u)}$. For every directed edge $(u,v)\in E$, let $m'_i(u,v)$ be the message \emph{received} by $v$ from $u$ in this round. Since the adversary might corrupt at most $f$ edges in this round, it might be the case where $m'_{i}(u,v)\neq m_i(u,v)$ for at most $2f$ ordered pairs. Note that since the identifiers of the messages are known to every node, we can assume that the last bits of the messages $m'_{i}(u,v), m_i(u,v)$ specify the message-ID\footnote{I.e., the identifier of the message consists of the sender-ID (namely, $\ID(u)$) concatenated by the receiver-ID (namely, $\ID(v)$).} given by $\ID(u)\circ \ID(v)$. We denote the event where $m'_{i}(u,v)\neq m_i(u,v)$ as a \emph{mismatch}.

\smallskip 
\noindent \textbf{Step (2): Upcast of Mismatches using $\ell_0$-Samplers.}
Our algorithm works in iterations $j \in \{1,\dots,z\}$ for $z=\Theta(\log f)$. At the beginning of every iteration $j$, each node $v$ holds for each neighbor $u$, a variable $m'_{i,j-1}(u,v)$, which represents its estimate\footnote{By estimate, we mean that each node $v$ maintains a variable $m'_{i,j-1}(u,v)$ containing some message that is ideally the message $m_i(u,v)$, but could be an incorrect value. We show in the analysis that as the iterations pass, the number of estimates which are not equal to the correct value decrease exponentially.} for its received message from $u$ in round $i$ of Alg. $\cA$. Initially, $m'_{i,0}(u,v) = m'_i(u,v)$.  

Next, every node $v$ locally defines two multi-sets corresponding to its outgoing messages and and its $j$-estimated incoming messages: 
$\mathrm{Out}_i(v) = \{m_i(v,u_1),\ldots, m_i(v,u_{\deg(v)})\}$ and $\mathrm{In}_{i,j}(v) = \{m'_{i,j}(u_1,v),\ldots,m'_{i,j}(u_{\deg(v)},v)\}.$ 
Let $S_{i,j}(v)$ be a multi-set of tuples consisting of every element in $\mathrm{Out}_i(v)$ with frequency $1$, and every element in $\mathrm{In}_{i,j}(v)$ with frequency $-1$. In other words, $S_{i,j}(v) = \{(m,1)\}_{m \in \mathrm{Out}_i(v)} \cup \{(m,-1)\}_{m \in \mathrm{In}_i(v)}$. Let $S_{i,j} = S_{i,j}(v_1) \cup \dots \cup S_{i,j}(v_n)$ be the multi-sets formed by taking a union over all $n$ multi-sets. 

Each subgraph $T \in \mathcal{T}$ runs an RS-compiled sub-procedure, $L0_{RS}(T,S_{i,j})$, which is defined by applying the RS-compiler of \Cref{thm:less_than_one_error} for the following (fault-free) $L0(T,S_{i,j})$ procedure, which is well defined when $T$ is a spanning tree. In the case where $T$ is an arbitrary subgraph, the execution of $L0(T,S_{i,j})$ which is restricted to $\widetilde{O}(\kDiam)$ rounds, will result in an arbitrary outcome. 

\smallskip
\noindent \textbf{Procedure $L0(T,S_{i,j})$.} The node $v_r$ first broadcasts $\widetilde{O}(1)$ random bits $R_{i,j}(T) = R_{i,j,1}(T),\dots,R_{i,j,t}(T)$ over the edges of $T$, where $t=\Theta(\log{n})$.\footnote{In the case where $T$ is \emph{not} a spanning tree, $v_r$ might have no neighbors in $T$. Nevertheless, the correctness will be based on the $0.9k$ spanning trees in $\mathcal{T}$.} Then, each node $v$ initializes $t$ mutually independent $\ell_0$-sampler sketches on the multi-sets $S_{i,j}(v)$ with randomness $R_{i,j,h}(T)$ for $h=1,\dots,t$ using Theorem~\ref{thm:l_0_sampler}. Let $[\tau_1(v),\ldots, \tau_t(v)]$ be the $\ell_0$-sampling sketches obtained for $S_{i,j}$ with the randomness $R_{i,j}(T)$. The combined sketches $L(n,S_{i,j},R_{i,j,1}),\dots,L(n,S_{i,j},R_{i,j,t})$ are then computed in a bottom-up manner on $T$ from the leaves to the root $v_r$, in the following manner: first, each leaf node $v$ sends its sketches, $\tau_{u,1} = L(n,S_{i,j}(v),R_{i,j,1}),\dots,\tau_{u,t} =  L(n,S_{i,j}(v),R_{i,j,t})$, to its parent in $T$ in $\widetilde{O}(1)$ rounds. Any non-leaf node $v$ waits until it receives $t$ sketches $\tau_{u,1},\dots,\tau_{u,t}$ from each child $u \in \mathrm{Children}(v,T)$. For each $h \in \{1,\ldots t\}$, it merges all the sketches $\{\tau_{u,h}\}_{u \in \mathrm{Children}(v,T)}$ together with $L(n,S_{i,j}(v),R_{i,j,h})$ (using the merge operation described in \Cref{thm:l_0_sampler}), and propagates the resulting $t$ sketches to its parent. Finally, the root $v_r$ obtains $t$ sketches from each of its children, computes the combined $t$ sketches $L(n,S_{i,j},R_{i,j,1}),\dots,L(n,S_{i,j},R_{i,j,t})$, and (locally) applies the $\Query$ operation of \Cref{thm:l_0_sampler} on each combined sketch to sample a list of values $A_{i,j}(T) = [a_1(T),\dots,a_{t}(T)]$, where $a_h(T)=\Query(L(n,S_{i,j},R_{i,j,h}))$ for $h \in \{1,\dots,t\}$.

The round complexity of Procedure $L0(T,S_{i,j})$ is restricted to $\widetilde{O}(\kDiam)$ rounds. 
This concludes the description of $L0(T,S_{i,j})$ and hence also its RS-compilation $L0_{RS}(T,S_{i,j})$. Our algorithm implements the collection of the $k$ RS-compiled algorithms $\{L0_{RS}(T,S_{i,j})\}_{T \in \mathcal{T}}$, in \emph{parallel}, by employing the RS-scheduler of Lemma \ref{lem:scheduler_security}. 

\smallskip
\noindent \textbf{Detecting Dominating Mismatches.} 
A \emph{positive} element $a \in A(T)$ is denoted as an \emph{observed-mismatch}. For every observed-mismatch $a \in \bigcup_{T \in \mathcal{T}} A_{i,j}(T)$, denote its \emph{support} in iteration $j$, by $\supp_{i,j}(a) = |\{(\ell,T) \stt a = a_\ell(T), T \in \mathcal{T}, \ell \in \{1,\ldots, t\}\}|$. The root $v_r$ then selects a sub-list $\DM_{i,j}$ of \emph{dominating} observed mismatches, i.e., mismatches that have a sufficiently large \emph{support} in $\bigcup_{T} A(T)$, based on the given threshold $\Delta_j$. Specifically, for a sufficiently large constant $c''$, let:
\begin{equation}\label{eq:Mij}
\Delta_j=0.2c''2^j\eta \cdot t \mbox{~and~} \DM_{i,j} = \{a \in \bigcup_T A_{i,j}(T) \stt a>0, \supp_{i,j}(a) \geq \Delta_j \}~.
\end{equation}
The remainder of the $i^{th}$ phase is devoted to (resiliently) broadcasting the list $\DM_{i,j}$. 

\smallskip 
\noindent \textbf{Step (3): Downcast of Dominating Mismatches.} 
To broadcast $\DM_{i,j}$, Alg. $\ECCSafeBroadcast$ is applied with parameters $q=2^p$ for $p=\lceil \max(\log{k},\log^5{n})\rceil$ and $\ell=|\DM_{i,j}|$. Upon receiving $\DM_{i,j}$, each node $v$ computes its $j$-estimate $m'_{i,j}(u,v)$ as follows. If there exists a message $m \in \DM_{i,j}$ with $\ID(m)=\ID(u) \circ \ID(v)$, then $v$ updates its estimate for the received message by setting $m'_{i,j}(u,v)=m$. Otherwise, the estimate is unchanged and $m'_{i,j}(u,v)\gets m'_{i,j-1}(u,v)$. This completes the description of the $j^{th}$ iteration. Within $z=O(\log f)$ iterations, each node $v$ sets $\widetilde{m}_i(u,v)=m'_{i,z}(u,v)$ for every neighbor $u$. In the analysis we show that, w.h.p., $\widetilde{m}_i(u,v)=m_{i}(u,v)$ for all $(u,v)\in E$. 

\def\APPENDCODEBYZ{
\begin{mdframed}[hidealllines=false,backgroundcolor=gray!00]
\center \textbf{Algorithm $\ImprovedMobileByznatineSim$ ($i^{th}$ Phase):}

\raggedright\textbf{Input:} Weak $(k,\kDiam,\eta)$ tree packing $\mathcal{T}=\{T_1\dots,T_k\}$.\\
\raggedright \textbf{Output:} Each node $v$ outputs $\{\widetilde{m}_i(u,v)\}_{u \in N(v)}$, estimation for its received messages in round $i$ of Alg. $\cA$.
	\begin{enumerate}
		\item Exchange $m_i(u,v)$ over each edge $(u,v) \in E$.
		\item Let $\{m'_{i,0}(u,v)\}_{(u,v)\in E}$ be the received messages at the receiver endpoints $\{v\}$. 
\item For $j = 1,\dots,z=O(\log{f})$ do:
\begin{itemize}
\item Employ protocol $L0_{RS}(T,S_{i,j})$ over each $T \in \mathcal{T}$, in parallel, using \Cref{lem:scheduler_security}.

\item Set $\DM_{i,j}$ as in Eq. (\ref{eq:Mij}).
\item Broadcast $\DM_{i,j}$ by applying Alg. $\ECCSafeBroadcast$.
\item For every $v \in V$ and $u \in N(v)$: 
\begin{itemize}
\item If $\exists m \in \DM_{i,j}$ with $\ID(m)=\ID(u)\circ \ID(v)$: $m'_{i,j}(u,v) \gets m$.
\item Otherwise, $m'_{i,j}(u,v) \gets m'_{i,j-1}(u,v)$. 
\end{itemize}
\end{itemize}

\item For every $v \in V$ and $u \in N(v)$: Set $\widetilde{m}_{i}(u,v)=m'_{i,z}(u,v)$.

\end{enumerate}
\end{mdframed}
}

\APPENDCODEBYZ

\noindent \textbf{Analysis.} We focus on the correctness of the $i^{th}$ round and omit the index $i$ from the indices of the variables when the context is apparent. Let $\mathcal{T}'\subseteq \mathcal{T}$ be the collection of spanning-trees of depth $O(\kDiam)$, rooted $v_r$. We assume that $k \geq c' \cdot c'' \eta f$, where $c''$ is the constant of \Cref{lem:scheduler_security} and $c' > c''$. 

A tree $T_q$ is denoted as $j$-\emph{good} if (a) $T_q \in \mathcal{T}'$ and (b) $L0_{RS}(T_q,S_{i,j})$ ended correctly when using the scheduler of \Cref{lem:scheduler_security}. Let $\mathcal{T}^{(j)}_{\good}\subseteq \mathcal{T}'$ be the collection of good trees. By \Cref{lem:scheduler_security}, it holds:

\begin{equation}\label{eq:good-trees-num}
|\mathcal{T}^{j}_{\good}| \geq \left(0.9-\frac{c''f\eta}{k}\right)|\mathcal{T}| \geq \left(0.9-\frac{1}{c'} \right)|\mathcal{T}|~.
\end{equation}
Finally, we say that a sent message $m_i(u,v)$ a $j$-\emph{mismatch} if $m_i(u,v) \neq m'_{i,j}(u,v)$. Let $B_{j}$ denote the number of $j$-mismatches. The following lemma follows by the Chernoff bound:

\begin{lemma}
\label{lem:faulty_message_chernoff}
For every $j \in [z]$, if $B_{j-1} \leq 2f/2^{j-1}$, then for any $(j-1)$-mismatch $m$, $\supp_{i,j}(m) \geq \Delta_j$, w.h.p.
\end{lemma}	
\def\APPENDFMCHERNOFF{
\begin{proof}
Let $m$ be a $(j-1)$-mismatch. Since there are at most $2f/2^{j-1}$ mismatches, there are at most $4f/2^{j-1}$ non-zero entries in $S_{i,j}$. By \Cref{thm:l_0_sampler}, in each sketch of a $j$-good tree $T$ the root detects a given $(j-1)$-mismatch with probability at least $2^{j-1}/4f - \epsilon(n)$ for $\epsilon(n) = O(1/\poly(n))$. Therefore, for a sufficiently large $c'$, by Equations (\ref{eq:Mij},\ref{eq:good-trees-num}) and the fact that we use $t=\Theta(\log n)$ independent $\ell_0$ sketches, we get:

\[\mathbb{E}\left(\supp_{i,j}(m) \right) \geq \left(0.9-\frac{1}{c'}\right)\left(\frac{2^j}{4f}-\epsilon(n)\right)k\cdot t \geq 0.8\left(\frac{2^{j-1}}{4f}\right)kt = 0.2c'' \cdot c'2^j\eta\cdot t = c' \Delta_j~.\]
	
Moreover, the observed mismatches sampled by the good trees are mutually independent of each other. Given $(j-1)$-mismatch $m'$, the probability that it is sampled by less than $\Delta_j$ many $j$-good trees can be bounded by a Chernoff bound (Lemma~\ref{lem:chernoff}), as follows:
	
\begin{eqnarray}
\Pr(\supp_{i,j}(m) \leq \Delta_j) &\leq& \Pr\left(\supp_{i,j}(m) \leq \mathbb{E}\left(\supp_{i,j}(m)\right)/c'\right) e^{-(1-1/c')^2 c'\Delta_j/2} \nonumber
\\&=& e^{-(1-1/c')^2 \frac{0.2c'' \cdot c' 2^j\eta t}{2}} \leq \frac{1}{\poly(n)}, \nonumber
\end{eqnarray}

where the last inequality holds for a sufficiently large $c'$ (which can be chosen sufficiently large compared to $c''$).
\end{proof}
}\APPENDFMCHERNOFF

\begin{lemma}
\label{lem:faults_decrease}
For every $j \in [z]$, $B_{j} \leq 2f/2^{j}$ w.h.p. 
\end{lemma}
\begin{proof}
We prove it by induction on $j$. For $j = 0$, the number of $j$-mismatches at the start of the protocol is indeed at most $2f$ since the adversary corrupts at most $2f$ messages in the first round of Phase $i$.
Assume that the claim holds up to $j-1$ and consider $j\geq 1$. 
We say that an observed-mismatch $a$ has $j$-\emph{high support} if $\supp_{i,j}(a) \geq \Delta_j$, and has $j$-\emph{low support} otherwise. An observed-mismatch $a$ is \emph{competing} with $m_i(u,v)$ if $\ID(a) = \ID(u) \circ \ID(v)$ but $a \neq m_i(u,v)$. Note that assuming all sketches on $j$-good trees are successful (which indeed holds w.h.p.), then all \emph{competing mismatches} with any message $m_i(u,v)$ must be sampled by $j$-\emph{bad} trees, since they are not real $(j-1)$ mismatches. 
A necessary condition for $m_i(u,v)$ to be a $j$-mismatch, is either that (a) there is an observed-mismatch with $j$-high support that is competing with $m_i(u,v)$, or (b) it is a $(j-1)$-mismatch and has $j$-low support. By \Cref{lem:faulty_message_chernoff}, all $(j-1)$-mismatches have $j$-high support w.h.p., and therefore, there are no $j$-mismatches due to condition (b).

Since the number of $j$-bad trees is at most $c''f\eta+0.1k$ (see Eq. (\ref{eq:good-trees-num})), at most $\frac{(c''\eta f+0.1k)t}{\Delta_j} \leq 2f/2^{j+1}$ competing observed-mismatches have a $j$-high support. Therefore, $B_{j} \leq 2f/2^{j}$ w.h.p.
\end{proof}

\noindent The proof of \Cref{thm:res-mobile-improved} follows by noting that in the last iteration $z = {O(\log{f})}$, it holds that $B_z = 0$, w.h.p. In particular, at the last iteration $z$, each estimated message is the correct message, i.e. $\widetilde{m}_{i,z}(u,v) = m_i(u,v)$. 
%

%% file: tools.tex
\subsection{Tools}\label{sec:tools}
Our simulation is based on the following three main tools. 
 
\smallskip
\noindent\textbf{Tool 1: Tree-Packing of $(k,\kDiam)$ Connected Graphs.} We need the following definition that has been introduced by \cite{ChuzhoyPT20}.  A graph $G=(V,E)$ is $(k,\kDiam)$-\emph{connected} iff for every pair $u,v \in V$, there are $k$ edge-disjoint paths connecting $u$ and $v$ such that the length of each path in bounded by $\kDiam$. Observe that $\kDiam$ might be considerably larger than the graph diameter. 

A tree-packing is a decomposition of the graph edges into near edge-disjoint spanning trees. In the distributed setting, an important parameter of a tree-packing is the depth of the trees, which determines the round complexity of our procedures.  

\begin{definition}[Low-Diameter Tree Packing]
For a given graph $G=(V,E)$, an $(k,\kDiam,\eta)$ tree packing $\mathcal{T}=\{T_1,\ldots, T_k\}$ consists of a collection of $k$ spanning trees in $G$ such that (i) the diameter of each tree $T_i$ is at most $\kDiam$ and (ii) each $G$-edge appears on at most $\eta$ many trees in $\mathcal{T}$ (i.e., the load of $\mathcal{T}$ is at most $\eta$).
When $\eta=O(\log n)$, we may omit it and simply write $(k,\kDiam)$ tree packing. 
\end{definition}

Chuzhoy et al. \cite{ChuzhoyPT20} presented an efficient randomized centralized algorithm for computing a $(k,\kDiam \cdot \log n)$ tree-packing for any given $(k, \kDiam)$-connected graph $G$. 

\begin{theorem}[Theorem 4 \cite{ChuzhoyPT20}]
\label{thm:packing-KD}
There is an efficient randomized centralized algorithm that given $(k,\kDiam)$-connected $n$-node graph $G$ computes a collection $\mathcal{T}=\{T_1,\ldots,T_{k}\}$ of $k$ spanning trees of $G$, such that, w.h.p., for each $1\le \ell\le k$, the tree $T_{\ell}\subset G$ has diameter $O(\kDiam\log n)$ and the load of the trees is $O(\log n)$. 
\end{theorem}
It is well-known \cite{N61} that any $k$-edge connected graph, contains a $(\lfloor k/2 \rfloor ,n)$-tree packing. Moreover, Karger sampling \cite{Karger94} obtains a $(\widetilde{O}(k),n)$ tree packing (i.e., by randomly partitioning the edges into $\widetilde{O}(k)$ graphs). The main advantage of Thm. \ref{thm:packing-KD} is in providing low-diameter trees.  

In \Cref{sec:app_tree_packing}, we show a distributed algorithm for computing the tree packing of Theorem \ref{thm:packing-KD} in $\widetilde{O}(k\cdot \kDiam^2)$ rounds, in the fault-free setting, using the techniques of \cite{Ghaffari15}\footnote{This claim is implicitly shown in \cite{Ghaffari15}, but since their packing is optimized under different parameters, we provide the complete analysis for completeness.}. The algorithm presented in \Cref{sec:app_tree_packing} is useful when handling general graphs. For the latter, the trees are computed as part of a trusted preprocessing step. 

The main applications of this paper, namely, compilers for expander graphs and for the \congc\ model, do \emph{not} require such trusted pre-processing steps and rather compute the tree packing itself in the byzantine setting\footnote{This is trivial in the \congc\ model, but more challenging for expander graphs in the \congest\ model.}.  For the sake of the applications of Sec. \ref{sec:app-simple}, we consider a weaker notion of tree-packing, which allows for at most $0.1$ fraction of the subgraphs to be of arbitrary structure. 

\begin{definition}\label{def:weakTP}[Weak Tree-Packing]
A collection of $k$ $G$-subgraphs $\mathcal{T}=\{T_1,\ldots, T_k\}$ is a \emph{weak} $(k,\kDiam,\eta)$ tree packing if (i) at least $0.9k$ of the subgraphs in $\mathcal{T}$ are \emph{spanning trees} of diameter at most $\kDiam$, rooted at a common root $v_r$; and (ii) the load of 
$\mathcal{T}$ is at most $\eta$. 
\end{definition}


\smallskip
\noindent\textbf{Tool 2: Compilers against Bounded Adversarial Rate.} We use known compilers that can successfully simulate a given fault-free algorithm in the $1$-\congest\ model, as long as the adversary corrupts a bounded \emph{fraction} of the total \emph{communication}. The latter is referred to as \emph{communication-error-rate}. Specifically, we use the compilers of Rajagopalan and Schulman \cite{RajagopalanS94}, that we denote hereafter by \emph{RS-compilers}. While in their original work \cite{RajagopalanS94} the RS-compilers proved useful against stochastic noisy channels, Hoza and Schulman \cite{HozaS16} later observed that it is in-fact resilient against \emph{adversarial} corruptions.

\begin{theorem}[Slight restatement of Proposition 1 of \cite{HozaS16} (see also \cite{RajagopalanS94})]
\label{thm:less_than_one_error}

Given a network on $m$ edges on which a $1$-\congest\ protocol $\Pi$ runs in $r$ rounds (in the fault-free setting), there is an $r'$-round $1$-\congest\  protocol $\Pi_{\RS}$, where $r' \in [r, \RStime \cdot r]$ for some constant $\RStime \geq 1$, which simulates $\Pi$ by sending $1$-bit messages over \emph{all} graph edges in \emph{every} round with the following guarantee: if the adversary corrupts at most $1/(\RSthresh\cdot m)$-fraction of the total communication of the protocol for a constant $\RSthresh > 1$, the output distributions of $\Pi$ and $\Pi_{\RS}$ are the same. 
\end{theorem}

Our compilation scheme is based on applying the RS compiler on a collection of $k$ algorithms, each running over a distinct tree in a given $(k,\kDiam)$ tree-packing $\mathcal{T}$. To save on the round complexity (i.e., avoiding a multiplicative dependency in $k$), we provide the following scheduling lemma which allows us to run these $k$ algorithms in parallel, and enjoy the fact that the trees have a bounded level of overlap. The guarantee is that the vast majority of these algorithms end correctly in the presence of an $f$-mobile byzantine adversary (or, more strongly, in the presence of round-error rate $f$). 
Throughout, we assume that all nodes hold an upper bound on the runtime of the individual RS-algorithms, as well as an upper bound on the overlap of the trees. 

\begin{lemma}[$f$-Mobile (Almost) Resilient Scheduling of RS-Compiled Algorithms]
\label{lem:scheduler_security}
Let $\mathcal{G}=\{G_1,\ldots, G_s\}$ be an ordered $G$-subgraph family of load $\eta$. Let $\cA_1,\ldots, \cA_s$ be a collection of RS-algorithms, such that each $\cA_j$ sends messages over all $G_{j}$-edges in each of its $r_j$ rounds, where $r_j \in [r, \RStime \cdot r]$ for some fixed $r$ known to all nodes. Then, assuming a round-error rate of $f$, all $s$ algorithms can be run in parallel within $\leq\RStime \cdot r\eta$ rounds, such that the following holds: all but at most $\RStime \cdot \RSthresh \cdot f\cdot \eta$ algorithms end correctly (where $\RSthresh,\RStime$ are the constants from Thm. \ref{thm:less_than_one_error}). 
\end{lemma}
\def\APPENDSCHEDRS{
\begin{proof}
Scheduler $\RSScheduler$ runs in phases of $\eta$ rounds. In phase $i$, it exchanges all messages of the $i^{th}$ round of all $k$ algorithms $\{\cA_j\}_{j | r_j \geq i}$ as follows. For every directed edge $e=(u,v)$, let $j_1(e),\ldots, j_{q_e}(e)$ be the set of indices such that $e \in G_{j}$ and $r_j \geq i$, for every $j \in \{j_1(e),\ldots, j_{q_e}(e)\}$. By definition we have that $q_e \leq \eta$ for every $e$. Then in the $j^{th}$ round of phase $i$, $u$ sends $v$ the message $m_{j,i}(u,v)$, where $m_{j,i}(u,v)$ is the message that $u$ sends to $v$ in round $i$ of Alg. $\cA_j$. (Since $\cA_j$ sends messages over all edges in $G_j$ in every round, this message is well-defined). This completes the description of the scheduler.

We say that algorithm $\mathcal{A}_j$ ends \emph{correctly} if the communication-error-rate among the messages of $\mathcal{A}_j$, in the scheduled algorithm $\RSScheduler$, is below the RS-threshold, namely, $1/(\RSthresh|E(G_j)|)$. By \Cref{thm:less_than_one_error}, we get that in such a case, the RS-compilation is indeed successful.  We next turn to show that all but $\RStime \cdot \RSthresh \cdot f\cdot \eta$ algorithms end correctly.

The total communication of each algorithm $\cA_j$ is $r_j|E(G_j)|$. Therefore, if the total number of corrupted $\cA_j$-messages in Alg. $\RSScheduler$ is below $r/\RSthresh \leq r_j|E(G_j)|/(\RSthresh|E(G_j)|)$, then $\cA_j$ is correct. 
As the round complexity of $\RSScheduler$ is at most $\RStime \cdot r\cdot\eta$, the adversary can corrupt at most $f \cdot \RStime \cdot r \cdot\eta$ messages. By an averaging argument, it holds that for all but $\RStime \cdot \RSthresh \cdot f\cdot \eta$ many algorithms have at most $r/\RSthresh$ corrupted messages that are associated with it, hence experiencing a communication-error-rate below the RS-threshold. We conclude that all but $\RSthresh \cdot f\cdot \eta$ many algorithms end correctly.
\end{proof}
}\APPENDSCHEDRS

\smallskip
\noindent\textbf{Tool 3: Sketches for $\ell_0$-Sampling.}  
Sketching is a powerful compression tool which allows one to store aggregate information of a data stream using small space. To introduce the concept of sketching, consider the following setting. For some size parameter $n \in \mathbb{N}$, let $U = \{1,\dots,\poly(n)\}$ be a universe of elements, and let $\sigma$ be a multi-set of $O(\poly(n))$ tuples $\{(e_i,f_i)\}_{i \in [O(\poly(n))]}$, where $e_i \in U$ is an element in the universe and $f_i \in \{-\poly(n),\dots,\poly(n)\}$ is an integer referred to as the change in frequency. The \emph{frequency} of an element $e \in U$ in the multi-set $\sigma$ is defined as the sum of its changes of frequency, i.e. $f(e) = \sum_{i:e_i = e} f_i$. Denote by $N(\sigma)$ the set of elements in $U$ with non-zero frequency in $\sigma$, i.e. $N(\sigma) = \{e \in U \stt f(e) \neq 0\}$. Informally, the act of $\ell_0$-sampling a multi-set $\sigma$ is choosing a uniform element from $N(\sigma)$, or in other words, choosing a uniformly random non-zero frequency element from $\sigma$.

In the context of this paper, an $\ell_0$-sampling sketch $\tau$ of $\sigma$ is a randomized string of size $\polylog(n)$, which can be constructed using a multi-set $\sigma$ and randomness $R$ of $\polylog(n)$ bits, and on which two operations are defined: $\Query$ and $\Merge$. A  $\Query$ operation receives as input a sketch $\tau$ of a multi-set $\sigma$, and outputs a random element in $N(\sigma)$ w.h.p., where each element in $N(\sigma)$ is chosen with probability $\frac{1}{N(\sigma)} \pm \frac{1}{\poly(n)}$ (the randomness is taken entirely over the selection of $R$). The \emph{merge} operation receives as input two sketches $\tau_1,\tau_2$ constructed using the same randomness $R$ on the multi-sets $\sigma_1$,$\sigma_2$ respectively, and its output is a sketch $\sigma'$ (of $\polylog{n}$ bits) which is equal to a sketch obtained using randomness $R$ and the concatenated multi-set $\sigma_1 \cup \sigma_2$.

\begin{theorem}[$\ell_0$-sampler, rephrasing of \cite{CF14} Corollary 1]
	\label{thm:l_0_sampler}
	There exists an algorithm $L$, which given a size parameter $n \in 	\mathbb{N}$, a multi-set $\sigma$ of size $\poly(n)$ and $R$, a string of $O(\log^4 n)$ random bits, outputs a randomized string $L(n,\sigma,R)$ called a $\ell_0$-sketch, which is encoded using $O(\log^4{n})$-bits. The following operations are defined on the sketch:
	\begin{itemize}
		\item $\Query$: a deterministic procedure in which given a (randomized) sketch $\tau = L(n,\sigma,R)$ outputs an element $e \in N(\sigma)$, where any given $e \in N(\sigma)$ is sampled with probability at least $1/N(\sigma) - 1/n^c$, and probability at most $1/N(\sigma) + 1/n^c$ for any constant $c > 1$, w.h.p. (randomness taken over $R$). 
		\item $\Merge$: a deterministic procedure in which given two sketches $\tau_1 = L(n,\sigma_1,R) ,\tau_2 = L(n,\sigma_2,R)$ created using the same randomness $R$ on multi-sets $\sigma_1,\sigma_2$ respectively, outputs the sketch $\tau' = L(n,\sigma_1 \cup \sigma_2,R)$.
	\end{itemize}

\end{theorem}

In our context, the universe $U$ is the set $U = \{0,1\}^{O(\log{n})}$, i.e. the collection of all possible messages of size $O(\log{n})$. A multi-set $\tau$ contains each string with either zero, positive, or negative frequency. Intuitively, use the sketches in to find for some fixed round $i$ of Alg. $\cA$ the (original) messages in round $i$ that got corrupted by the adversary, by computing a sketch in which our non-zero frequency elements is the set of faulty messages and their corrections. To obtain this, we make sure that each sent message of round $i$ of Alg. $\cA$ is added into the set with positive frequency, and each received message in round $i$ with negative frequency. Messages that were sent correctly across an edge are then cancelled out. On the other hand messages that got corrupted do not. We note that the $\Merge$ operation does not modify the universe set, or the encoding size of the sketch (which is always of size $\widetilde{O}(1)$ bits).





%% file: byzantine-applications.tex
\subsection{Applications}\label{sec:app-simple}

\noindent \textbf{Compilers for $(k,\kDiam)$-connected Graphs.} By Thm. \ref{thm:packing-KD} for every $(k,\kDiam)$-connected $n$-node graph $G$ one can compute a $(k, O(\kDiam \cdot \log n))$ tree-packing with load $\eta=O(\log n)$. These trees can be computes using $\widetilde{O}(k \cdot \kDiam)$ $\congest$ (fault-free) rounds. By Theorem \ref{thm:res-mobile-improved}, we have: 

\begin{corollary}[Mobile-Resilient Compilers for $(k,\kDiam)$-Connected Graphs]\label{thm:kd-compiler}
Given a $(k,\kDiam)$-connected $n$-vertex graph $G$ for $k=\Omega(\log n)$, then one can compile any $r$-round algorithm $\cA$ into equivalent $r'$-round algorithm $\cA'$ that is $f$-mobile resilient, where $f=\Omega(k/\log n)$ and $r'=\widetilde{O}(\kDiam)$, provided that either (i) all nodes know the graph topology (a.k.a. the supported-$\congest$ model), or (ii) there is a fault-free preprocessing step of $\widetilde{O}(k \cdot \kDiam)$ rounds. 
\end{corollary} 


\subsubsection*{Compilers for Expander Graphs, Proof of Theorem \ref{thm:expander-compiler}} By \Cref{thm:res-mobile-improved}, it is sufficient to show that for every $\phi$-expander graph with minimum degree $\Omega(f/\phi^2)$, there is an $f$-mobile-resilient algorithm for computing a weak $(k,\kDiam,\eta)$ tree packing with $k=\Omega(f/\phi^2)$, $\kDiam=O(\log n/\phi)$ and $\eta=2$. Specifically, we show:

\begin{lemma}[Weak Tree-Packing in Expander Graphs]\label{lem:weakTP-expanders}
For every $\phi$-expander $n$-node graph with minimum degree $\Omega(f/\phi^2)$, there is an $f$-mobile-resilient algorithm for computing a weak $(k,\kDiam,\eta)$ tree packing with $k=O(f/(\phi\log{n}))$, $\kDiam=O(\log n/\phi)$ and $\eta=2$. The round complexity is $O(\log n/\phi)$. 
\end{lemma}

\paragraph{The Algorithm.} In the first round, for every edge $(u,v) \in E$ with $\ID(u)>\ID(v)$, let $u$ choose a color $c(u,v)$ sampled uniformly at random in $[k]$ and send this color to $v$, who sets the color of the edge $(u,v)$ to the received value $c(v,u)$. For every $i \in [k]$, define the directed subgraph: 
$$G_i=\{(u,v) ~\mid~ \ID(u) > \ID(v), c(u,v)=i\} \cup \{(v,u) ~\mid~ \ID(u) > \ID(v), c(v,u)=i\}.$$

Let $N_i(u)$ be the outgoing neighbors of $v$ in $G_i$. Each node $u$ proceeds to do the following BFS procedure in each $G_i$: Set a variable $I_{i,0}(u) = \ID(u)$ and $\parent_i(u) = \bot$ for each color $i \in [k]$. 
For $\ell = 1,\dots,z = O(\log{n}/\phi)$ rounds, in parallel for all colors $[k]$, each node $u$ sends to each (outgoing) neighbor $v \in N_i(u)$ the value $I_{i,\ell-1}(u)$. Let $L_{i,\ell}$ be the set of messages $u$ receives from neighbors $\{v \in N(u)  \mid c(u,v) = i\}$. Each node $u$ sets $I_{i,\ell}(u) = \max_{\ID(v) \in L_{i,\ell} \cup \{\ID(u)\}}(\ID(v))$, and if this value strictly increases, then it sets $\parent_i(u)$ to be the neighbor from which this value has arrived, oriented towards the other endpoint. In the final round, each node $u$ sends in parallel for each $i$ a message to $\parent_i(u)$, to orient the edge $\parent_i(u)$ towards itself. Each node that receives an edge orientation request, locally sets the orientation of that edge towards itself.

\paragraph*{Analysis, Proof of Lemma \ref{lem:weakTP-expanders}.}
For every round $j$, let $F_j$ be the set of (undirected) edges controlled by the mobile byzantine adversary in that round. A color $i \in [k]$ is denoted as \emph{good} if the adversary did not control any of the edges of $G_i$ during the first phase (i.e., no edge in $\bigcup_{j=1}^L F_j$ appears in $G_i$, in any orientation). Otherwise, the color $i$ is defined as \emph{bad}. For every $i \in [k]$, let $G'_i \subseteq G_i$ denote the output marked (directed) subgraph obtained at the end of the algorithm. We use the following result of Wullf-Nilsen \cite{wulff2017fully}:

\begin{theorem}[\cite{wulff2017fully}, Lemma 20]\label{thm:conduct-karger}
	Given $c>0, \gamma \geq 1$ and $\alpha\leq 1$, let $G=(V,E)$ be an $n$-node multigraph with degree at least $\gamma \alpha$. Let $G'=(V,E')$ be the multigraph obtained from $G$ by sampling each edge independently with probability 
$p=\min\{1, (12c+24)\ln n/(\alpha^2 \gamma)\}~$. Then, with probability $1-O(1/n^c)$, for every cut $(S,V \setminus S)$ in $G$, it holds that:
	\begin{itemize}
		\item if $\phi_{G}(S)\geq \alpha$ then $\phi_{G'}(S)$ deviates from $\phi_{G}(S)$ by a factor of at most $4$, and
		\item if $\phi_{G}(S)< \alpha$ then $\phi_{G'}(S)<6\alpha$.
	\end{itemize}
\end{theorem}

\begin{lemma}[Section 19.1 from \cite{Spielman09}, Fact 32 \cite{HitronP21-arXiv}]
\label{lem:large_cond_low_diam}
The diameter of every $n$-node $\phi$-expander is at most $O(\log{n}/\phi)$.
\end{lemma}

\begin{lemma}\label{lem:good-expander-sampled}
For each good color $i \in [k]$, the subgraph $G_i$ is an $\Omega(\phi)$-expander with diameter at most $d\log{n}/\phi$, for some constant $d$, w.h.p.
\end{lemma}
\begin{proof}
For each good color $i \in [k]$, we have that $G_i=G[p]$ where each edge in $G$ is sampled independently with probability $p  = 1/k = O(\frac{\phi\log{n}}{f})$. Note that since the adversary does not control any of the $G_i$ edges, all edges in $G_i$ are bidirectional. 

By \Cref{thm:conduct-karger} for $\alpha = \phi$, $\gamma = O(\frac{f}{\phi^3})$ such that $p=\min\{1,\ln n/(\alpha^2 \gamma)\}$, then w.h.p. a subgraph obtained by sampling by taking each edge in $G$ with probability $p$ has conductance at least $\phi/4$. By the union bound, all graphs $G_i$ have conductance $\geq \phi/4$, w.h.p.  Therefore, by \Cref{lem:large_cond_low_diam} the graph $G_i$ has diameter at most $d \cdot\log n/\phi$ for some constant $d$.
\end{proof}

\begin{lemma}\label{lem:good-properties}
For each good color $i \in [k]$, the output directed subgraph $G'_i$ defined by $\{\parent_i(u)\}_{u \in V}$ is a directed spanning tree of depth $O(\log n/ \phi)$ oriented towards the node with the largest ID in the network, $v_r = \argmax_v \ID(v)$.
\end{lemma}
\begin{proof}
Since $i$ is good, the adversary did not control any of the edges of $G_i$ throughout the entire algorithm. Since $v_r = \argmax_v \ID(v)$, and all the message $v_r$ receives are real IDs of nodes, it never changes the value $\parent_i(v_r) = \bot$. Let $V_i(j)$ be nodes of distance $j$ from $v_r$ in $G_i$. Let $V_i(-1)=\{\bot\}$. We show by induction on $j\geq 0$ that for any $u \in V_i(j)$, it holds that $I_{i,j}(u) = \ID(v_r)$ and $\parent_i(u) \in V_i(j-1)$. Clearly, $I_{i,0}(v_r)=\ID(v_r)$ and as $\parent_i(v_r)=\bot$, the claim holds for $j=0$. Assume this is the case for $j-1$. We note that $V \setminus \left(\bigcup_{a=0}^{j-1} V_a\right)$ has not seen the value $\ID(v_r)$ before round $j$, as they are of distance at least $j$ from $v_r$ in $G_j$. By induction, all $u\in V_{j-1}$ set $I_{i,j-1}(u) = v_r$, therefore any node $u \in V_j$ receives a message from a node $v \in V_{j-1}$ with the value $\ID(v_r)$ and sets its parent to be in $V_{j-1}$. The induction holds and the claim follows.
\end{proof}

Set the duration of the first phase to $L=3d\log{n}/\phi$ where $d$ is the constant from Lemma \ref{lem:good-expander-sampled} and let $k =20fL$. 
\begin{lemma}\label{lem:stepone-correct}
The collection of subgraphs $\{ G'_i\}_{i \in [k]}$ is a \emph{weak}-$(k, \kDiam,\eta)$ Tree-Packing with $\kDiam=O(\log n/\phi)$ and $\eta=2$. 
\end{lemma}
\begin{proof}
As the adversary can control at most $f$ undirected edges in each round, it corrupts at most $2f L$ directed edges. Since each directed edge $(u,v)$ belongs to a unique subgraph $G_i$, in total there are at most $2fL$ bad colors. Since $k =20fL$, there are at least $0.9k$ good subgraphs. By Lemma \ref{lem:good-properties}, for every good color $i$, the subgraph $G'_i$ corresponds to a spanning tree of depth at most $\kDiam$ and rooted at $v_r$. Since each directed edge is committed to a single color, we get that $\eta\leq 2$ (even among the bad subgraphs).
\end{proof}

%% file: byzantine-budgeted.tex
\section{Resilience to Bounded Round-Error Corruption Rate}
\label{sec:byz_budgeted}

	 In this section, we consider a stronger adversary, which is allowed to corrupt $f$ edges per round ``on average'', allowing e.g. for specific rounds with unbounded faults. We assume that all nodes terminate exactly at some time $r'$ for some given $r'$, and the adversary may corrupt in total $r' \cdot f$ communication transmissions in the duration of the protocol. Our goal is given a protocol $\cA$, to construct a protocol $\cA'$ with the same output as $\cA$ which is resilient to this adversary.
	 
	 \begin{theorem}
	 	\label{thm:budgeted_main_intro}
	 	Given a distributed knowledge of a weak $(k,\kDiam,\eta)$ tree packing $\cT$ for graph $G$, then for any $r$-round algorithm $\cA$ over $G$, there is an equivalent $r'$-round algorithm $\cA'$ for $G$ that is resilient to round-error rate $f$, where $r'=\widetilde{O}(r \cdot \kDiam)$ and $f=\Theta(\frac{k}{\eta \log{n}})$.
	 \end{theorem}
	 
	 Similarly to \Cref{sec:comp-exactly_k_edges_fast}, we would like to simulate $\cA$ in a round-by-round manner, but in this setting we cannot correct all mismatches in every round, since the adversary may invest a ``large budget'' of faults in it.
	 
	 We use an approach referred to as the ``rewind-if-error'' paradigm for interactive coding.  For a comprehensive exposition of the rewind-if-error paradigm, and interactive coding in general, we refer to the excellent survey of \cite{G17}.
	 
	 Intuitively, our approach is as follows: If we have a round error rate of $f$, the number of rounds in which there can be more than $\alpha f$ faulty messages is at most $O(1/\alpha)$-fraction of the rounds. This means that if we apply the algorithm of Section~\ref{sec:exactly_k_edges_improved} on each simulated round such that it is resilient to $\alpha f$ total faults, in all but $O(1/\alpha)$-fraction of the rounds there are no mismatches after a correction phase in the \emph{entire network}. Rounds with more faults may cause errors; to remedy this, we combine the approach of Section~\ref{sec:exactly_k_edges_improved} with a global variant of the rewind-if-error scheme, in which if \emph{any} error in the past transcript is detected in \emph{any} part of the network, the entire network ``rewinds'' a step (by having nodes delete the last symbol of the transcript). This in turn may cause several issues, such that causing some nodes to have more rounds in their transcript than other nodes, which is why rewinding needs to be done with some care.

	\subsection{Algorithm Description}
	
	We assume that the algorithm $\cA$ for the fault-free setting is for the $1$-$\congest$ model. This can be assumed w.l.o.g. by paying a $O(\log{n})$ factor in its round complexity. Throughout the section, a symbol is a value in $\{0,1,\bot\}$.
	
	Let $\cH_k = \left\{h_i:\{0,1\}^k \rightarrow \{0,1\}^{\Theta(\log{n})} \right\}$ be the pairwise-independent hash function family of \Cref{lem:hash}.

	Our algorithm has $r' = 5r$ iterations (called \emph{global-rounds}), containing three phases: a round-initialization phase, a message-correcting phase, and a rewind-if-error phase. We design each phase to take $O(t)$ rounds for some integer $t = \widetilde{\Theta}(\kDiam)$, and to be resilient (in a sense) to a total of $O(ft)$ faults (i.e. a round-error rate of $O(f)$).
	
	Each node $v$ maintains for each edge $(u,v)$ variables $\widetilde{\pi}_i(u,v)$ (resp., $\pi_i(v,u)$) which informally contain an ``estimated'' transcript of all the messages received from $u$ at $v$ (resp., sent from $v$ to $u$) in Alg. $\cA$ up to round $|\widetilde{\pi}_i(u,v)|$ (resp., $\pi_i(v,u)$). Both these variables are initially set to be $\emptyset$, i.e., $\widetilde{\pi}_1(u,v),\pi_i(u,v)=\emptyset$. Throughout the algorithm, the following invariant is maintained:
	
\begin{invariant}
\label{inv:interactive-same-size-transcript}
For any global-round $i$ and for any node $v$, there exists a value $\gamma(v,i)$ such that for every neighbor $u \in N(v)$, $|\pi_i(v,u)| = \gamma(v,i)$ and $|\widetilde{\pi}_i(u,v)| = \gamma(v,i)$.
\end{invariant}
Clearly the invariant holds vacuously for $i=1$. We next describe the $i^{th}$ global-round. Note that it might be the case that in the $i^{th}$ global-round, node $u$ simulates some round $j_u$ of Alg. $\cA$ where $j_u \leq i$, and possibly $j_u \neq j_v$ for neighboring nodes $u,v$.

\paragraph{Round-Initialization Phase:} At the start of a global-round $i$, each node $u$ chooses for each neighbor $v$ a random string of size $O(\log{n})$, denoted by $R_i(u,v)$. Using $R_i(u,v)$, $u$ chooses a random hash $h_{R_i(u,v)} \in \cH$. For $2t$ rounds, it repeatedly sends to every neighbor $v$ the message 
	$$M_i(u,v) = (m_i(u,v),R_i(u,v),h_{R_i(u,v)}(\pi_i(u,v)),|\pi_i(u,v)|),$$ where $m_i(u,v)$ is the message $v$ sends $u$ according to $\cA$ in round $i$ assuming its incoming transcript is given by:
	$$(\widetilde{\pi}_i(v,u_1),\dots,\widetilde{\pi}_i(v,u_{\deg(v)})).$$ 
We note that this is well defined due to \Cref{inv:interactive-same-size-transcript}. If the node $u$ has terminated according to $\cA$ (i.e., based in its incoming transcript $\{\widetilde{\pi}_i(w,u)\}_{w \in N(u)}$), it sends the special symbol $m_i(u,v) = \bot$, instead. In these $2t$ rounds, each node $v$ receives from each $u$ the messages $\widetilde{M}_i(u,v,1),\dots,\widetilde{M}_i(u,v,2t)$. Following this, each $v$ locally sets for each neighbor $u$ a variable $\widetilde{M}_i(u,v) = \MAJ(\widetilde{M}_i(u,v,1),\dots,\widetilde{M}_i(u,v,2t))$, where the $\MAJ(\cdot)$ function returns the majority value or $0$, if no such exists. 
	
\paragraph*{Message-Correcting Phase:} For an integer parameter $d$, in a $d$-Message Correction Procedure, each vertex $u$ holds as input a message for each neighbor $M(u,v)$ (its \emph{outgoing values}) and a received message $\widetilde{M}(u,v)$ (its \emph{incoming values}). Define the number of \emph{mismatches} in the input instance by $|\{(u,v)\in E(G) \stt M(u,v) \neq \widetilde{M}(u,v)\}|$. Then, the message correction procedure \emph{corrects} all the mismatches under the promise that the number of mismatches in the input instance is at most $d$ (and that the error-rate is at most $f$). By correcting mismatches, we mean that each node $v$ obtained the correct received values $M(u,v)$ from each neighbor $u$. We prove in \Cref{sec:budgeted_k_edges_message_correction} the existence of the following procedure (which essentially mimics Step (2) of Subsection~\ref{sec:exactly_k_edges_improved}):
	
\begin{lemma}
\label{lem:correction_budgeted}
Let $t = \widetilde{\Omega}(\kDiam)$ be an integer known to all nodes. Given a distributed knowledge of a \emph{weak} $(k,\kDiam,\eta)$ tree packing $\cT$ for graph $G$, then the $d$-message correction procedure is a $\Theta(t)$-round procedure that is resilient to a round-error-rate is at most $d$, assuming that $d=O(\frac{k}{\eta \log{n}})$. The round complexity of the algorithm does not depend on $d$. 
\end{lemma}
	
Specifically, the algorithm applies the $d$-Message Correction Procedure of Lemma \ref{lem:correction_budgeted} with $d=\alpha f$ for a sufficiently large $\alpha \geq 1$ and the lists $\{M_i(u,v)\}_{v \in N(u)}$ (resp., $\{\widetilde{M}_i(v,u)\}_{v \in N(u)}$), as the outgoing (resp., incoming) values for each $u$. At the end of the correction, each node $v$ obtains for each neighbor $u$ a value $M'_i(u,v)=(m'_i(u,v),R'_i(u,v),x'_i(u,v),\ell'_i(u,v))$. If the promise for the correction procedure holds (number of mismatches is at most $d$) then $M'_i(u,v)=M_i(u,v)$ for every $(u,v) \in E$. 
	
\noindent \paragraph{Rewind-If-Error Phase:} Every node $v$ locally checks for every neighbor $u$ whether $|\widetilde{\pi}_i(u,v)| = \ell'_i(u,v)$ and $h_{R'_i(u,v)}(\widetilde{\pi}_i(u,v)) = h_{R'_i(u,v)}(x'_i(u,v))$. If these conditions hold for all neighbors, then $v$ sets $\GoodState_i(v) = 1$, and otherwise $\GoodState_i(v) = 0$. The network then computes $\GoodState_i = \min_v \GoodState_i(v)$ and $\ell_{i} = \max_{v} \gamma(v,i)$ in the following manner. 
	
For every $T_j \in \mathcal{T}$, we define the following $\Pi_j$ protocol: In tree $T_j$, perform an upcast of the variables $\GoodState_i(v)$ and $|\widetilde{\pi}_i(u,v)|$, taking the minimum and maximum during the upcast, respectively. Then, in a downcast from the root $v_r$ to the leaves, these values are propagated downwards. Finally, all nodes wait until a total of $t$ rounds elapse. Let $\Pi_{RS,j}$ be the RS-compiled procedure using \Cref{thm:less_than_one_error}. We schedule all $k$ algorithms $\Pi_{RS,1},\dots,\Pi_{RS,k}$, in parallel, using \Cref{lem:scheduler_security}. 
	
	Let $\ell'_{i}(v),\GoodState'_i(v)$ be the value that node $v$ receives in the majority of trees (or zero it there is no majority). If $\GoodState'_i(v) = 1$, the $v$ updates its incoming and received transcripts by:
	$$\widetilde{\pi}_{i+1}(u,v) = \widetilde{\pi}_{i}(u,v) \circ m_i'(u,v) \mbox{~and~} \pi_{i+1}(v,u) = \pi_{i}(v,u) \circ m_i(v,u), \forall u\in N(v)~.$$ 
If $\GoodState'_i(v) = 0$ and $\gamma(v,i) = \ell'_{i}(v)$, then $v$ sets:
$$\widetilde{\pi}_{i+1}(v,u) = \DeleteLast(\widetilde{\pi}_i(v,u)) \mbox{~and~} \pi_{i+1}(v,u)= \DeleteLast(\pi_i(v,u)), \forall u \in N(v),$$ 
where $\DeleteLast$ is a function that removes the last symbol of its input string. In the remaining case where $\GoodState'_i(v) = 0$ and $\gamma(v,i) \leq \ell'_{i}(v)-1$, the $(i+1)^{th}$ transcripts are unchanged. I.e., $\widetilde{\pi}_{i+1}(v,u) = \widetilde{\pi}_i(v,u)$ and $\pi_{i+1}(v,u)= \pi_i(v,u)$ for every $u \in N(v)$. 

This concludes the description of a global round. After the final global round, each node $v$ outputs a value according to what it would output in $\cA$ if its incoming transcript from each neighbor $u$ was $\widetilde{\pi}_{r'}(u,v)$ and according to $v$'s inputs.
	
\subsection{Analysis}
Let $t' = \widetilde{O}(\kDiam)$ be an upper bound on the round complexity of each of the phases, and set $\alpha = 15t'/t$. We note that the round complexity of  \Cref{lem:correction_budgeted} is unaffected by $\alpha$, so this is well defined.
	
We say a global-round is \emph{bad} if during its execution the adversary corrupts at least $\alpha ft$ messages (i.e., the error-rate for that global-round is $\Omega(f)$). Otherwise, the global-round is \emph{good}. Let $\Gamma(u,v)$ denote the transcript of the protocol $\cA$ on $(u,v)$ in a fault-free network, padded by a suffix of $\poly(n)$many $\bot$ symbols\footnote{Intuitively, it's important to add the $\bot$ symbols to the end of the estimated transcripts and to $\Gamma$, in order to allow the potential function, defined later, to grow beyond $r$. Otherwise, if not handled carefully, a single adversarial rewind in the last global round could have caused an error.}.
	
First, we prove \Cref{inv:interactive-same-size-transcript}, i.e., that in any time $i$, and for any node $v$, there exists a value $\gamma(v,i)$ such that for neighbor $u \in N(v)$, $|\pi(v,u)| = \gamma(v,i)$ and $|\widetilde{\pi}_i(u,v)| = \gamma(v,i)$.
	
\begin{proof}[Proof of \Cref{inv:interactive-same-size-transcript}]
Let $v$ be a node. We prove the invariant by induction on $i$. For $i=1$ the claim is trivial, since for any neighbor $u$, $\widetilde{\pi}_i(u,v) = \emptyset$ and  $\pi_i(v,u) = \emptyset$. Assume this holds for $i-1$, where $i \geq 2$. 
		
By induction assumption, there is some value $\gamma(v,i-1)$ such that $|\pi_{i-1}(v,u)| = \gamma(v,i-1)$ and $|\widetilde{\pi}_{i-1}(u,v)| = \gamma(v,i-1)$. Recall that for any neighbor $u$, the variables of $v$, $\widetilde{\pi}_{i}(u,v)$ and $\pi_{i}(v,u)$ are only set in the Rewind-If-Error phase. If $\GoodState'_{i-1}(v) = 1$ then $v$ adds a single symbol to $\widetilde{\pi}_{i}(u,v)$ and $\pi_{i}(v,u)$ for each neighbor $u$ compared to the previous global-round, and the induction claim follows with $\gamma(v,i) = \gamma(v,i-1)+1$. If $\GoodState'_{i-1}(v) = 0$ then if $\gamma(v,i-1) = \ell'_{i-1}$, $v$ removes a symbol from all $\widetilde{\pi}_{i}(u,v)$ and $\pi_{i}(v,u)$ compared to the previous global-round. Otherwise, $|\widetilde{\pi}_{i}(u,v)|$ and $|\pi_{i}(v,u)|$ remain unchanged compared to the previous global round.
	\end{proof}

	
	\begin{lemma}
		\label{lem:number_of_bad_rounds}
		At most $r$ global-rounds are bad.
	\end{lemma}
	\begin{proof}
		Since there are in total at most $3t'r'$ rounds, the total number of corrupted messages is at most $15t'\cdot r\cdot  f = \alpha \cdot t\cdot r \cdot f$. Since a bad global-round is a round that has at least $\alpha tr$ total faults, the number of bad global-rounds is at most $r$.
	\end{proof}

	For strings $a,b$, let $\prefix(a,b)$ be the maximum index $j$ for which $a,b$ agree on the first $j$ symbols. Let $g(u,v,i) = 2 \cdot \prefix(\widetilde{\pi}_i(u,v),\Gamma(u,v))$, $g(i) = \min_{(u,v) \in E} g(u,v,i)$. To analyze the progress of the algorithm, define the potential function $\Phi(i)$ by: 
\begin{equation}\label{eq:potential}
\Phi(i)=\min_{(u,v) \in E} \left(2\cdot\prefix(\widetilde{\pi}_i(u,v),\Gamma(u,v))\right)- \max_{(u,v) \in E} |\widetilde{\pi}_i(u,v)|~.
\end{equation}
In other words, $\Phi(i)=g(i) - \ell_i$. In the following we bound the the potential $\Phi(i)$ in an inductive manner, depending on whether $i$ is a good or bad global-round.

\begin{lemma}\label{lem:bad-potential}
If $i$ is a bad global-round, then $\Phi(i+1) \geq \Phi(i)-3$.
\end{lemma}
	\begin{proof}
In any global-round $i$, it holds that $\widetilde{\pi}_{i}(u,v)$ and $\widetilde{\pi}_{i+1}(u,v)$ differ by at most one symbol in the suffix (either (a) the last symbol is removed, (b) a new last symbol is added, or (c) the variable remains the same), for any $(u,v)\in E$. Therefore, $g(i+1) \geq g(i)-2$, and $\ell_{i+1} \leq \ell_{i}+1$. It follows that $\Phi(i+1) \geq \Phi(i)-3$.
	\end{proof}

To bound the potential increase for good global-rounds, we need the following auxiliary claims.
\begin{lemma}
\label{lem:budgeted_correct_message_in_good}
If $i$ is a good global-round, then each node $v$ receives from  every neighbor $u$ the correct value $M_i(u,v)$  in the Message-Correcting phase (i.e., $M'_i(u,v)=M_i(u,v)$ for every $u \in N(u)$), and the correct values of $\GoodState_i$ and $\ell_{i}$ in the Rewind-If-Error phase.
	\end{lemma}
	\begin{proof}
		 
		Since $i$ is a good global-round, there are at most $\alpha ft$ corrupted messages throughout its execution. Recall that for the message from $M(u,v)$ to be received incorrectly by $v$, there must be at least $t$ many faults on the edge $(u,v)$ in this phase. Therefore, in the Round-Initialization phase, there are at most $\alpha \cdot f$ adjacent pairs $u,v$ such that $M_i(u,v)$ is not correctly decoded by $v$.
		
		Assuming that after the Round-Initialization phase, there are at most $\alpha \cdot f$ adjacent pairs $u,v$ such that $\widetilde{M}_i(u,v) \neq M_i(u,v)$, the promise on the input of \Cref{lem:correction_budgeted} holds. By \Cref{lem:correction_budgeted}, the Message-Correcting Phase has round complexity $O(t)$ and resilience to $\Theta(\frac{k}{\eta\log{n}})$ corrupted messages, which for the assumed $k$ is more than $\alpha f t$ many total faults. Therefore, it succeeds assuming at most $\alpha \cdot f \cdot t$ corrupted messages.
		
		The Rewind-If-Error phase consists of a single application of the scheduler \Cref{lem:scheduler_security} on protocols with round complexity $t$, and some local computation. Therefore, at least $k - \alpha ft$ protocols succeed. By assumption of $k$, $k - \alpha ft \geq k/2 + 1$, therefore the majority value of $\GoodState_i,\ell_{i}$ received by any node $v$ is the correct values.
	\end{proof}

\begin{lemma}
\label{lem:if_good_so_far_then_next_sent_good}
Let $v$ be a node, $i \geq 1$ be an integer, and $0 \leq j \leq \gamma(v,i)-1$. If for all neighbors $u \in N(v)$ it holds that $\prefix(\widetilde{\pi}_i(u,v),\Gamma(u,v)) \geq j$, then for any neighbor $w \in N(v)$ it holds that $\prefix(\pi_i(v,w),\Gamma(u,v)) \geq j+1$.  
	\end{lemma}
	\begin{proof}
		Let $i^* \leq i-1$ be last global-round before $i$ in which $\gamma(v,i^*) = j$. By choice of $i^*$, for any index $i^* \leq i' \leq i$ it holds that $\gamma(v,i') \geq j+1$, and in particular, all $j+1$ symbols remain the same in between iteration $i^*$ and $i$, i.e.
		
		\begin{equation}
		\label{eq:if_good_so_far_index_star}
			\prefix(\widetilde{\pi}_{i^*}(u,v),\widetilde{\pi}_{i}(u,v)) \geq j~.
		\end{equation}

		By \Cref{eq:if_good_so_far_index_star} and the assumption, it holds that $\prefix(\widetilde{\pi}_{i^*}(u,v),\Gamma(u,v)) \geq j$, for all neighbors $u \in N(v)$. In addition, by the choice of $i^*$, $\gamma(v,i^*+1) = \gamma(v,i^*)+1$. Therefore, in that round $v$ set for every $w \in N(v)$ the value $\pi_{i^*+1}(v,w)$ such that  $\prefix(\pi_{i^*}(v,w),\Gamma(u,v)) = j+1$. By \Cref{eq:if_good_so_far_index_star} it follows that $\prefix(\pi_i(v,w),\Gamma(u,v)) \geq j+1$.
	\end{proof}

	\begin{lemma}
		\label{lem:if_all_transcripts_consistent_then_is_real}
		Assume that there exists $\gamma_i$ such that for any $u,v$, it holds that $|\widetilde{\pi}_i(u,v)| = |\pi_i(u,v)| = \gamma_i$. Then the following two conditions are equivalent:
		\begin{enumerate}[(i)]
			\item For all adjacent nodes $u,v$, it holds that $\pi_i(u,v) = \widetilde{\pi}_i(u,v)$.
			\item For all adjacent nodes $u,v$ it holds that $\prefix(\widetilde{\pi}_i(u,v),\Gamma(u,v)) = |\widetilde{\pi}_i(u,v)|$.
		\end{enumerate}
	\end{lemma}
	\begin{proof}
	First, we show that (1) implies (2). Assume that for all adjacent nodes $u,v$, it holds that $\pi_i(u,v) = \widetilde{\pi}_i(u,v)$. We prove by induction on $j \leq \gamma_i$ that $\prefix(\widetilde{\pi}_i(u,v),\Gamma(u,v)) \geq j$, for all $(u,v)\in E$. For $j=0$ the claim is trivial. Assume the claim holds up to $j-1$ and consider $j \geq 1$. By the induction assumption, for any node $u$, the first $j-1$ symbols $\widetilde{\pi}_i(w,u)$ are the same as $\Gamma(w,u)$, therefore, by \Cref{lem:if_good_so_far_then_next_sent_good}, the $j$'th entry in $\pi_i(u,v)$ is also consistent with $\Gamma(u,v)$. Moreover, since $\pi_i(u,v) = \widetilde{\pi}_i(u,v)$, it also holds that the $j$'th entry in $\widetilde{\pi}_i(u,v)$ is also consistent with $\Gamma(u,v)$. This concludes the induction.
	
	Next, we show that (2) implies (1). Assume that for all adjacent nodes $u,v$ it holds that 
$$\prefix(\widetilde{\pi}_i(u,v),\Gamma(u,v)) = |\widetilde{\pi}_i(u,v)|.$$ 
By \Cref{lem:if_good_so_far_then_next_sent_good}, since $\prefix(\widetilde{\pi}_i(u,v),\Gamma(u,v)) \geq \gamma_i -1$ for all neighbors $u \in N(v)$, then \\$\prefix(\pi_i(v,w),\Gamma(v,w)) \geq \gamma_i$ for any neighbor $w \in N(v)$. Since by assumption, \\$\prefix(\widetilde{\pi}_i(v,w),\Gamma(v,w)) \geq \gamma_i$, then it holds that $\prefix(\widetilde{\pi}_i(v,w),\pi_i(v,w)) \geq \gamma_i$. The claim follows.
	
\end{proof}

	\begin{lemma}
		\label{lem:in_good_round_prefix_not_decreases}
		If $i$ is a good global-round, then $g(i+1) \geq g(i)$.
	\end{lemma}
	\begin{proof}
By \Cref{lem:budgeted_correct_message_in_good}, each node $v$ receives at the end of the global-round the correct value $M_i(u,v)$. Let $(u,v) \in E$ be a pair of nodes minimizing $g(u,v,i)$.  We distinguish between the following cases:
		
\smallskip		
\noindent \textbf{Case 1: $|\widetilde{\pi}_i(u,v)| \leq \ell_i-1$ or $\prefix(\widetilde{\pi}_{i}(u,v),\Gamma(u,v)) < |\widetilde{\pi}_i(u,v)|$.} In both cases, it holds that $\prefix(\widetilde{\pi}_{i+1}(u,v),\Gamma(u,v)) = \prefix(\widetilde{\pi}_{i}(u,v),\Gamma(u,v))$. To see this note that only the last symbol of transcripts of length $\ell_i$ are deleted. 
		
\smallskip		
\noindent \textbf{Case 2:  $|\widetilde{\pi}_i(u,v)| = \ell_i$, and $\prefix(\widetilde{\pi}_{i}(u,v),\Gamma(u,v)) = |\widetilde{\pi}_i(u,v)|$}. Since the pair $u,v$ minimizes $g(i)$, then for all neighboring pairs $u',v'$ it holds that $g(u',v',i) \geq g(u,v,i) = 2\ell_i$. Consequently, $\prefix(\widetilde{\pi}_{i}(u',v'),\Gamma(u',v')) = |\widetilde{\pi}_i(u',v')|$, and $|\widetilde{\pi}_i(u',v')| = |\widetilde{\pi}_i(u,v)|$. By \Cref{inv:interactive-same-size-transcript} it holds that $\gamma(u',i) = \gamma(v',i)$ for any $u',v'$, and by \Cref{lem:if_all_transcripts_consistent_then_is_real}, it follows that $\GoodState_i = 1$, and no symbol is deleted in any transcript, i.e. $g(i+1) \geq g(i)$.
	\end{proof}

	\begin{lemma}\label{lem:good-potential}
		If $i$ is a good global-round, then $\Phi(i+1) \geq \Phi(i)+1$ w.h.p.
	\end{lemma}
	\begin{proof}
By \Cref{lem:budgeted_correct_message_in_good}, each node $v$ receives at the end of the global-round the correct value $M_i(u,v)$.

		If for some two nodes $v_1,v_2$, it holds that $\gamma(v_1,i) \neq \gamma(v_2,i)$, then every $v$ such that $\gamma(v,i) = \ell_i$ deletes the last symbol, meaning $f(i+1) = f(i) - 1$. Moreover, by \Cref{lem:in_good_round_prefix_not_decreases}, the value $g(i+1) \geq g(i)$, and the claim follows.
		
		Otherwise, the exists some $\gamma_i$ such that $\gamma(v,i) = \gamma_i$. We split into two cases: first, assume that there exist adjacent nodes $u,v$ such that $\pi_i(u,v) \neq \widetilde{\pi}_i(u,v)$. Since $\cH$ is a pairwise-independent hash function, $h_{R_i(u,v)}(\pi_i(u,v)) \neq h_{R_i(u,v)}(\widetilde{\pi}_i(u,v))$ w.h.p. over the randomness $R_i(u,v)$ during the ``Rewind-If-Error'' Phase.\footnote{We note that $\pi_i(u,v)$ and $\widetilde{\pi}_i(u,v)$ are set before $R_i(u,v)$ is chosen. Therefore  $h_{R_i(u,v)}(\pi_i(u,v)) \neq h_{R_i(u,v)}(\widetilde{\pi}_i(u,v))$ w.h.p. irregardless of the actions of the adversary.} Therefore, $\GoodState_i = 0$, and by \Cref{lem:budgeted_correct_message_in_good} all nodes in the network receive this value. Meaning each node $v$ deletes for each neighbor $u$ the last symbol of $\widetilde{\pi}_i(u,v)$, implying that $\ell_{i+1} = \ell_i - 1$. On the other hand, by \Cref{lem:in_good_round_prefix_not_decreases}, $g(i)$ does not decrease. This implies that in this case $\Phi(i+1) \geq \Phi(i)+1$.
		
		
		It remains to consider the case where $\gamma(v,i) = \gamma_i$ and $\pi_i(u,v) = \widetilde{\pi}(u,v)$ for all $(u,v) \in E$. By \Cref{lem:if_all_transcripts_consistent_then_is_real} on round $i$, combined with \Cref{lem:budgeted_correct_message_in_good}, each node $v$ receives from each $u$ the value $m_i(u,v)$, which is the next message according to $\Gamma$. Therefore, $g(i+1) = g(i)+2$, while $\ell_{i+1} \leq \ell_i+1$, meaning $\Phi(i+1) \geq \Phi(i)+1$.
	\end{proof}
	
\begin{lemma}\label{lem:budgeted-correctness}
$\prefix(\widetilde{\pi}_{r'}(u,v),\Gamma(u,v)) \geq r$, for every $(u,v)\in E$. Consequently, the protocol $\cA'$ has the same output as $\cA$.
\end{lemma}
\begin{proof}
By Lemma \ref{lem:good-potential}, in a good global-round, the potential function increases by at least one. By Lemma \ref{lem:bad-potential}, in a bad global-round, it decreases by at most three. By \Cref{lem:number_of_bad_rounds}, there are at most $r$ bad global-rounds, and at least $r'-r = 4r$ good global-rounds. Therefore, at the end of all the global-rounds, $\Phi(r') \geq 4r - 3r = r$. Since $\Phi(r')\geq r$, by Eq. (\ref{eq:potential}), it holds that 
$\prefix(\widetilde{\pi}_i(u,v),\Gamma(u,v))\geq \Phi(r') \geq r$ for every $(u,v)\in E$. . Recall that the output of $v$ in $\cA'$ is determined by $\{\widetilde{\pi}(u,v)\}_{u \in N(v)}$ and $v$'s input. Since the first $r$ symbols of the estimated incoming transcript are equal to the incoming transcript in $\cA$, the output is also the same. 


	\end{proof}

\noindent We are now ready to complete the proof of Thm. \ref{thm:budgeted_main_intro}.
\begin{proof}[Proof of Theorem \ref{thm:budgeted_main_intro}]
The correctness follows by \Cref{lem:budgeted-correctness}. The round complexity of the first phase is $\widetilde{O}(\kDiam)$ rounds. By Lemma \ref{lem:correction_budgeted}, the correction phase takes also $\widetilde{O}(\kDiam)$ rounds. Each Alg. $\Pi_j$ of the last phase takes $\widetilde{O}(\kDiam)$ rounds. Therefore, also $\Pi_{j,RS}$ takes $\widetilde{O}(\kDiam)$ rounds. The final round complexity then follows by the RS-scheduling of Lemma \ref{lem:scheduler_security}. The proof follows. 
\end{proof}

\subsection{Applications}\label{sec:app_budgeted}

\noindent \textbf{Congested Clique Model.} We provide a variant of Theorem \ref{thm:CC-compiler} for the bounded error rate setting.
	
\begin{theorem}[Mobile-Resilient Compilers in the Congested Clique]\label{thm:CC-compiler-budgeted}
For any algorithm $\cA$ that runs in $r$ congested-clique rounds, there is an equivalent algorithm $\cA'$ with $f$-round error rate resilience for $f=\Theta(n/\log{n})$ that runs in $\widetilde{O}(r)$ \congc\ rounds. 
	\end{theorem}
	\begin{proof}
		Similarly to \Cref{thm:CC-compiler}, The proof follows by noting that an $n$-node clique over vertices $V=\{v_1,\ldots, v_n\}$ contains a $(k,\kDiam,\eta)$ tree packing $\mathcal{T}=\{T_1,\ldots, T_k\}$ for $k=n$ and $\kDiam,\eta=2$.  Specifically, for every $i \in \{1,\ldots, n\}$, let $T_i = (V,E_i)$ where $E_i = \{(v_i,v_j) \mid v_j \in V\}$, be the star centered at $v_i$. It is easy to see that the diameter and the load is exactly $2$. The claim follows immediately from  \Cref{thm:budgeted_main_intro}. 
	\end{proof} 
	
\noindent \textbf{Expander Graphs.} Next we show an analog of \Cref{thm:expander-compiler} in to the bounded error rate setting. 
	
\begin{theorem}\label{thm:expander-compiler-budgeted}[Mobile-Resilient Compilers for Expander Graphs]
		Assume $G$ is a $\phi$-expander with minimum degree $k=\widetilde{\Omega}(1/\phi^2)$. Then, for any algorithm $\cA$ that runs in $r$ \congest\-rounds, there is an equivalent algorithm $\cA'$ which is resilient to round-error-rate $f$, for $f=\widetilde{O}(k\phi^2)$ that runs in $\widetilde{O}(r/\phi)$ \congest\-rounds. 
\end{theorem}

Let $c''$ be the constant of \Cref{lem:large_cond_low_diam}, and let $f = \frac{k\phi^2}{c\log^{c'}{n}}$ for sufficiently large constants $c,c'$, in particular with regards to $c''$. We assume $\cA$ is an $r$ round algorithm, and let $r'$ be the round complexity of an $r$ round algorithm compiled \Cref{thm:budgeted_main_intro} against round-error rate of at most $(4c''+2)f$. We slightly adapt the tree packing algorithm used in \Cref{thm:expander-compiler}, to have the part of computing the weak tree packing to be resilient to many faults. This is done very simply by repeating each round in this computation for $O(r' \cdot \phi/\log{n})$ rounds, and then running the algorithm of \Cref{thm:budgeted_main_intro}, such that it is resilient to round-error rate of $(4c''+2)f$.
	
\paragraph{The Algorithm.} We define a padded-round as a round of communication where each node $u$ sends to each neighbor $v$ a message $m(u,v)$ repeatedly for $s = r' \cdot \phi/\log{n}$ rounds. The received message of $v$ from $u$ is defined as the majority value seen by $v$ during these $s$ rounds, or $0$ if a majority does not exist. The following protocol is defined by a series of padded rounds:
	  
In the first padded-round, for every edge $(u,v) \in E$ with $\ID(u)>\ID(v)$, let $u$ choose a color $c(u,v)$ sampled uniformly at random in $[k]$ and send this color to $v$. Let $c(u,v)$ be the received message of $u$ from $v$ in this padded-round. $u$ sets the color of the edge $(u,v)$ to the received value $c(v,u)$. For every $i \in [k]$, define the directed subgraph: 
	$$G_i=\{(u,v) ~\mid~ \ID(u) > \ID(v), c(u,v)=i\} \cup \{(v,u) ~\mid~ \ID(u) > \ID(v), c(v,u)=i\}.$$

Let $N_i(u)$ be the outgoing neighbors of $v$ in $G_i$. Each node $u$ proceeds to do the following padded-round BFS procedure in each $G_i$: Set a variable $I_{i,0}(u) = \ID(u)$ and $\parent_i(u) = \bot$ for each color $i \in [k]$. For $\ell = 1,\dots,z = 4c''\log{n}/\phi$ padded-rounds, in parallel for all colors $[k]$, each node $u$ sends to each (outgoing) neighbor $v \in N_i(u)$ the value $I_{i,\ell-1}(u)$. Let $L_{i,\ell}$ be the set of messages $u$ receives from neighbors $\{v \in N(u)  \mid c(u,v) = i\}$. Each node $u$ sets $I_{i,\ell}(u) = \max_{\ID(v) \in L_{i,\ell} \cup \{\ID(u)\}}(\ID(v))$, and if this value strictly increases, then it sets $\parent_i(u)$ to be the neighbor from which this value has arrived, oriented towards the other endpoint. In the final padded-round, each node $u$ sends in parallel for each $i$ a message to $\parent_i(u)$, to orient the edge $\parent_i(u)$ towards itself. Each node that receives an edge orientation request, locally sets the orientation of that edge towards itself.
	
\paragraph*{Analysis}
	
	A color $i \in [k]$ is denoted as \emph{good} if the adversary did not control any edge of $G_i$ for more than $s/2$ rounds during the tree computation phase. Otherwise, the color $i$ is defined as \emph{bad}. 
	
	In particular, in any good color, the adversary did not change the outcome of any of the BFS procedures, and the guarantees of Lemmas \ref{lem:good-expander-sampled} and \ref{lem:good-properties} hold for the altered algorithm and altered notion of good color, as formalized in the following lemma:
	
	\begin{lemma}\label{lem:good-expander-sampled-budgeted}
		For each good color $i \in [k]$, the output directed subgraph $G'_i$ defined by $\{\parent_i(u)\}_{u \in V}$ is a directed spanning tree of depth $4c''\log n/ \phi$ oriented towards the node with the largest ID in the network, $v_r = \argmax_v \ID(v)$.
	\end{lemma}
	\begin{proof}
		For each good color $i \in [k]$, we have that $G_i=G[p]$ where each edge in $G$ is sampled independently with probability $p  = 1/k = O(\frac{\phi\log{n}}{f})$. Note that since the adversary does not affect the padded-received message of any of the $G_i$ edges, all edges in $G_i$ are bidirectional. 
		
		By \Cref{thm:conduct-karger} for $\alpha = \phi$, $\gamma = O(\frac{f}{\phi^3})$ such that $p=\min\{1,\ln n/(\alpha^2 \gamma)\}$, then w.h.p. a subgraph obtained by sampling by taking each edge in $G$ with probability $p$ has conductance at least $\phi/4$. By the union bound, all graphs $G_i$ have conductance $\geq \phi/4$, w.h.p.  Therefore, by \Cref{lem:large_cond_low_diam} the graph $G_i$ has diameter at most $4c'' \cdot\log n/\phi$.
		
		Since $i$ is good, the adversary did not control any of the edges of $G_i$ throughout the entire algorithm. Since $v_r = \argmax_v \ID(v)$, and all the message $v_r$ receives are real IDs of nodes, it never changes the value $\parent_i(v_r) = \bot$. Let $V_i(j)$ be nodes of distance $j$ from $v_r$ in $G_i$. Let $V_i(-1)=\{\bot\}$. We show by induction on $j\geq 0$ that for any $u \in V_i(j)$, it holds that $I_{i,j}(u) = \ID(v_r)$ and $\parent_i(u) \in V_i(j-1)$. Clearly, $I_{i,0}(v_r)=\ID(v_r)$ and as $\parent_i(v_r)=\bot$, the claim holds for $j=0$. Assume this is the case for $j-1$. We note that $V \setminus \left(\bigcup_{a=0}^{j-1} V_a\right)$ has not seen the value $\ID(v_r)$ before round $j$, as they are of distance at least $j$ from $v_r$ in $G_j$. By induction, all $u\in V_{j-1}$ set $I_{i,j-1}(u) = v_r$, therefore any node $u \in V_j$ receives a message from a node $v \in V_{j-1}$ with the value $\ID(v_r)$ and sets its parent to be in $V_{j-1}$. The induction holds and the claim follows.
	\end{proof}

	We conclude with the following claim:
	
	\begin{lemma}\label{lem:stepone-correct-budgeted}
		The collection of subgraphs $\{ G'_i\}_{i \in [k]}$ is a \emph{weak}-$(k, \kDiam,\eta)$ Tree-Packing with $\kDiam=O(\log n/\phi)$ and $\eta=2$. 
	\end{lemma}
	\begin{proof}
		The total number of corrupted messages is at most $(4c''+2)r'f$. For a color to be bad, there has to be at least $r'(\phi/\log{n})/2$ corrupted messages associated with it. Since we assume $f=\widetilde{O}(k\phi^2) = \widetilde{O}(r'\phi)$, there are at most $\frac{(4c''+2)r'f}{r'(\phi/\log{n})/2} = \frac{(4c''+2) \cdot f}{(\phi/\log{n})/2} \leq k/10$ bad colors, where the last inequality follows from choosing large enough constants $c,c'$ for $f$. Therefore, there are at least $0.9k$ good colors. By Lemma \ref{lem:good-expander-sampled-budgeted}, for every good color $i$, the subgraph $G'_i$ corresponds to a spanning tree of depth at most $\kDiam$ and rooted at $v_r$. Since each directed edge is committed to a single color, we get that $\eta\leq 2$ (even among the bad subgraphs).
	\end{proof}

	\begin{proof}[Proof of \Cref{thm:expander-compiler-budgeted}]
		The round complexity of this procedure is at most $(4c''+2)r'$. By \Cref{lem:stepone-correct-budgeted}, the first phase of the algorithm computes a weak-$(k, \kDiam,\eta)$ tree packing. In the second phase, we run $\cA$ compiled by \ref{lem:good-expander-sampled-budgeted} on round-error rate $(4c''+2)f$. Since its round complexity of the second phase is $r'$ and since there are at most $(4c''+2)r'f$ corrupted messages in total, we are guaranteed that the output is the same as in $\cA$. 
	\end{proof}

%% file: correction-via-cycle-covers.tex
\section{Mobile Resilience using Fault-Tolerant Cycle Covers}\label{sec:FTcycle-cover}

In this section we show a $f$-mobile-resilient algorithms. We use the same techniques as in \cite{PatraCRSR09} for the unicast case, and extend this approach to obtain a general compiler using cycle covers. We also note that this approach is also a direct generalization of the approach of the compiler of \cite{ParterYPODC19} against a $1$-mobile byzantine adversary.

\begin{definition}[Low-Congestion FT Cycle-Covers]
For a given graph $G=(V,E)$, an $f$-FT $(\congestion,\dilation)$-Cycle Cover is a collection of paths\footnote{For our purposes, it is instructive to view it as a collection of $f$ edge-disjoint paths, between each neighboring pair $u,v$. This is equivalent to $(f-1)$ cycles covering the edge $(u,v)$.} $\mathcal{P}=\bigcup_{e \in E}\mathcal{P}(e)$ where each $\mathcal{P}(e=(u,v))$ consists of $k$ edge-disjoint $u$-$v$ paths, with the following properties: (i) $\max_{P \in \mathcal{P}}|P|\leq \dilation$ and (ii) each $e \in E$ appears on at most $\congestion$ paths in $\mathcal{P}$ (i.e., $\load(\mathcal{P})\leq \congestion$). 
\end{definition}

\begin{theorem}[Existence of Low-Congestion FT-Cycle Covers, \cite{HitronP21a}]\label{thm:LCFT}
Every $f$-edge connected graph $G=(V,E)$ admits an $f$-FT $(\congestion,\dilation)$-Cycle Cover with $\congestion, \dilation= D^{O(f)}\log(n)$. Moreover, in the fault-free setting these cycles can be computed in $\congestion+ \dilation$ rounds. 
\end{theorem}

For an $f$-FT cycle cover $\mathcal{P}$, we define the \emph{path-conflict graph} to be a graph $H=(E,E_H)$ with a vertex $v_e$ for every edge $e \in E$, and where $\{v_{e_1},v_{e_2}\} \in E_H$ if and only if there are paths $P_1 \in P(e_1)$ and $P_2 \in P(e_2)$ that share at least one edge. 

\begin{lemma}\label{lem:coloring-ft-cycle-cover}
	Let $\mathcal{P}$ be an $f$-FT $(\congestion,\dilation)$-Cycle Cover. There exists a coloring of the edges of $E$, $\Col:E \rightarrow [f \cdot \dilation \cdot \congestion+1]$, such that for two distinct edges $e_1,e_2 \in E$, if $\Col(e_1) = \Col(e_2)$, then any two paths $P_1 \in P(e_1)$ and $P_2 \in P(e_2)$ are edge-disjoint.
\end{lemma}
\begin{proof}
	
	We note that in the \emph{path-conflict graph} $H$, the degree of each vertex $v_e$ in $H$ is at most $f\cdot\dilation\cdot\congestion$, since for each $v_e$ there are $f$ paths of size at most $\dilation$, and for each such edge $e$, there are at most $\congestion$ other paths in $\mathcal{P}$ containing $e$. Therefore, there exist a coloring $\Col$ of the vertices by $(f\cdot\dilation\cdot\congestion+1)$ colors so that no two adjacent vertices have the same color.
	
	In the graph $G$, coloring each edge of $e \in E$ by the color $\Col(v_e)$, we obtain a coloring such that if two edges  $e_1,e_2$ share the same color, then all paths $P(e_1)$ and $P(e_2)$ are edge disjoint.
\end{proof}

We call an edge coloring $\Col$ with the above properties of \Cref{lem:coloring-ft-cycle-cover} good cycle coloring.  We note that in a fault-free network, a good cycle coloring can be computed for an $f$-FT $(\congestion,\dilation)$-Cycle Cover by simulating a $(\Delta+1)$-coloring in the the \emph{path-conflict graph} $H$ (as defined in \Cref{lem:coloring-ft-cycle-cover}). Next, we show that  $G$ can simulate a round in $H$ using $O(\dilation\cdot\congestion^2)$ rounds.

\begin{lemma}\label{lem:good_color_simul}
	For a $\mathcal{P}$ be an $f$-FT $(\congestion,\dilation)$-Cycle Cover, a $\congest$-round of its \emph{path-conflict graph} can be simulated in $O(\dilation\cdot\congestion^2)$ $\congest$-rounds in $G$ under a fault-free assumption.
\end{lemma} 
\begin{proof}
	In the first step, for each edge $(u,v)$, the vertex $u$ sends through all paths $P(u,v)$ the identifier $\ID(u) \circ \ID(v)$. This can be done in $O(\dilation \cdot \congestion)$ rounds, by a standard pipelining argument. Following this step, the endpoints of each edge $e$ knows entire list $W_e$ of edge identifiers $(u,v)$ such that $e \in P(u,v)$. In the second step, each edge $e \in E$ sends back through each of these paths the entire list $W_e$. Since for any $e$ the list $|W_e| \leq \congestion$, this can be done in $O(\congestion^2 \cdot \dilation)$ rounds using a standard pipelining argument. Given this, each node $u$ knows for each incident edge $(u,v)$ the neighbors of $(u,v)$ in the \emph{path-conflict graph}. To simulate a $\congest$ round in this graph, each node transmits in reverse the communication it received in the second step, but for each edge identifier $(u',v')$ it received, it also sends the message $(u,v)$ would send to $(u',v')$ in the simulated $\congest$ round on $H$. The claim follows.
\end{proof}

\begin{corollary}
	For a $\mathcal{P}$ be an $f$-FT $(\congestion,\dilation)$-Cycle Cover, a good coloring can be found in $\widetilde{O}(\congestion^2 \cdot \dilation)$ rounds of $G$. 
\end{corollary}
\begin{proof}
	By \Cref{lem:good_color_simul}, we can simulate the deterministic $(\Delta+1)$-coloring algorithm of \cite{GK21}, which runs in $\widetilde{O}(1)$ congest rounds. If a vertex $v_e \in H$ is assigned the color $c(e)$, the edge $e$ colors itself with $c(e)$. By definition of the \emph{path-conflict graph}, any edge pair $e_1,e_2$ share an edge if and only if their paths intersect. Therefore, a vertex coloring on $H$ such that two adjacent vertices share the same color implies a good coloring on the edges of $G$ w.r.t. the path cover $\mathcal{P}$.
\end{proof}

\begin{theorem}[$f$-Mobile Resilient Simulation via FT Cycle-Covers]\label{thm:CC-mobile-simulation}
Assume that a $k$-FT $(\congestion,\dilation)$-Cycle Cover is known in a distributed manner, and a good cycle coloring is known for the edges, and let $\cA$ be an $r$-round algorithm in the fault-free setting. Then, there is an equivalent $f$-mobile-resilient algorithm $\cA'$ with round complexity of $r'$ for $f\leq k/c$ for a sufficiently large constant $c$ and $r'= \dilation \cdot \congestion \cdot r$. 
\end{theorem}

By combining Theorem \ref{thm:LCFT} and \ref{thm:CC-mobile-simulation}, we obtain Theorem \ref{cor:CC-mobile-simulation-final}. The rest of the section is devoted for proving Thm. \ref{thm:CC-mobile-simulation}.

%
%
\noindent\paragraph{Simulation of the $i^{th}$ Round of Alg $\cA$.}
Assume the network simulated rounds $1,\dots,i-1$, and each node $v$ is given as input the some outgoing messages $m_i(v,u_1),\dots,m_i(v,u_{\deg(v)})$ for every neighbor $u_i$, where our goal is for every node $v$ to output $m_i(u_1,v),\dots,m_i(u_{\deg(v)},v)$.

	
	The network runs the following protocol: For any $j \in [k]$, let $E_j = \{e \mid \Col(e) = j-1\}$. The network performs $j = 1,\dots,f \cdot \dilation \cdot \congestion+1$ iterations. In each iteration, for $t = 1,\dots,(2f\cdot\dilation+\dilation+1)$, in parallel for each $(u,v) \in E_j$, node $u$ sends $m_i(u,v)$ repeatedly over each path $P \in P(e)$, and whenever an edge in in $P(u,v)$ receives a message, the receiving endpoint propagates it forward to the next edge in the path. Let $m_i(u,v,\ell,t)$ be the message $v$ receives over the last edge of $P_\ell$ exactly $t$ rounds after the start of iteration $\Col(u,v)$. Let $M_i(u,v) = \{m_i(u,v,\ell,t) \stt \ell \leq k \land \dilation \leq t \leq 2f\cdot\dilation+\dilation+1\}$ be the multi-set of messages $v$ receives between rounds $\kDiam$ and $2f\cdot\dilation+\dilation+1$. For each neighbor $u$, the node $v$ outputs the majority value over $\MAJ(M_i(u,v))$.
	
	\begin{lemma}
		Let $(u,v) \in E_j$. Then $\MAJ(M_i(u,v)) = m_i(u,v)$.
	\end{lemma}
	\begin{proof}
		The number of rounds in iteration $j$ is $t = 2f\dilation+\dilation+1$. Therefore, there are at most $(2f\dilation+\dilation+1)f$ faults. Let $L_i(u,v)$ be the number of messages in $M_i$ that are not equal to $m_i(u,v)$. Since any fault changes the value of at most one $m_i(u,v,\ell,t) \in M_i$, then $L_i(u,v) \leq (2f\dilation+\dilation+1)f$. But the number of $m_i(u,v,\ell,t)$ that $v$ receives is $(2f\dilation+\dilation+1 - \dilation) \cdot k = (2f\dilation+1) \cdot (2f+1) = 2f(f\dilation +\dilation+1)+1 \geq 2L_i(u,v)+1$. Therefore, at least $L_i(u,v)+1$ of the messages are equal to $m_i(u,v)$, meaning $\MAJ(M_i(u,v)) = m_i(u,v)$.
	\end{proof}
	
	
	
	Theorem \ref{thm:CC-mobile-simulation} follows, since in each round $i$, and each pair of adjacent nodes $u,v$, $v$ successfully receives $m_i(u,v)$.

%% file: eavesdropper-mobile-improved.tex
\section{Translation of Fault-Free Algorithms into  $f$-Mobile-Secure Algorithms}\label{sec:fault_free_to_f_secure}

In this section we prove \Cref{thm:mobile-compilers} by presenting a simulation result that translates any low-congestion (fault-free) algorithm $\cA$ into an equivalent $f$-mobile-secure algorithm $\cA$. 
The connectivity requirements and the round overhead match those obtained by \cite{EavesMST2023} for the $f$-\emph{static} setting. The additional benefit of our compiler is that in contrast to \cite{EavesMST2023}, it provides perfect, rather than statistical, security guarantees. Similarly to \cite{EavesMST2023}, the key task is in exchanging $\widetilde{O}(f \cdot \congestion)$ secrets to all nodes in the graph, that are hidden from the adversary. This calls for providing $f$-mobile-secure algorithms for unicast and broadcast task. In Subsec. \ref{sec:fmobile-unicast}, we describe the adaptation of secure unicast to the mobile setting. Subsec. \ref{sec:fmobile-broascast} describes the modified $f$-mobile-secure broadcast algorithm. Finally, Subsec. \ref{sec:fmobile-LWcompiler} presents the final $f$-mobile compiler, with prefect-security guarantees. 

Throughout, we use the following lemma, which follows by Theorem \ref{thm:extract}.

\begin{lemma}\label{lem:sufficient-secrets}
Let $t,r,f$ be input integers. Then, there is an $r+t$ round procedure that allows each neighbors pair $u,v$ to output a list $\mathcal{K}(u,v)$ of $r$ many $B$-bit words for $B=O(\log n)$ such that the following holds: the $f$-mobile adversary learns nothing, on all but the $\mathcal{K}(u,v)$ sets of at most $f'$ edges, where $f'=\lfloor (f \cdot (t+1))/(r+t) \rfloor$. Moreover, for $t \geq 2fr$, it holds that $f'=f$. 
\end{lemma}
\begin{proof}
Let each neighboring pair $u,v$ exchange $r+t$ random messages $\mathcal{R}(u,v)=\{R_1(u,v),\ldots, R_{r+t}(u,v)\}$. Each $u,v$ apply the algorithm of Thm. \ref{thm:extract} with $n=r+t$ and $k=O(\log n)$ to compute the set $\mathcal{K}(u,v)$ of $r$ words of length $B=O(\log n)$. An edge $(u,v)$ is good if the adversary eavesdrops the edge, during this $r+t$ round procedure, for at most $t$ rounds. Otherwise, the edge is bad. By \Cref{thm:extract}, it holds that the adversary learns nothing on the $\mathcal{K}(u,v)$ set of a good edge. Since the adversary controls at most $f$ edges in each round, it eavesdrops over at most $f(r+t)$ edges. By averaging, there are at most $\lfloor f(r+t)/(t+1) \rfloor$ bad edges. 
\end{proof}

\subsection{Secure Unicast with Mobile Adversaries} \label{sec:fmobile-unicast}

In the static-secure unicast problem, a given source node $s$ is required to send a (secret) message $m^*$ to a designated target $t$, while leaking no information to an eavesdropper adversary controlling a fixed (i.e, static) subset of edges $F \subseteq G$ in the graph. Recently, \cite{EavesBroad2022} observed that the secure network coding algorithm of Jain \cite{Jain04} to the problem can be implemented in an optimal number of \congest\ rounds and with minimal congestion. 

\begin{theorem}[(Static) Secure-Unicast, \cite{EavesBroad2022,Jain04}]\label{lemma:unicast}
Given a $D$--diameter graph $G$, there is an $O(D)$-round $1$-congestion algorithm $\StaticSecureUnicast$ that allows a sender $s$ to send a message $m^*$ to a target $t$, while leaking no information to an eavesdropper adversary controlling $F^*$ edges in the graph, provided that $s$ and $t$ are connected in $G \setminus F^*$. 
\end{theorem} 

We start by observing that Jain's unicast algorithm, upon slight modifications, can become secure against mobile eavesdropper adversary, provided that the set of edges $F_i$ controlled by the adversary in round $i$ does not disconnect the graph. In fact, our requirements on the $\{F_i\}$ sets can be relaxed even further. We show the following which matches the known bounds for the static setting.

\begin{lemma}[Mobile Secure-(Multi) Unicast]\label{lemma:unicast-mobile}
Given a $D$--diameter graph $G$, there is an $O(D)$-round $2$-congestion algorithm $\MobileSecureUnicast$ that allows a sender $s$ to send a message $m^*$ to a target $t$, while leaking no information to a mobile eavesdropper adversary controlling distinct set of edges $F_i$ in every round $i$ of the algorithm,  provided that $s$ and $t$ are connected in $G \setminus F_1$ (where possible $F_i=E(G)$ for every $i \geq 2$).  
Moreover, there is an algorithm $\MobileSecureMulticast$ which can solve $R$ secure unicast instances $(s_j,t_j,m_j)$ in parallel using $O(D+R)$ rounds. The security holds provided that each pair $s_j$ and $t_j$ are connected in $G \setminus F_j$, for every $j \in \{1,\ldots, R\}$. 
\end{lemma} 
%
\noindent \textbf{Mobile-Secure-Unicast Alg. $\MobileSecureUnicast$.} The algorithm starts with a preliminary round in which nodes exchange random $O(\log n)$-bit messages $K(u,v) \in \mathbb{F}_q$ over all graph edges $(u,v)\in E$, for a prime $q \in \poly(n)$. These messages are next used as OTP keys for encrypting the messages of Alg. $\StaticSecureUnicast$ applied with the input $s,t$ and the secret message $m^*$. For every $(u,v)\in E$, letting $m_i(u,v)$ be a non-empty message sent from $u$ to $v$ in round $i$ of Alg. $\StaticSecureUnicast$, then in corresponding round $i$ of Alg. $\MobileSecureUnicast$ (i.e., round $i+1$) $u$ sends $v$ the message $m_i(u,v)+K(u,v)$. The correctness follows immediately and it remains to prove security.

\begin{claim}\label{cl:security-mobile-unicast}
Let $F^*$ be the set of edges controlled by the adversary in the first round of Alg. $\MobileSecureUnicast$. Then, provided that $s$ and $t$ are connected in $G \setminus F^*$, the adversary learns nothing on the secret message $m^*$. This holds even if the adversary controls all the edges in the graph in every round $i\geq 2$.
\end{claim}
\begin{proof}
The security uses the fact that Jain's unicast algorithm of Lemma \ref{lemma:unicast} sends exactly one message on each edge $(u,v)$ (in exactly one direction, either from $u$ to $v$ or vice-versa). For a subset of edges $E' \subseteq E$, let $\mathcal{K}(E')=\{K(u,v)\}_{(u,v)\in E'}$. 

First observe that the security is immediate, by the OTP guarantees, if all keys in $\mathcal{K}(E)$ are hidden from the adversary. This holds as every key $K(u,v)$ is used exactly once, as we send at most one message from $u$ to $v$ throughout the entire algorithm. Since the adversary controls only the edges $F^*$ in the first round, it fully knows $\mathcal{K}(F^*)$, and has \emph{no} information on all the remaining keys $\mathcal{K}(E \setminus F^*)$

Consider the worst case scenario where the adversary controls all edges in $G$ in every round $i\geq 1$ when simulating Alg. $\MobileSecureUnicast$. Since the adversary knows the keys $\mathcal{K}(F^*)$, it knows all the messages sent throughout the edges of $F^*$ in Jain's algorithm. All remaining messages sent through edges in $E \setminus F^*$ are distributed uniformly at random (by the OTP security). Therefore, the view obtained by the mobile adversary in this case is equivalent to the view obtained by a static adversary controlling $F^*$. The security follows by Lemma \ref{lemma:unicast}. 
\end{proof}

\begin{proof}[Proof of Lemma \ref{lemma:unicast-mobile}]
By Claim Claim \ref{cl:security-mobile-unicast}, it remains to handle the multi-unicast problem, where the instance consists of $R$ triplets $\{(s_i,t_i,m_i\}_{i=1}^R$. Alg. $\MobileSecureMulticast$ starts with a phase of $R$ rounds in which random messages are exchanged over all graph edges. For every $(u,v)\in E$, denote these messages as a list of $R$ keys: $K_1(u,v), \ldots, K_R(u,v)$. For a subset of edges $E'$, let $\mathcal{K}_i(E')=\{K_i(u,v)\}_{(u,v) \in E'}$. 
For every $i \in \{1,\ldots, R\}$, let $\cA_i$ denote that application of Alg. $\StaticSecureUnicast$ with input $(s_i,t_i,m_i)$ in which messages are encrypted (and decrypted) with the keys of $\mathcal{K}_i(E)$. The algorithm runs all $R$ algorithms $\{\cA_i\}$ in parallel. Using the standard random delay approach of Theorem \ref{thm:delay}, this can be done within $\widetilde{O}(D+R)$ rounds.  This completes the description of the algorithm.

The correctness is immediate and we consider security in the worst case setting where that adversary controls all edges in round $i\geq R+1$. That is, $F_i=E$ for every $i \geq R+1$. In the random delay approach, each of the $R$ algorithms proceed at a speed of phases, each consisting of $O(\log n)$ rounds. We note that such a scheduling might be risky in the general secure-mobile setting, as a single round is now simulated in $\ell=O(\log n)$ rounds which allows the adversary to control $f \cdot \ell$ faults rather than $f$ faults (when the algorithm runs in isolation). This concern does not appear in the context of Jain's algorithm due to the fact that it has congestion of exactly $1$. For the reason, we can assume (as in Claim \ref{cl:security-mobile-unicast}), that the adversary controls all edges in rounds $i \geq R+1$ and consequently the fact that a round is now completed within $O(\log n)$ rounds cannot further increase the power of the adversary. 

Specifically, for every algorithm $\cA_i$ the adversary knows all keys of $\mathcal{K}_i(F_i)$ and it knows nothing on the keys $\mathcal{K}_i(E \setminus F_i)$. Since there is at most one message of $\cA_i$ on each edge, the adversary learns nothing on the messages exchanged over $E \setminus F_i$ and fully knows the messages exchanged over $F_i$. Consequently, the view of the adversary in the simulation of $\cA_i$ is analogous to the view of the static adversary controlling the edges in $F_i$. The security then holds by the security guarantees of Alg. $\StaticSecureUnicast$ (Lemma \ref{lemma:unicast}).  
%
\end{proof}

\input{mobile-secure-broadcast.tex}

\subsection{Congestion-Sensitive Compiler with $f$-Mobile Security} \label{sec:fmobile-LWcompiler}

We are now ready to describe our $f$-mobile compilers. While the round overhead of our compiler match those obtained by \cite{EavesMST2023} they differ from it in two major points. First, our compilers provide $f$-mobile rather than $f$-static security; The compiler consists of the following three steps:

\noindent\textbf{Step 1: Local Secret Exchange.} Step 1 consists of $\ell=\Theta(r)$ rounds in which all neighboring pairs exchange random field elements $R_j(u,v) \in \mathbb{F}_q$ for every $j \in \{1,\ldots, \ell\}$. Using Theorem \ref{thm:extract}, at the end of this phase, each pair $(u,v) \in E$ holds $r$ keys $\{K_i(u,v)\}_{i=1}^{r}$ such that the adversary learns nothing on the keys of all but at most $4f$ \emph{bad} edges\footnote{I.e., corresponding to the edges occupied by the adversary for at least $t+1=\Theta(r)$ rounds in that phase.}. 

\smallskip

\noindent\textbf{Step 2: Global Secret Exchange.} The second step shares a collection of $b=O(f \cdot \congestion \cdot \log n)$ random bits, denoted as $R \in \{0,1\}^{b}$, to the entire network, while guaranteeing that the adversary learns nothing on $R$. This is done by applying our $f$-mobile-secure broadcast algorithm $\MobileBroadcast$ of Theorem \ref{thm:bc-sqrt-mobile}. Let $\mathcal{H} = \{h : \{0,1\}^p \rightarrow \{0,1\}^q\}$ be a family of $\cj$-wise independent hash functions for $\cj=4f \cdot \congestion$ and $p,q=O(\log n)$ (see Lemma \ref{lem:hash}). At the end of this step, all nodes use $R$ to choose random function $h^* \in \mathcal{H}$. 

\smallskip

\noindent\textbf{Step 3: Round-by-Round Simulation of Alg. $\cA$.} For every $i \in \{1,\ldots, r\}$, let $m_i(u,v)$ be the message sent from $u$ to $v$ in round $i$ of (the fault-free) Alg. $\cA$. We assume, w.l.o.g., that all the messages $\{m_i(u,v)\}_{(u,v) \in E, i \in \{1,\ldots, r\}}$ are \emph{distinct}. This can be easily done\footnote{I.e., assuming that the bandwidth of the CONGEST model is $c\log n$ for some constant $c$. Otherwise, we can simulate each round using a constant number of rounds.} by appending to each message $m_i(u,v)$ the round-number and the edge-ID. As Algorithm $\cA$ has bounded congestion, some of these $m_i(u,v)$ messages are empty, however, the compiled algorithm must still exchange messages on all edges in every round. We let every vertex $u$ act differently according to whether $m_i(u,v)$ is empty, but in a way that is indistinguishable to the adversary. If $m_i(u,v)$ is \emph{not} empty, then $u$ pads the message $m_i(u,v)$ with $0$'s of length $B'$, by letting $m=m_i(u,v) \circ 0^{B'-B}$ and sending the message $\widehat{m}_i(u,v)=h^*(m) + K_i(u,v)$ to $v$.  Otherwise, if $m_i(u,v)$ is empty, then $u$ simply sends a uniformly random string length $B'$.

While the adversary will not be able to distinguish between real messages and random messages, the receiver $v$ will indeed be able to distinguish between the two, by applying following decoding procedure for each received message $\widehat{m}_i(u,v)$: Iterate over all possible messages of length $O(\log n)$ and let $m'_i(u,v)$ be the first (lexicographically) message $m$ satisfying that $h^*(m)=\widehat{m}_i(u,v)+K_i(u,v)$. Note that all nodes, and in particular, $v$, know $h^*$. If $m'_i(u,v)$ ends with $B'-B$ 0's, then $v$ knows that the message was not empty, w.h.p, and thus sets the final message $\widehat{m}_i(u,v)$ to be the first $B$ bits of $m'_i(u,v)$. Otherwise, $v$ knows that a random string was sent, and thus sets $\widehat{m}_i(u,v)$ to be an empty message. The summary of this protocol is given below.

\begin{mdframed}[hidealllines=false,backgroundcolor=white!30]
	\begin{description}
		\item[\noindent \textbf{$f$-Mobile Secure Simulation of Alg. $\cA$:}]
		
		\item[\textbf{Step (1): Local Secret Exchange.}] \; 
		\begin{enumerate}
			\item Exchange $\ell=\Theta(r)$ random messages over all edges $(u,v)\in E$.  
			\item Each node $u,v$ locally applies Theorem \ref{thm:extract} to compute $r$ keys $\mathcal{K}(u,v)=\{K_i(u,v)\}_{i=1}^{r}$ for every $(u,v)\in E$.
		\end{enumerate}
		
		\item[\textbf{Step (2): Global Secret Exchange.}] \; 
		\begin{enumerate}
			\item Let an arbitrary vertex $s$ compute a random $R \in \{0,1\}^b$ for $b=\Theta(f \cdot \congestion \cdot \log n)$.
			\item Apply Alg. $\MobileBroadcast$ to securely share the (secret) message $R$.
		\end{enumerate}
		
		\item [\textbf{Step (3): Round by Round Simulation, Sim. Round $i$ of $\cA$}:] Let $m_i(u,v)$ be the unique message the $u$ needs to send to $v$ in round $i$. 
		\begin{enumerate}
			\item If $m_i(u,v)$ is not empty set $\widehat{m}_i(u,v) \gets h^*(m_i(u,v) \circ 0^{B'-B})+K_i(u,v)$.
			\item Otherwise, set $\widehat{m}_i(u,v)$ to be a uniformly random message of length $B'$.
			\item $u$ sends $\widehat{m}_i(u,v)$ to $v$.
		\end{enumerate}
		
		\item [\textbf{Receiving round $i$}] Let $\widehat{m}_i(u,v)$ be the message received by $v$ from $u$.
		\begin{enumerate}
			\item Compute $m' \gets \widehat{m}_i(u,v)+K_i(u,v)$.
			\item Compute $m$ as the first lexicographically message satisfying $h^*(m)=m'$.
			\item If $m$ does not end with $B'-B$ $0$'s, then drop the message (i.e., no message from $u$ has been received).
			\item Otherwise, set $\widetilde{m}_i(u,v)$ to be the first $B$ bits of $m$.
		\end{enumerate}
	\end{description}
\end{mdframed}

\paragraph{Proof of Theorem \ref{thm:mobile-compilers}.} We next turn to analyze the correctness of the construction. 

\begin{lemma}[Correctness of Simulation]\label{lem:correct-compiler}
W.h.p., for every rounds $i \in [r]$ and every pair of vertices $(u,v) \in E$, it holds that:
	\begin{itemize}
		\item If $m_{i}(u,v)$ is not empty $\widetilde{m}_i(u,v)=m_{i}(u,v)$.
		\item If $m_{i}(u,v)$ is empty then $v$ drops the received message.
	\end{itemize}
\end{lemma}
\begin{proof}
Fix a round $i$ and a pair $(u,v)$. If $m_{i}(u,v)$ is not empty, then $u$ sends $\widehat{m}_i(u,v) \gets h^*(m_i(u,v) \circ 0^{B'-B})+K_i(u,v)$. The vertex $v$ gets $\widehat{m}_i(u,v)$ and since it knows $K_i(u,v)$ and $h^*$ (as it knows the seed $R$), w.h.p. it can obtain the message $m'=m_i(u,v) \circ 0^{B'-B}$ (by checking $h^*(m')+K_i(u,v)=\widehat{m}_i(u,v)$) as $m'$ ends with $B'-B$ 0's, $v$ sets $\widetilde{m}_i(u,v)=m_{i}(u,v)$, as desired. 

If $m_{i}(u,v)$ is empty, then $u$ sends a random string $m=\widehat{m}_i(u,v)$ of $B'=B + 3\log(nr)$ bits. To decrypt the message, the receiver $v$ first computes $m'=m + K_i(u,v)$ which is still a random uniform string in $\{0,1\}^{B'}$. 
Since all non-empty messages have $B'-B=3\log(nr)$ bits of $0$s, there are at most $2^B$ distinct such messages, and $h^*$ maps these messages into at most $2^B$ values in $\{0,1\}^{B'}$. Since $m'$ is sampled uniformly at random in $\{0,1\}^{B'}$, the probability that it is in the set of these at most $2^B$ values is $2^{B-B'}= (nr)^{-3}$. Hence, w.h.p., $(h^*)^{-1}(m')$ will not end with $B'-B$ bits of $0$, and $v$ drops the message, as desired. 
%
%
\end{proof}

\begin{lemma}[Security of Simulation]\label{lem:secure-compiler}
Alg. $\cA'$ is $f$-mobile-secure. 
\end{lemma}
\begin{proof}
By running the first phase and using Lemma \ref{lem:sufficient-secrets}, the adversary knows nothing on the keys $\mathcal{K}(u,v)$ of all but $4f$ edges, denoted hereafter by $F^*$. Hence, the adversary learns nothing on the messages exchanged in Step (3) over the edges in $E \setminus F^*$. In fact, those messages are distributed uniformly at random in the view of the adversary. Since the empty-messages of $\cA$ are replaced by pure random messages, it remains to show that the adversary learns nothing on the at most $|F^*|\cdot \congestion$ \emph{non-empty} messages going through $F^*$ in Alg. $\cA$. 

We show that the security on these messages holds even if the adversary knows all keys $\mathcal{K}(u,v)$ for every $(u,v)\in F^*$. By the security of the broadcast algorithm of \Cref{thm:bc-sqrt-mobile}, the adversary knows nothing on the random seed $R$. Since the algorithm views a collection of at most $f'=f\cdot \congestion$ values: $h^*(m_1), \ldots, h^*(m_{f'})$ where all messages $m_1,\ldots, m_{f'}$ are distinct, by the properties of the $f'$-wise independent hash family $\mathcal{H}$, we get that these values are distributed independently and uniformly at random. Consequently, all messages observed by the adversary in the last phase, are pure random messages. The security then holds.   
%
%
%
	%
\end{proof}

The proof of Theorem \ref{thm:mobile-compilers} follows. The number of rounds of the simulated algorithm is 
$$
\widetilde{O}(D + f \cdot  \sqrt{\congestion \cdot n} + f \congestion + r \cdot \log n) = 
\widetilde{O}(r + D + f \cdot  \sqrt{\congestion \cdot n} + f  \cdot \congestion),
$$
as desired. The correctness and security follow directly from \Cref{lem:correct-compiler,lem:secure-compiler}, respectively.

%% file: mobile-secure-broadcast.tex
\subsection{$f$-Mobile-Secure Broadcast}\label{sec:fmobile-broascast}

We next show that the $f$-static secure broadcast algorithm of \cite{EavesBroad2022} can be slightly modified to provide $f$-mobile security with the same round complexity.  

\begin{theorem} [$f$-Mobile-Secure Broadcast] \label{thm:bc-sqrt-mobile}
For every $(2f+3)(1+o(1))$ edge-connected $n$-vertex $D$-diameter graph $G$, there exists a randomized $f$-mobile-secure broadcast algorithm $\MobileBroadcast$ for sending w.h.p a $b$-bit message $m^*$ that runs in $\widetilde{O}(D+\sqrt{f \cdot b\cdot n}+b)$ rounds. The edge congestion of the algorithm is $\widetilde{O}(\sqrt{f \cdot b\cdot n}+b)$. 
\end{theorem}
This matches the bounds provided by \cite{EavesBroad2022} for the $f$-static-secure setting. We first provide a quick high-level description of the static-secure algorithm of \cite{EavesBroad2022} and then highlight the required modifications. For the sake of clarity, we provide a simplified description assuming that all nodes know the network topology. 
%
%
%

As in the $f$-static setting of \cite{EavesBroad2022}, our $f$-mobile-secure algorithm is also based on the distributed computation of a tree collection denoted as \emph{fractional tree packing} \cite{Censor-HillelGK14}. 
\begin{definition}[Fractional Tree Packing]\label{def:tree-packing}
	A \emph{fractional tree-packing} of a graph $G$ is a collection of spanning trees $\mathcal{T}$, and a weight function $w:\mathcal{T} \rightarrow (0,1]$, such that for every edge $e\in E$, $\sum_{T_i \in \mathcal{T}: e\in T_i} w(T_i) \leq 1$. The size of the fractional tree packing $\mathcal{T}$ is denoted by $\chi(\mathcal{T})=\sum_{ T_i \in \mathcal{T}} w(T_i)$.
\end{definition}
We use the distributed computation of fractional tree packing, due to Censor-Hillel, Ghaffari, and Kuhn \cite{Censor-HillelGK14}, slightly adapted to our setting by \cite{EavesBroad2022}.
\begin{lemma}[A slight adaptation of Theorem 1.3 \cite{Censor-HillelGK14}] \label{lem:large-weights}
	Given a $D$-diameter $\lambda$ edge-connected $n$-node graph, and an integer parameter $\lambda' \leq \frac{\lambda-1}{2}(1-o(1))$, one can compute a fractional tree packing $\mathcal{T}$ and a weight function $w:\mathcal{T} \rightarrow (0,1]$, such that:
	\begin{enumerate}
		\item for every $T_i \in \mathcal{T}$, $w(T_i)=\frac{n_i}{\lceil \log^8 n \rceil}$ for some positive integer $n_i \geq 1$,
		\item $\chi(\mathcal{T}) \in [\lambda',\lambda'(1+o(1))]$,
		\item the round complexity of the algorithm is $\widetilde{O}(D+\sqrt{\lambda'\cdot n})$ and the edge-congestion is $\widetilde{O}(\sqrt{\lambda'\cdot n})$.
	\end{enumerate}
\end{lemma}

We are now ready to provide the complete description of Alg. $\MobileBroadcast$ given a source $s$ holding a broadcast message $m^*$ of $b$-bits (where possibly $b=\Omega(\log n)$). 

\paragraph{Step (0): Fractional Tree Decomposition and Landmarks Selection.} Apply the fractional tree-packing algorithm of Lemma \ref{lem:large-weights} with $\lambda'=f+1$. Denote the output tree packing by $\mathcal{T}$ and its size by $\chi(\mathcal{T})$. By the end of this computation, all nodes also know the weights $n_1,\ldots, n_{|\mathcal{T}|}$ of the trees in $\mathcal{T}$ (see Lemma \ref{lem:large-weights}(1)). Decompose each tree $T_i \in \mathcal{T}$ into fragments $T_{i,j}$ of size $s=\Theta(\min\{\sqrt{f\cdot b \cdot n},n\})$ (as in \cite{EavesMST2023}). Let $L \subseteq V$ be a subset of nodes, denoted as \emph{landmarks}, obtained by sampling each node with probability $p=\Theta(\log n/s)$. The identities of the sampled nodes $L$ are broadcast to all the nodes in $O(D+|L|)$ rounds. For every fragment $T_{i,j}$, let $\ell_{i,j}$ be some chosen node in $L \cap V(T_{i,j})$, denoted as the \emph{leader} of the fragment, which exists with high probability. This leader can be chosen in $O(s)$ rounds by a simple convergecast over each fragment, simultaneously. This entire step can be done in a non-secure manner, as it does not leak any information on the broadcast message, $m^*$.

\smallskip
\noindent \textbf{Step (1): Secure Secret Transmission to Landmarks.} The source node $s$ secret shares its $b$-bit broadcast message $m^*$ into 
\begin{equation}\label{eq:num-shared-broadcast}
\widehat{f}=\chi(\mathcal{T})\cdot \lceil \log^8 n \rceil \mbox{~~shares}, 
\end{equation}
denoted as $M^*=(m_1,\ldots, m_{\widehat{f}})$, where each share $m_i$ has $b$ bits. Note that by Lemma \ref{lem:large-weights}(1), $\widehat{f}$ is an integer.   Then algorithm then applies Alg. $\MobileSecureMulticast$ of Lemma \ref{lemma:unicast-mobile} to securely send the landmarks $L$ the collection of all $\widehat{f}$ shares. 

\smallskip
\noindent \textbf{Step (2): Local Secret Exchange.} The step consists of $r=3f \cdot b \cdot \lceil \log^8 n \rceil$ rounds, in each round, random $O(\log n)$-bit messages $R_j(u,v)$ are exchanged over all edges $(u,v) \in G$. At the end of this step, each neighboring pair $u,v$ locally apply the algorithm of Theorem \ref{thm:extract} to compute a collection of $r'=b \cdot \lceil \log^8 n \rceil$ secret keys $\mathcal{K}(u,v)=\{K_j(u,v)\}_{j=1}^{r'}$. In the analysis, we show that the adversary knows nothing on the $\mathcal{K}(u,v)$ sets of all but $f$ many (undirected) edges.  The third step, described next, is simulated exactly as in the static-secure algorithm of \cite{EavesBroad2022}, with the only distinction that the messages are OTP encrypted with the keys of $\mathcal{K}(u,v)$.

\smallskip
\noindent \textbf{Step (3): Shares Exchange in each Fragment.} Recall that by Lemma \ref{lem:large-weights}(1), $w(T_i)=\frac{n_i}{\lceil \log^8 n \rceil}$ for some positive integer $n_i \geq 1$ for every tree $T_i \in \mathcal{T}$. Our goal is to propagate a distinct collection of $n_i$ shares in $M^*$ over each tree $T_i$. 
To do that, the landmark nodes partition locally and canonically the shares of $M^*$ into disjoint subsets $M^*_1, \ldots , M^*_k$ such that $|M^*_i|=n_i$ for every $i \in \{1,\ldots, k=|\mathcal{T}|\}$. Note that by \Cref{eq:num-shared-broadcast}, $\sum_{T_i \in \mathcal{T}} n_i = \widehat{f}$.
The leader $\ell_{i,j}$ of each fragment $T_{i,j}$ sends all the shares in $M^*_i$ over its fragment, as follows. Every edge $(u,v)$ encrypts its $j^{th}$ message using the key $K_j(u,v) \in \mathcal{K}(u,v)$, for every $j \in \{1,\ldots, r'\}$. We show in the analysis, that each edge is required to send at most $r'$ messages and therefore it indeed has sufficiently many keys in $\mathcal{K}(u,v)$. Finally, each node $v$ recovers the $b$-bit broadcast message, by setting $m^*=\oplus_{i=1}^k \oplus_{m \in M^*_i} m~.$
This completes the description of the algorithm. 

\smallskip
\noindent \textbf{Analysis.} The correctness is immediate by the $f$-static-secure description of \cite{EavesBroad2022}, noting that each node indeed receives all needed shares, w.h.p., to recover the broadcast message $m^*$. 
We now turn to address the round complexity and congestion analysis. In Step (0) the computation of the fractional tree packing can be done in $\widetilde{O}(D+\sqrt{f\cdot n})$ rounds by Lemma \ref{lem:large-weights}(3) with $\lambda'=f+1$. By \cite{EavesBroad2022}, the fragmentation of the trees takes $s=\widetilde{O}(\min\{\sqrt{f \cdot b \cdot n},n\})$ rounds. Picking the leader of each fragment takes $O(s)$ rounds, as well. 

Step (1) takes $\widetilde{O}(D+f \cdot b\cdot |L|)$ rounds, by \Cref{lemma:unicast-mobile}. Since, w.h.p., $|L|=\widetilde{O}(\sqrt{n/(f \cdot b)})$, it takes $\widetilde{O}(D+\sqrt{f \cdot b \cdot n})$ rounds.
As for Step (3), since $w(T_i)\leq 1$, it holds that $n_i \leq \lceil \log^8 n\rceil$ for every $T_i \in \mathcal{T}$. Therefore, Step (3) propagates $r'$-bit messages, for $r'=b \cdot \lceil \log^8 n\rceil$, over the edge-disjoint fragments of size $O(s)$. This can be done in $\widetilde{\Theta}(b+\sqrt{f \cdot b \cdot n})$ rounds via standard pipeline. Note that edge congestion of this step is bounded by $r'$, and that the edge congestion in the entire algorithm is bounded by $\widetilde{O}(b+\sqrt{f \cdot b \cdot n})$.

\noindent \textbf{Security.} It remains to show that the algorithm is $f$-mobile-secure. 
Recall that the eavesdropper is assumed to know $G$, therefore Step (0) can be implemented in a non-secure manner. Step (1) is $f$-mobile-secure by Lemma \ref{lemma:unicast-mobile}.

We turn to consider Step (3) in which the shares in $M^*$ are sent over the tree fragments of $\mathcal{T}$. By the above, Step (3) sends at most $r'$ messages along each edge and therefore the number of keys in $\mathcal{K}(u,v)$ is sufficient. By Lemma \ref{lem:sufficient-secrets}, we get that the adversary knows nothing on the keys in $\mathcal{K}(u,v)$ of all but $F^* \subseteq G$ edges, where $|F^*|\leq f$. 
Consequently, the adversary learns nothing from the messages exchanged over the edges in $E \setminus F^*$ during this step. We next show that there exists at least one share in $M^*$ that the eavesdropper did not learn, and therefore it knows nothing on $m^*$. Since the tree fragments of each $T_i \in \mathcal{T}$ are edge-disjoint, the number of shares in $M^*$ sent over an edge $e$ is bounded by:
\begin{equation}\label{eq:edge-shares-packing}
	\sum_{e\in T_i}|M^*_i|=\sum_{e\in T_i}n_i \leq \lceil \log^8 n \rceil ~, \nonumber
\end{equation}
where the last inequality follows by the fractional tree packing guarantee that $\sum_{e\in T_i} w(T_i) \leq 1$, and using the fact that $w(T_i)=\frac{n_i}{\lceil \log^8 n \rceil}$ for every $T_i \in \mathcal{T}$.  Therefore, the number of shares exchanged over the edges in $F^*$ and possibly observed by the eavesdropper can be bounded by:
$$|F^*| \cdot \lceil \log^8 n \rceil < \chi(\mathcal{T})\cdot \lceil \log^8 n \rceil = \widehat{f}=|M^*|~,$$
where the first inequality is by Lemma \ref{lem:large-weights}(2) (with $\lambda'=f+1$), and the last equality follows by \Cref{eq:num-shared-broadcast}.

%% file: byzantine-improved-for-budgeted.tex
\section{Proof of \Cref{lem:correction_budgeted}}  
\label{sec:budgeted_k_edges_message_correction}

In this section, we prove \Cref{lem:correction_budgeted} by adapting Subsection~\ref{sec:exactly_k_edges_improved} to be resilient to an average number of faults per round rather than $f$ faults per round. The adaptation is almost a word-for-word re-analysis of Subsection~\ref{sec:exactly_k_edges_improved}. 

Let $t \geq \kDiam \cdot \log^{\bar{c}}{n}$ for some large enough constant $\bar{c}$ be an integer known to all nodes, and let $\tau = \frac{t}{2\log{n}}$.  Let $d =O(\frac{k}{\eta\log{n}})$be an integer known to all nodes, such that we are guaranteed $|\{(u,v)\in E(G) \stt m(u,v) \neq \widetilde{m}(u,v)\}| \leq d$.

\paragraph{Algorithm $\ModifiedECCSafeBroadcast$:} Let $\ell$ be an integer such that $k \geq c''\ell$ for a sufficiently large $c'' > 0$, and let $\tau$ be an input integer. The broadcast message held by the root $v_r$ consists of a list $[\alpha_1,\dots,\alpha_\ell] \in [q]$ where $q=2^p$ for some positive integer $p$ and $q \geq k$. Let $C$ be a $[\ell,k,\delta_C]_q$-code for $\delta_C=(k-\ell+1)/k$, known as the Reed Solomon Code.
The root encodes the broadcast message $[\alpha_1,\dots,\alpha_\ell]$ into a codeword $C([\alpha_1,\dots,\alpha_\ell]) =[\alpha'_1,\dots,\alpha'_k]$. Next, the algorithm runs $k$ RS-compiled $\kDiam$-hop broadcast algorithms, in parallel. That is, for every $T_j \in \mathcal{T}$, let $\Pi(T_j)$ be a $\kDiam$-hop broadcast algorithm in which the message $\alpha'_j$ starting from the root $v_r$, propagates over $T_j$ for $\kDiam$ hops (hence taking $O(\kDiam+\log q)$ rounds). Then the nodes wait until a total of $\tau$ rounds have elapsed since the start of the procedure. Let $\Pi_{RS}(T_j)$ be the RS-compilation of that algorithm, denoted hereafter as RS-broadcast algorithm. All $k$ RS-broadcast algorithms are implemented in parallel by using the scheduling scheme of Lemma \ref{lem:scheduler_security}.

Let $\alpha'_j(u)$ be the value that a node $u$ receives from $T_j$ (or $\alpha'_j(u) = 0$ if it received no value) at the end of this execution. To determine the broadcast message $[\alpha_1,\dots,\alpha_\ell]$, each $u$  calculates first the closest codeword $\alpha(u)$ to $\alpha'(u) = [\alpha'_{1}(u),\dots,\alpha'_{k}(u)]$. 
Its output is then given by $\widetilde{\alpha}(u)=[\widetilde{\alpha}_1(u),\dots,\widetilde{\alpha}_\ell(u)]=C^{-1}(\alpha(u))$. 

\begin{lemma}\label{lem:correct-BroadcastECC_budg}
Consider the execution of Alg. $\ModifiedECCSafeBroadcast$ in the presence of at most $td$ corrupted messages, with a given broadcast message $[\alpha_1,\dots,\alpha_\ell] \in [q]^\ell$, a distributed knowledge of a \emph{weak} $(k,\kDiam,\eta)$ tree packing for $k\geq \max\{c''\cdot \ell, c^*d\eta\log{n}\}$ for large constants $c'', c^*$, and assuming that $\tau = \Omega((\kDiam+\log{q})\eta)$. Then $\widetilde{\alpha}(u)=[\alpha_1,\dots,\alpha_\ell]$ for every node $u$. In addition, the round complexity is $\Theta(\tau)$  $1$-$\congest$ rounds. 
\end{lemma}
\begin{proof}
Let $\mathcal{T}'\subseteq \mathcal{T}$ be the collection of $\kDiam$-spanning trees rooted at a common root $v_r$. By the definition of weak tree-packing, we have that $|\mathcal{T}'|\geq 0.9k$. 

By \Cref{lem:scheduler_security}, all $k$ algorithms $\{\Pi_{RS}(T_j)\}_{j=1}^k$ can be performed in $\Theta(\tau)$ rounds, such that all but $c \cdot d\cdot \eta (t/\tau)$ end correctly, for some constant $c$.  Therefore, at least $|\mathcal{T}'|-c \cdot d\cdot \eta\log{n} \geq (1-1/c')k$ algorithms are valid, by taking $k \geq c^*\eta d \log{n}$ for a sufficiently large $c^*$. This implies that for at least $(1-1/c')k$ algorithms we have $\alpha'_{j}(u) = \alpha'_{j}$. (For that to happen, we have $T_j \in \mathcal{T}'$ and the RS-compiled algorithm $\Pi_{RS}(T_j)$ is valid. Or in other words,
	\[\frac{\Hamm(\widetilde{\alpha}(u),C(\alpha_1,\dots,\alpha_\ell))}{k} \leq \frac{1}{c'}.\]

\noindent On the other hand, since $k \geq c''\ell$, the relative distance of the code $C$ can be bounded by:
$$\delta_C =\frac{k-\ell+1}{k} \geq 1-\frac{\ell}{k} \geq 1-1/c''.$$  
Therefore, for any given point $x \in \mathbb{F}_{q}^k$, there is at most one codeword of relative distance less than $\delta_C/2$. As for sufficiently large $c',c''$, one obtains $(1-\frac{1}{c''})/2 \geq \frac{1}{c'}$, we get that: \[\frac{\Hamm(\widetilde{\alpha}(u),C(\alpha_1,\dots,\alpha_\ell))}{k} < \frac{\delta_C}{2}.\] 
We conclude that the decoding of every node is correct, i.e., that $\widetilde{\alpha}(u)=[\alpha_1,\dots,\alpha_\ell]$ for every $u$.
\end{proof}

\paragraph*{Message Correction Procedure}
Our algorithm works in iterations $j \in \{1,\dots,z\}$ for $z=\Theta(\log{n})$. At the beginning of every iteration $j$, each node $v$ holds for each neighbor $u$, a variable $m'_{j-1}(u,v)$, which represents its estimate for its received message from $u$ in round $i$ of Alg. $\cA$. Initially, $m'_{0}(u,v) = m'(u,v)$.  

Next, every node $v$ locally defines two multi-sets corresponding to its outgoing messages and and its $j$-estimated incoming messages: 
$$\mathrm{Out}(v) = \{m(v,u_1),\dots,,m(v,u_{\deg(v)})\} \mbox{~and~} \mathrm{In}_{j}(v) = \{m'_{j}(u_1,v),\dots,m'_{j}(u_{\deg(v)},v)\}.$$ 
Let $S_{j}(v)$ be a stream formed by inserting each element in $\mathrm{Out}(v)$ with frequency $1$, and each element in $\mathrm{In}_{j}(v)$ with frequency $-1$. Let $S_{j} = S_{j}(v_1) \cup \dots, \cup S_{j}(v_n)$ be the stream formed by concatenating all $n$ individual streams. 

Each subgraph $T \in \mathcal{T}$ runs an RS-compiled sub-procedure, $L0_{RS}(T,S_{j},\tau)$, which is defined by applying the RS-compiler of \Cref{thm:less_than_one_error} for the following (fault-free) $L0(T,S_{j},\tau)$ procedure, which is well defined when $T$ is a spanning tree. In the case where $T$ is an arbitrary subgraph, the execution of $L0(T,S_{j})$ which is restricted to $\widetilde{O}(\kDiam)$ rounds, will result in an arbitrary outcome.

\noindent \textbf{Procedure $L0(T,S_{j},\tau)$.} The node $v_r$ first broadcasts $\widetilde{O}(1)$ random bits $R_{j}(T)$ over the edges\footnote{In the case where $T$ is \emph{not} a spanning tree, $v_r$ might have no neighbors in $T$. Nevertheless, the correctness will be based on the $0.9k$ spanning trees in $\mathcal{T}$.} of $T$. Then, each node $v$ initializes $s=O(\log{n})$ mutually independent $\ell_0$-sampler sketches on $S_{j}(v)$ with randomness $R_{j}(T)$ using Theorem~\ref{thm:l_0_sampler}. Let $[\sigma_1(v),\ldots, \sigma_t(v)]$ be the $\ell_0$-sampling sketches obtained for $S_{j}$ with the randomness $R_{j}(T)$. The sum of sketches $\sigma_1(S_{j}),\dots,\sigma_{t}(S_{j})$ is then computed in a bottom-up manner on $T$ from the leaves to the root $v_r$.
Using these $s$ sketches, the root $v_r$ locally samples a list of values $A_{j}(T) = [a_1(T),\dots,a_{s}(T)]$. The nodes terminate after exactly $\tau$ rounds from the start of the protocol.
 
This concludes the description of $L0(T,S_{j},\tau)$ and hence also its RS-compilation $L0_{RS}(T,S_{j},\tau)$. Our algorithm implements the collection of the $k$ RS-compiled algorithms $\{L0_{RS}(T,S_{j},\tau)\}_{T \in \mathcal{T}}$, in \emph{parallel}, by employing the RS-scheduler of Lemma \ref{lem:scheduler_security}. 

\smallskip 
\noindent \textbf{Detecting Dominating Mismatches.} 
A \emph{positive} element $a \in A(T)$ is denoted as an \emph{observed-mismatch}. For every observed-mismatch $a \in \bigcup_{T \in \mathcal{T}} A_{j}(T)$, denote its \emph{support} in iteration $j$, by $\supp_{j}(a) = |\{(\ell,T) \stt a = a_\ell(T), T \in \mathcal{T}, \ell \in \{1,\ldots, t\}\}|$. The root $v_r$ then selects a sub-list $\DM_{j}$ of \emph{dominating} observed mismatches, i.e., mismatches that have a sufficiently large \emph{support} in $\bigcup_{T} A(T)$, based on the given threshold $\Delta_j$. Specifically, for a sufficiently large constant $c''$
\begin{equation}\label{eq:Mij_budg}
\Delta_j=0.4c''2^j\eta s \mbox{~and~} \DM_{j} = \{a \in \bigcup_T A_{j}(T) \stt a>0, \supp_{j}(a) \geq \Delta_j \}~.
\end{equation}
The remainder of the procedure is devoted to (safely) broadcasting the list $\DM_{j}$. 

\smallskip 
\noindent \textbf{Downcast of Dominating Mismatches.} 
To broadcast $\DM_{j}$, Alg. $\ModifiedECCSafeBroadcast$ is applied with parameters $q=2^p$ for $p=\lceil \max(\log{k},\log^5{n})\rceil$ and $\ell=|\DM_{j}|$. Upon receiving $\DM_{j}$, each node $v$ computes its $j$-estimate $m'_{j}(u,v)$ as follows. If there exists a message $m \in \DM_{j}$ with $\ID(m)=\ID(u) \circ \ID(v)$, then $v$ updates its estimate for the received message by setting $m'_{j}(u,v)=m$. Otherwise, the estimate is unchanged and $m'_{j}(u,v)\gets m'_{j-1}(u,v)$. This completes the description of the $j^{th}$ iteration. Within $z=O(\log n)$ iterations, each node $v$ sets $\widetilde{m}(u,v)=m'_{z}(u,v)$ for every neighbor $u$. In the analysis we show that, w.h.p., $\widetilde{m}(u,v)=m(u,v)$ for all $(u,v)\in E$. The pseudocode of Algorithm $\mathsf{MessageCorrectionProtocol}$ is described next.

\begin{mdframed}[hidealllines=false,backgroundcolor=gray!00]
\center \textbf{Algorithm $\mathsf{MessageCorrectionProtocol}$:}

\raggedright\textbf{Input:} Weak $(k,\kDiam,\eta)$ tree packing $\mathcal{T}=\{T_1\dots,T_k\}$.\\
\raggedright \textbf{Output:} Each node $v$ outputs $\{\widetilde{m}(u,v)\}_{u \in N(v)}$, estimation for $m(u,v)$.
	\begin{enumerate}
\item For $j = 1,\dots,z=O(\log{n})$ do:
\begin{itemize}
\item Employ protocol $L0_{RS}(T,S_{j},\tau)$ over each $T \in \mathcal{T}$, in parallel, using \Cref{lem:scheduler_security}.

\item Set $\DM_{j}$ as in Eq. (\ref{eq:Mij_budg}).
\item Broadcast $\DM_{j}$ by applying $\ModifiedECCSafeBroadcast(\tau)$.
\item For every $v \in V$ and $u \in N(v)$: 
\begin{itemize}
\item If $\exists m \in \DM_{j}$ with $\ID(m)=\ID(u)\circ \ID(v)$: $m'_{j}(u,v) \gets m$.
\item Otherwise, $m'_{j}(u,v) \gets m'_{j-1}(u,v)$. 
\end{itemize}
\end{itemize}

\item For every $v \in V$ and $u \in N(v)$: Set $\widetilde{m}(u,v)=m'_{z}(u,v)$.

\end{enumerate}
\end{mdframed}

\noindent \textbf{Analysis.} For simplicity of presentation, we assume that all $\ell_0$-sampler procedures succeed (by paying a $1/\poly(n)$ additive term in the failure probability of the protocol). Let $\mathcal{T}'\subseteq \mathcal{T}$ be the collection of spanning-trees of depth at most $\kDiam$, rooted at the common root $v_r$. 

We assume that $k \geq c' \cdot c'' \eta d \log{n}$, where $c''$ is the constant in the number of failing algorithms in \Cref{lem:scheduler_security}, and $c' > 1$ is a sufficiently large constant compared to $c''$. See Eq. (\ref{eq:Mij_budg}) for the definition of $\Delta_j$.

A tree $T_q$ is denoted as $j$-\emph{good} if (a) $T_q \in \mathcal{T}'$ and (b) $L0_{RS}(T_q,S_{j},\tau)$ ended correctly when using the scheduler of \Cref{lem:scheduler_security}. Let $\mathcal{T}^{(j)}_{\good}\subseteq \mathcal{T}'$ be the collection of good trees. Since the round complexity of each scheduler is $\Theta(\tau)$ and the total number of faults is $td = 2\tau d\log{n}$, at most $cd\log{n}$ scheduled algorithms fail for some constant $c$.
By \Cref{lem:scheduler_security}, it then holds that:

\begin{equation}\label{eq:good-trees-num_budg}
|\mathcal{T}^{j}_{\good}| \geq \left(0.9-\frac{c''d\eta\log{n}}{k}\right)|\mathcal{T}| \geq \left(0.9-\frac{1}{c'} \right)|\mathcal{T}|~.
\end{equation}
Finally, we say that a sent message $m(u,v)$ a $j$-\emph{mismatch} if $m(u,v) \neq m'_{j}(u,v)$. Let $B_{j}$ denote the number of $j$-mismatches.

\begin{lemma}
\label{lem:faulty_message_chernoff_budg}
For every $j \in \{0,\ldots, z\}$, if $B_{j-1} \leq d/2^{j-1}$, then for any $(j-1)$-mismatch $m$, $\supp_{j}(m) \geq \Delta_j$, w.h.p.
\end{lemma}	
\begin{proof}
Let $m$ be a $(j-1)$-mismatch. Since there are at most $d/2^{j-1}$ mismatches, there are at most $2d/2^{j-1}$ non-zero entries in $S_{j}$. By \Cref{thm:l_0_sampler}, in each sketch of a $j$-good tree $T$ the root detects a given $(j-1)$-mismatch with probability at least $2^{j-1}/d - \epsilon(n)$ for $\epsilon(n) = O(1/\poly(n))$. Therefore, for a sufficiently large $c'$, by Equations (\ref{eq:Mij_budg},\ref{eq:good-trees-num_budg}) and the fact that we use $s=O(\log n)$ independent $\ell_0$ sketches, we get:

\[\mathbb{E}\left(\supp_{j}(m) \right) \geq \left(0.9-\frac{1}{c'}\right)\left(\frac{2^j}{2d}-\epsilon(n)\right)k\cdot s \geq 0.8\left(\frac{2^{j-1}}{2d}\right)k \geq 0.4c'' \cdot c'2^j\eta\cdot s = c' \Delta_j~.\]
	
Moreover, the observed mismatches sampled by the good trees are mutually independent of each other. Given $(j-1)$-mismatch $m'$, the probability that it is sampled by less than $\Delta_j$ many $j$-good trees can be bounded by a Chernoff bound (Lemma~\ref{lem:chernoff}), as follows:
	
	\[\Pr(\supp_{j}(m) \leq \Delta_j) \leq \Pr\left(\supp_{j}(m) \leq \mathbb{E}\left(\supp_{j}(m)\right)/c'\right)  \leq e^{-(1-1/c')^2 c'\Delta_j/2} \leq \frac{1}{\poly(n)},\]
	
	where the last inequality holds for a sufficiently large $c'$ (which can be chosen sufficiently large compared to $c''$).
\end{proof}

\begin{lemma}
\label{lem:faults_decrease_budg}
For every $j \in \{0,\ldots, z\}$, w.h.p. $B_{j} \leq d/2^{j}$.
\end{lemma}
\begin{proof}
We prove the claim by induction on $j$. For $j = 0$, the number of $j$-mismatches at the start of the protocol is assumed to be at most $d$. Assume that the claim holds up to $j-1$ and consider $j\geq 1$. 

We say that an observed-mismatch $a$ has $j$-\emph{high support} if $\supp_{j}(a) \geq \Delta_j$, and has $j$-\emph{low support} otherwise. An observed-mismatch $a$ is \emph{competing} with $m(u,v)$ if $\ID(a) = \ID(u) \circ \ID(v)$ but $a \neq m(u,v)$. Note that assuming all sketches on $j$-good trees are successful (which indeed holds w.h.p.), then all \emph{competing mismatches} with any message $m(u,v)$ must be sampled by $j$-\emph{bad} trees, since they are not real $(j-1)$ mismatches. 

A necessary condition for $m(u,v)$ to be a $j$-mismatch, is either that (a) there is an observed-mismatch with $j$-high support that is competing with $m(u,v)$, or (b) it is a $(j-1)$-mismatch and has $j$-low support. 
By \Cref{lem:faulty_message_chernoff}, all $(j-1)$-mismatches have high support w.h.p., and therefore, there are no $j$-mismatches due to condition (b).

Since the number of $j$-bad trees is at most $c''d\eta\log{n}+0.1k$ (see Eq. (\ref{eq:good-trees-num_budg})), at most $\frac{(c''\eta d\log{n}+0.1k)s}{\Delta_j} \leq d/2^{j+1}$ competing observed-mismatches have $j$-high support. Therefore, $B_{j} \leq d/2^{j}$ w.h.p.
\end{proof}

\begin{proof}[Proof of \Cref{lem:correction_budgeted}]
In the last iteration $z = {O(\log{n})}$, it holds that $B_z = 0$ (since $d \leq n^2$ w.l.o.g.). In particular, at the last iteration $z$, each estimated message is the correct message, i.e. $\widetilde{m}_{z}(u,v) = m(u,v)$. 
\end{proof}

%% file: treecomputation.tex
\section{Distributed Computation of a Low Depth Tree Packing}
\label{sec:app_tree_packing}

In this section, we show to compute a low depth tree packing in $\congest$. The procedure and result is essentially the same as in \cite{Ghaffari15}, except we optimize for different parameters, and use a more direct approach to analyze the resulting packing using some modification to a well known greedy multiplicative weights analysis. We give the details of this procedure for completeness.

Assume a graph contains a packing $\{T^*_1,\dots,T^*_k\}$ of $d$-depth spanning trees with load at most $\eta$ . We describe a packing process in $\congest$ of $k$ many $O(d\log{n})$-depth spanning trees, which in the analysis we show has load $O(\eta\log^2{n})$. We use the following algorithm of \cite{Ghaffari15} as a sub-procedure of our protocol:
\begin{lemma}[\cite{Ghaffari15} Theorem 2]
	\label{lem:d_diam_find}
	Given parameter $d \geq D$, there is a $\congest$ algorithm that in $O(d\log{n})$ rounds computes a $O(d\log{n})$-depth spanning tree with cost within $O(\log{n})$ factor of the min-cost $d$-depth spanning tree.
\end{lemma}

In the remainder, denote $\alpha = O(\log(n))$ the approximation guarantee of \Cref{lem:d_diam_find}. Our procedure is almost exactly the same as in \cite{Ghaffari15}. Consider the following process in which $k$ many $O(d\log{n})$-depth spanning trees $T_1,\dots,T_k$ are added one after the other in a sequence of $k$ iterations into the packing. For iteration $i=1,\dots,k$, let $h^i_e = \{j < i \stt e \in T_i\}$ be the load of edge $e$ at the start of iteration $i$, and let $w_i(e) = a^{h^i_e+1}-a^{h^i_e}$, for $a = \log_2\frac{\alpha+2}{\alpha+1}$. In iteration $i$, we find using the procedure of \Cref{lem:d_diam_find} a $\widetilde{O}(\kDiam)$-depth spanning tree, whose weight is within an $O(\log{n})$ factor of the min-cost $\kDiam$-depth spanning tree according to weight function $w_i$. We add this tree to the packing, and update the weight of the edges according to their new load. This concludes the description of the algorithm.  Let $l_e$ be the load of edge $e$ after the final iteration.

\begin{theorem}
	\label{thm:packing_exponential}
	Let $T_1,\dots,T_k$ be a collection of trees obtained by a packing process described above, where $w_i(T_i) \leq \alpha w_i(T^*_i)$ for any $1 \leq i \leq k$. Then the maximum load is at most $\max_e l_e = O(\eta\alpha(\log{n}+\log{\alpha}))$.
\end{theorem}
\begin{proof}
	For any step $i$, the following holds due to our choice of $T_i$:
	\begin{flalign*}
		\sum_{e \in T_i}(a^{h^{i+1}_e/\eta}-a^{h^i_e/\eta}) &= \sum_{e \in T_i}(a^{(h^i_e+1)/\eta}-a^{h^i_e/\eta}) \leq \alpha \sum_{e \in T^*_i}(a^{(h^i_e+1)/\eta}-a^{h^i_e/\eta}) \\&= \alpha \sum_{e \in T^*_i}a^{h^i_e/\eta}(a^{1/\eta}-1) \leq \alpha \sum_{e \in T^*_i}a^{l_e/\eta}(a-1)/\eta~.
	\end{flalign*}

	We sum on $i=1,\dots,k$:
	
	\[\sum_{i=1}^k\sum_{e \in T_i}(a^{(h^i_e+1)/\eta}-a^{h^i_e/\eta}) \leq \alpha \sum_{i=1}^k \sum_{e \in T^*_i}a^{l_e}(a-1)/\eta~,\]
	
where the last inequality holds since $a^x-1 \leq x(a-1)$ for any $0\leq x\leq 1$. By switching the order of summation, we get
	
	\begin{align*}
	\sum_{e \in E}\sum_{i:e \in T_i}(a^{(h^{i+1}_e)/\eta}-a^{h^i_e/\eta}) &= \sum_{e \in E}\sum_{i:e \in T_i}(a^{(h^i_e+1)/\eta}-a^{h^i_e/\eta})  \\&\leq \alpha (a-1)\sum_{e \in E}(a^{l_e/\eta} \cdot \sum_{i:e \in T^*_i}(1/\eta)) \\&\leq \alpha(a-1)\sum_{e \in E}a^{l_e/\eta}~.
	\end{align*}
	
As the left sum is a telescopic, we have:
	\begin{flalign*}
	\sum_{e \in E}(a^{l_e/\eta}-1) \leq \alpha (a-1)\sum_{e \in E}a^{l_e/\eta} &\implies \left(\sum_{e \in E}a^{l_e/\eta}\right)-|E| \leq \alpha (a-1)\sum_{e \in E}a^{l_e/\eta} \\&\implies
	(1+\alpha-\alpha a)\sum_{e \in E}a^{l_e/\eta} \leq |E| \implies (1+\alpha-\alpha a)\max_{e \in E}a^{l_e/\eta} \leq |E|~.
	\end{flalign*}
	
Finally, by taking $\log_a$ from both sides:
	\begin{flalign*}
		\max_{e \in E}(l_e/\eta) &\leq \log_a|E|-\log_a(1+\alpha-\alpha a) = (\log_2{|E|}-\log_2(1+\alpha-\alpha a))/\log_2{\frac{\alpha+2}{\alpha+1}} \\&= (\log_2{|E|}+\log_2(1+\alpha))/\log_2{\frac{\alpha+2}{\alpha+1}} \leq (\log_2{|E|} + \log_2(1+\alpha))/(1/(2+2\alpha))  \\&= O(\alpha(\log{n}+\log{\alpha})),
	\end{flalign*}
where the last inequality holds since $2\log_2(1+x) \geq x$ for any $x \leq 1$. It follows that $\max_{e \in E} l_e = O(\eta\alpha(\log{n}+\log{\alpha}))$, i.e. the load of the packing is $O(\eta\alpha\log{n}) = O(\eta\log^2{n})$.
	
\end{proof}